\documentclass[a4paper,twocolumn,11pt,accepted=2025-04-06]{quantumarticle}
\pdfoutput=1
\usepackage[utf8]{inputenc}
\usepackage[english]{babel}
\usepackage[T1]{fontenc}
\usepackage{amsmath}

\usepackage{tikz}
\usepackage{lipsum}

\usepackage{setspace}

\usepackage{amsmath}
\usepackage{amssymb}
\usepackage{amsthm}
\usepackage{tikz}
\usetikzlibrary{decorations.markings,decorations.pathreplacing,calc}
\usetikzlibrary{arrows.meta}
\usepackage{enumerate}
\usepackage{comment}
\usepackage{soul}
\usepackage{mathrsfs}

\usepackage[margin=1in]{geometry}
\usepackage{subfig}
\usepackage[utf8]{inputenc}

\usepackage[colorlinks = true, linkcolor=black, citecolor=blue,urlcolor=purple]{hyperref}
\usepackage{array}

\usepackage{mathrsfs} 

\usepackage{indentfirst}

\usepackage{enumitem}
    \setlist[itemize]{noitemsep, topsep=0pt}
    \setlist[enumerate]{noitemsep, topsep=0pt}
\usepackage{verbatim}
\usepackage{mathtools,amssymb,amsmath}
\usepackage{bm}
\usepackage{yhmath}
\usepackage[inline]{asymptote}
\usepackage{cancel}
\usepackage{relsize}
\usepackage{array}
\usepackage{bm}

\usepackage{eczoo} 

\newtheorem{thm}{Theorem}

\newcommand{\bra}[1]{\langle#1|}
\newcommand{\ket}[1]{|#1\rangle}

\newcommand{\be}{\begin{equation}}
\newcommand{\ee}{\end{equation}}
\newcommand{\<}{\langle}
\renewcommand{\>}{\rangle}

\newcommand{\Tr}{{\rm Tr\,}}
\newcommand{\tr}{{\rm tr\,}}
\usepackage[normalem]{ulem}
\usepackage{color}
\def\hL{\mathcal{L}}

\newcommand{\vz}{\vec{0}}
\newcommand{\vk}{\vec{k}}
\newcommand{\hZ}{\mathcal{Z}}

\newcommand{\bsmu}{{\boldsymbol \mu}}
\newcommand{\bsnu}{{\boldsymbol \nu}}

\usepackage{algorithm}
\usepackage{algpseudocode}

\renewcommand{\vec}[1]{{\bf #1}}
\newtheorem{lem}{Lemma}

\newtheorem{dfn}{Definition}

\newcommand{\x}{\raisebox{-0.027cm}{\includegraphics[width=.014\linewidth]{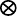}}}
\renewcommand{\o}{\raisebox{-0.06cm}{\includegraphics[width=.012\linewidth]{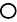}}}
\renewcommand{\*}{\raisebox{-0.06cm}{\includegraphics[width=.012\linewidth]{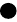}}}
\newcommand{\sss}{\raisebox{0.05cm}{\includegraphics[width=.035\linewidth]{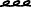}}}
\newcommand{\sq}{\raisebox{-0.027cm}{\includegraphics[width=.014\linewidth]{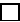}}}
\newcommand{\zzz}{\raisebox{0.05cm}{\includegraphics[width=.021\linewidth]{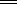}}}
\newcommand{\vvv}{\raisebox{0.1\height}{\includegraphics[ width=.035\linewidth]{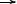}}}
\newcommand{\nnn}{\raisebox{0.5\height}{\includegraphics[ width=.021\linewidth]{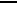}}}

\newcommand{\review}[1]{\textcolor{black}{#1}}
\definecolor{crimson}{RGB}{220, 20, 60}
\newcommand{\reviewtwo}[1]{\textcolor{black}{#1}}

\usepackage{authblk}

\let\svlabel\label
\let\svref\ref
\makeatletter
\newcommand\labelx[2][\relax]{
  \svlabel{#2}
  \ifx\relax#1\else
    \immediate\write\@mainaux{\string
      \expandafter\gdef\noexpand\csname custom#2\noexpand\endcsname{#1}}
  \fi
}
\makeatother
\newcommand\refx[2][\relax]{
  \ifx s#1\csname custom#2\endcsname\else\svref{#2}\fi
}

\begin{document}

\title{Bounds on Autonomous Quantum Error Correction}

\date{}

 \author[1]{Oles Shtanko$^*$}
 \author[2,3]{Yu-Jie Liu$^*$}
 \author[4,5]{Simon Lieu}
\thanks{This work was done prior to joining AWS.}
\thanks[*]{These authors contributed equally to this work.}
 \author[4,5]{Alexey V. Gorshkov}
 \author[5]{Victor V. Albert}
 \affil[1]{\small{IBM Quantum, IBM Research -- Almaden, San Jose, CA, USA}}
 \affil[2]{\small{Technical University of Munich, TUM School of Natural Sciences, Physics Department, 85748 Garching, Germany}}
 \affil[3]{\small{Munich Center for Quantum Science and Technology (MCQST), Munich, Germany}}
 \affil[4]{\small{Joint Quantum Institute, NIST/University of Maryland, College Park, MD, USA}}
 \affil[5]{\small{Joint Center for Quantum Information and Computer Science, NIST/University of Maryland, College Park, MD, USA}}

\maketitle

\vspace{-1cm}

\begin{abstract}
Autonomous quantum memories are a way to passively protect quantum information using engineered dissipation that creates an ``always-on'' decoder. We analyze Markovian autonomous decoders that can be implemented with a wide range of qubit and bosonic error-correcting codes, and derive several upper bounds and a lower bound on the logical error rate in terms of correction and noise rates. These bounds suggest that, in general, there is always a correction rate, possibly size-dependent, above which autonomous memories exhibit arbitrarily long coherence times. For any given autonomous memory, size dependence of this correction rate is difficult to rule out: we point to common scenarios where autonomous decoders that stochastically implement active error correction must operate at rates that grow with code size.
For codes with a threshold, we show that it is possible to achieve faster-than-polynomial decay of the logical error rate with code size by using superlogarithmic scaling of the correction rate. We illustrate our results with several examples. One example is an exactly solvable global dissipative toric code model that can achieve an effective logical error rate that decreases exponentially with the linear lattice size, provided that the recovery rate grows proportionally with the linear lattice size.
\end{abstract}

\vspace{0.8cm}

One of the biggest challenges in quantum computing is the problem of noise. Any realistic qubit architecture is prone to dissipation due to interactions with the environment, leading to errors and subsequent loss of quantum information. Traditional error correction strategies have focused on manual periodic error \review{detection} and
correction  \cite{terhal2015quantum, gottes1997}. 
In recent years, however, there has been a surge in autonomous ``hardware'' methods designed to compensate for noise using engineered dissipation  \cite{Paz1998,Ahn2002,Sarovar2005,Mabuchi2009,Oreshkov2013}. 
Several dissipative quantum memories have been successfully implemented, in particular using various \eczoohref[bosonic codes]{oscillators}\footnote{The name of each code, when introduced for the first time in the manuscript, will contain a hyperlink to the code's corresponding webpage in the Error Correction Zoo  \cite{ErrorCorrectionZoo}.} \cite{science.aaa2085,PhysRevX.8.021005,gertler2021protecting,lescanne_exponential_2020,de_neeve_error_2022,campagne-ibarcq_quantum_2020,sivak_real-time_2023}, but they have not been fully exploited for error correction in real many-body systems consisting of several qubits or qudits. 

A topic closely related to dissipatively stabilized quantum memories is \eczoohref[self-correction]{self_correct}  \cite{brown2016quantum}, a process in which the propagation of errors is naturally limited 
without performing 
active error correction.
An example of self-correction in the classical world is the storage of information in magnetic hard drives.
Here, classical information is encoded in a ferromagnet by its collective spin magnetization, and errors resulting from spontaneous individual spin flips become energetically unfavorable and are therefore eliminated by 
a local thermalization process.  
This mechanism ensures a memory lifetime that scales exponentially with system size  \cite{peierls_1936, griffiths1964}. 

The principle of classical self-correction can be extended to quantum systems, such as those stabilized by frustration-free Hamiltonians  \cite{gottes1997}. 
For example, the \eczoohref[four-dimensional toric code]{higher_dimensional_surface}  \cite{dennis2002} demonstrates a finite-temperature topological order that naturally protects quantum information  \cite{Alicki:2010,pastawski:2011,liu2023dissipative}.
Unfortunately, several no-go results preclude local frustration-free Hamiltonians from achieving self-correction with a constant error rate in two dimensions  \cite{Bravyi2009,cardinal2013local,haah2012logical,hastings2011topological} and for some three-dimensional models  \cite{Yoshida2011,haah2013commuting,pastawski2015fault}. 
Such studies are also hampered by broad challenges associated with quantum complexity  \cite{hallgren2013local,Aharonov2009,Schuch2007,Piddock2015}. 

The use of dissipative processes to aid or induce self-correction in quantum systems continues to be actively studied.
In some cases, it is possible to show numerically that lower-dimensional systems can still offer quantum memory times that grow with system size  \cite{Haah2011,hamma2009toric,wootton2013topological,chesi2010self,kapit2015passive,lieu2024candidate}. It is also possible to establish general rules and construct examples of memory whose performance is provably suppressed by increasing the recovery dissipation strength  \cite{lihm2018implementation, lebreuilly2021autonomous}.
However, there is currently no unified approach that allows general conclusions to be drawn about scalability of a generic memory performance, and there are no universal bounds on the aid that a dissipative process can provide.

Continuing existing line of work  \cite{lihm2018implementation, lebreuilly2021autonomous}, we derive a general non-perturbative bound (see Theorem \ref{generic_bound}), valid for a wide class of autonomous quantum memories 
and restricted by 
only a few intuitive assumptions. This bound expresses the logical error rate of a dissipative memory in terms of the noise-to-recovery ratio --- the ratio between the benevolent recovery dissipation rates and the strength of any malevolent noise.  The core idea is to use the resummation of Dyson's perturbative expansion, which allows one to derive the logical error in the late-time limit. We show that, as soon as the noise-to-recovery ratio is less than a critical value, autonomous memory can achieve lifetimes that grow exponentially with the inverse noise rate. At the same time, this general bound promises an exponentially small logical error rate only if the recovery rate grows with system size.

To understand better the results of Theorem 1, we specialize our analysis to a subclass of dissipative memories that we call \textit{global decoders}. This specific model represents an oversimplified recovery process that takes any error state \textit{directly} into the codespace. Although this model is technically different from many dissipative memories that use the gradual relaxation of error states into the codespace, it serves as 
an exactly
solvable example of decoding dynamics. As we will show below, it also saturates the logical error rate bound established in Theorem 1.

Our global decoder model can be seen as an autonomous version of active error correction with the assumption that rounds of error correction essentially ``take zero time'' (cf.  \cite{Gottesman2013}).
We compress syndrome extraction, any classical post-processing involved in decoding, and the corresponding recovery operation into a pre-compiled procedure---in the form of a series of \textit{jump operators}---that is implemented instantaneously and autonomously.
As such, the model can be seen as an idealized form of autonomous correction
for which 
local syndrome measurements, efficient decoders, and recoveries are implemented instantaneously at random times.

Our global decoder assumption, while seemingly unphysical, demonstrates that 
autonomous memories, including those utilizing seemingly limitless resources, do not guarantee performance comparable to active error correction.
We find that the assumption of immediate system-wide corrections does not automatically yield memory times that scale exponentially, or even polynomially, with system size. For multi-qubit systems undergoing Pauli noise, we show that memory lifetime scales exponentially only with the ratio of two values: the code distance multiplied by recovery process rate and the \textit{total} error rate (i.e. sum of error rates for each physical qubit)---see Theorem~\ref{lem:all_time1} and the discussion around it.  Thus, even if the code distance grows linearly with system size (as in the case of asymptotically good \eczoohref[quantum low-density parity-check (QLDPC) codes]{qldpc}  \cite{lin2022good,leverrier2022quantum, Panteleev2022asymptotically}), the correction rate must grow with system size to yield system-size growing memory lifetime.
By providing a lower bound derived for a stabilizer code with Pauli noise under a few reasonable assumptions, we also show that this scaling cannot be improved in this scenario.

The conclusion implied by the bounds described above is that it is impossible to obtain a global decoder with either exponential or polynomial lifetime using our seemingly powerful engineered dissipation with a constant recovery rate. This result illustrates that the details of the recovery process matter, and the simple fact that it has dissipative gap---i.e. it gets the system into the codespace {in finite time---}alone is not enough for a scalable dissipative memory.

It is noteworthy that the global decoder model combined with Pauli noise is an example of a \textit{Poissonian} process. Such Lindbladian dynamics can be represented as a sum of independently sampled stochastic trajectories, including error and recovery quantum channels. This representation allows us to bound the performance of the global decoder by counting the trajectories that lead to a logical error. This representation also provides a simple intuitive picture of the complex noisy dynamics. This idea may also be useful to obtain bounds on the relaxation time for perturbed local Lindbladians~ \cite{liu2023dissipative}, which would otherwise typically require detailed balance, among other assumptions~ \cite{Alicki:2009,Alicki:2010,lucia2023thermalization,bardet:2023}.

\begin{table*}[t!]
\review{
\begin{tabular}{llll}
\textbf{Logical error bound}  
& \textbf{Statement} & \textbf{Noise} & \textbf{Decoder} \\
& & & \\
Upper bound*               &   Theorem~\ref{generic_bound} &  General  & Strictly error-reducing \\
Upper bound* &   Theorem~\ref{lem:all_time1}    &   Poissonian     &  Global       \\
Upper bound  &   Theorem~\ref{upper_bound_sl1}    &   Uniform Pauli     &  Global      \\
Early-time upper bound       &   Theorem~\ref{upper_bound_early_time}    &
Poissonian &
Global      \\
Lower bound*        &  Theorem~\ref{them:lower_bound}    &   Pauli
 &   Global for stabilizer codes     \\
\end{tabular}
}
\caption{\review{Summary of results. The notions of logical error, strictly error-reducing decoder, and global decoder are explained in Section~\ref{sec:model}. The notion of Poisson noise is discussed in Section~\ref{sec:poisson_error_models}, the definition of Pauli noise is given in Section~\ref{sec:upper_bound1_2}. The asterisk denotes main results.}}
\label{tab:summary}
\end{table*}

\review{The rest of the manuscript is organized as follows. In Section~\ref{sec:model}, we introduce the model of a dissipative memory and, in particular, the global decoder. Then, in Section~\ref{sec:general}, we provide a logical error bound for general dissipative memories that satisfy a few assumptions. To improve this result, in Section~\ref{sec:poisson_error_models} we introduce the Poisson error model, which, in the presence of the global decoder model, provides a better bound on the logical error probability obtained by counting the stochastic trajectories. In Section~\ref{sec:upper_bound1}, we provide an upper bound on the performance of the global decoder with Poisson noise. To justify the tightness of this result, in Section~\ref{sec:lower} we derive a lower bound on the logical error probability for a special case of models with Pauli noise. In Section~\ref{sec:asymptotic}, we derive the asymptotic late-time behavior of the process using this stochastic trajectory representation. Finally, in Section~\ref{sec:examples} we present several examples and study them numerically.  
A summary of the results can be found in Table~\ref{tab:summary}.}

\setcounter{thm}{0}

\section{Model}
\label{sec:model}

Our model of an autonomous quantum memory consists of three ingredients: a noisy quantum system, a codespace, and a recovery map.

 \textbf{Noisy system}.
We consider a noisy quantum system with Markovian noise  \cite{breuer2002theory} and accessible Hilbert space $\mathsf H$ of dimension $D$.
The evolution of the system, in the absence of external control, is characterized by the Lindblad master equation  \cite{gorini1976completely,Lindblad1976}
\be
\frac d{dt} \rho = \Delta\mathcal L_E \rho,
\ee
where $\rho$ is the density matrix of the system, $\Delta$ is the noise rate, and $\mathcal L_E$ is the \textit{error Lindbladian} that takes the form
\be\label{eq:noisy_L}
\mathcal L_E = \sum_{\mu=1}^{N}  \lambda_\mu \Bigl(E_\mu \rho E^\dag_\mu -\frac 12 \{E^\dag_\mu E_\mu,\rho\}\Bigl).
\ee
Here, $\{E_\mu\}_{\mu=1}^N$ are error operators 
and $\lambda_\mu>0$ are real weights satisfying $\sum_\mu\lambda_\mu = N$.

We choose error operators to be sufficiently general to encompass both noise models
in many-body systems, for which the Hilbert space $\mathsf H$ has a tensor-product structure, as well as
in \eczoohref[bosonic modes]{oscillators} \cite{science.aaa2085,PhysRevX.8.021005,gertler2021protecting,lescanne_exponential_2020,de_neeve_error_2022,campagne-ibarcq_quantum_2020,sivak_real-time_2023}, for which $\mathsf H$ is embedded in a single countably infinite space.
As such, the error jump operators $E_\mu$ should be interpreted as the generators of error combinations, or strings, that constitute our model's error set.
The accumulated error operators are $K_{\boldsymbol \mu} := K_{\mu_1}\dots K_{\mu_k}$ labeled by all possible error sequences $\boldsymbol \mu = (\mu_1,\dots,\mu_k)$, where $k\geq 1$, and the elementary errors in the sequence are
\begin{equation}\label{eq:elem_err}
    K_\mu \in \{E_1,\dots, E_N, E_1^\dag E_1,\dots, E_N^\dag E_N \}~.
\end{equation}
The set of elementary errors includes quadratic combinations of jump operators, such as $E^\dag_\mu E_\mu$, manifested in the last term of the Lindblad equation. When we decompose the Lindblad evolution into many stochastic trajectories, this term of the Lindblad equation characterizes the error accumulated between the error jumps. It is important to note, however, that these quadratic errors loose their relevance for unitary error operators $E_\mu$. In these situations, quadratic errors become trivial and should be ignored.

 In the case of \eczoohref[qubit codes]{qubits_into_qubits}, the error jump operators $E_\mu$ are often Pauli operators acting on a single or a few qubits, while the $K_{\boldsymbol \mu}$ are tensor products of such operators. In this scenario, many error sequences are equivalent since, for example, two Pauli errors cancel each other out, $E_\mu^2 = I$. Therefore, we count only unique errors as part of the full set of $K_{\boldsymbol \mu}$.
In another example of a \eczoohref[bosonic mode]{oscillators}  \cite{science.aaa2085,PhysRevX.8.021005,gertler2021protecting,lescanne_exponential_2020,de_neeve_error_2022,campagne-ibarcq_quantum_2020,sivak_real-time_2023} undergoing photon loss, the only jump operator $E_\mu$ is the bosonic annihilation operator, and $K_{\boldsymbol \mu}$ are powers of that operator and 
its adjoint, the creation operator. In all cases, the length of $\boldsymbol{\mu}$, denoted as $|\boldsymbol\mu|$, provides an upper bound on the number of resulting accumulated errors and quantifies their potential severity. We will sometimes refer to $|\boldsymbol\mu|$ as the weight of the error.

\textbf{Codespace}. We consider a logical quantum information encoded in the $q$-dimensional codespace $\mathsf C\subset \mathsf H$.
We further assume that $\mathsf C$ is a \eczoohref[quantum code]{qecc} with \textit{error radius} $\ell$, defined as the largest number for which all errors of weight $\ell$ and below are correctable.
In particular,
a code with error radius $\ell$ satisfies the Knill-Laflamme condition  \cite{PhysRevA.55.900} 
 \be\label{eq:knill_laflamme}
 \begin{split}
PK_{\boldsymbol{\mu}}^{\dag}K_{\boldsymbol{\nu}}&P=C_{\boldsymbol{\mu\nu}}P\\
&\forall\boldsymbol{\mu},\boldsymbol{\nu}\quad\text{such that}\quad|\boldsymbol{\mu}|,|\boldsymbol{\nu}|\leq\ell\,,
 \end{split}
 \ee
where $C_{\boldsymbol{\mu \nu}}$ are constants, and $P$ is the projector on the codespace. At the same time, there exist one or more errors of weight $\ell+1$ such that, when added to the above set of correctable errors, they violate Eq.~\eqref{eq:knill_laflamme}. 
For example, qubit codes undergoing Pauli noise with distance $d$ have error radius $\ell=\lfloor(d-1)/2\rfloor$ with respect to local Pauli noise.
On the other hand, \eczoohref[bosonic rotation codes]{bosonic_rotation}  \cite{PhysRevX.10.011058} undergoing photon loss and being able to detect $S$ photon losses have an error radius of $\lfloor S/2\rfloor$ with respect to loss errors.
\label{defs}

\begin{figure*}[t!]
    \centering
    \includegraphics[width=1\textwidth]{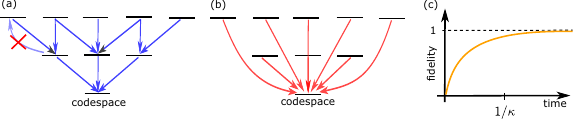}
    \caption{\textbf{Autonomous error correction.} The horizontal lines in (a) and (b) represent different error subspaces in the Hilbert space $\mathsf H$, including the codespace $\mathsf C$. (a) General model of \review{strictly error-reducing} autonomous error correction: the correction process 
    causes transitions between error states, 
    lowering the effective error weight 
    until the system reaches the codespace. \review{No transition to larger-weight errors allowed.} (b) Global decoder: the recovery process causes transitions directly into the codespace.  (c) Starting with an arbitrary initial state, the fidelity 
    (i.e.~the probability of being in the codespace) in the absence of error processes approaches one with rate $\kappa$, for both general and global decoders.}
    \label{fig:main_fig}
\end{figure*}

\textbf{Recovery process}.
We consider the autonomous recovery process described by the master equation
\be\label{eq:recovery_lindblad}
\frac{d}{dt}\rho = \mathcal L_R(\rho),
\ee
where $\mathcal L_R$ is the Lindbladian generator of the recovery dynamics. For this work, we specialize to a class of maps that can be described as decoders by making three major assumptions.

First, we assume that the recovery dynamics has a certain non-zero convergence rate, i.e. that the solution $\rho(t) : = e^{\mathcal L_Rt}(\rho)$ has a well-defined limit $\rho_\infty:=\lim_{t\to\infty}\rho(t)$ that generally depends on $\rho(0)$, and there exist both time-independent rate $\kappa>0$ and a parameter $0<\chi<\infty$ such that
\be\label{eq:contraction_lR}
\forall \rho\in \mathsf N,\, t>0: \qquad \|\rho(t)-\rho_\infty\|_1\leq \chi e^{-\kappa t}~,
\ee
where we denote $\mathsf N = \{\rho \in \mathsf {\rm End}(\mathsf H): \rho\geq0, \Tr\rho = 1\}$ to be the set of all positive semidefinite operators with unit trace and $\|A\|_1: = \Tr \sqrt{A^\dag A}$ is the trace norm. The $\kappa$ parameter plays the role of an effective correction rate. For example, this condition holds for systems that exhibit a generalized notion of the rapid mixing condition  \cite{cubitt2015stability,brandao2015area} applied to states that share the same fixed point, i.e.
\be\label{eq:contraction_rm_cond}
\begin{split}
\forall \rho, \sigma \in \mathsf N: \rho_\infty = &\sigma_\infty,\, t > 0:\\
&\|\rho(t) - \sigma(t)\|_1 \leq \chi e^{-\kappa t},
\end{split}
\ee
 where $\kappa$ is a constant and parameter $\chi$ is either a constant or depends at most polylogarithmically on the dimension of the Hilbert space $D$ (for finite-dimensional systems). For the sake of generality, we will keep arbitrary $\chi = \chi(D)<\infty$ as a function of $D$, although the result we present is primarily relevant to systems where it is a constant, $ \chi  = O(1)$.
As we will see below, for such cases, the ratio of the correction rate to the noise rate determines the ability of the decoder to counteract the noise.

The second assumption implies that the dynamics described in Eq.~\eqref{eq:recovery_lindblad} constitute a legitimate recovery process. This statement imposes a condition on the late-time recovery map, defined as $\mathcal{R} := \lim_{t \to \infty} \exp(\mathcal{L}_R t)$, requiring it to correct all errors within the error radius $\ell$, i.e.,
\begin{align}\label{eq:assumptions}
    \forall\rho_0 \in \mathsf{N}_0,\,|\boldsymbol{\mu}|,|\boldsymbol{\nu}|\leq\ell, \quad \mathcal R(K_{\bsmu}\rho_0K^\dag_{\bsnu}) \propto \rho_0,
\end{align}
 where $\mathsf{N}_0 = \{\rho \mid \operatorname{Tr}(\rho) = 1, \, \rho \geq 0, \rho \in {\rm End}(\mathsf C)\oplus 0\}$ is the set of density operators in the codespace $\mathsf C$. Here, ${\rm End}(\mathsf C) \oplus 0$ denotes the space of operators with range and support in the subspace $\mathsf C\subset \mathsf H$. In other words, any state affected by a correctable error is restored to the codespace without introducing any logical errors.

The third assumption requires that the decoder process be \textit{strictly error-reducing}. This means that any recovery operation performed at intermediate times does not increase the weight of the error state (see Fig.~\ref{fig:main_fig}a). In other words, for all $t \geq 0$, there exists a representation
\be\label{eq:good_correction}
\exp(\mathcal L_R t)(K_{\bsmu}\rho_0K^\dag_{\bsnu}) = \sum_{\bsmu'\bsnu'}a_{\bsmu\bsnu,\bsmu'\bsnu'}(t)K_{\bsmu'}\rho_0K^\dag_{\bsnu'},
\ee
where $a_{\bsmu\bsnu,\bsmu'\bsnu'}(t)=0$ if $|\bsmu'|>|\bsmu|$ or $|\bsnu'|>|\bsnu|$. 
This idealization may technically exclude certain high-performing decoders that might increase the error weight, even in negligible subsets of error configurations. Nonetheless, it enables us to derive our main result presented in Theorem~\ref{generic_bound}. Note that the right side of Eq.~\eqref{eq:good_correction} may admit multiple representations in terms of error operators and different coefficients \( a_{\bsmu\bsnu,\bsmu'\bsnu'} \). This arises because two distinct operators \( K_{\bsmu} \) and \( K_{\bsmu'} \) can have the same action on the code space, i.e., \( K_{\bsmu}|\psi_0\> \propto K_{\bsmu'}|\psi_0\> \) for all \( |\psi_0\> \in \mathsf{C} \). It suffices that at least one such representation satisfies the condition mentioned above.
Starting from any initial state $\rho(0) = \rho_0$, the evolution to time $t\geq0$ under the Lindbladian satisfying the first condition in Eq.~\eqref{eq:contraction_lR}  can be formally written as
\be\label{eq:full_model}
\rho(t) = e^{-\kappa t}\mathcal{K}_t(\rho_0)+(1-e^{-\kappa t})\mathcal{R}(\rho_0),
\ee
where $\mathcal K_t: = e^{\kappa t}\left(\exp(\mathcal L_R t)-(1-e^{-\kappa t})\mathcal R\right)$ represents the early-time evolution and whose action can always be bounded as
\be
\begin{split}
\|\mathcal K_t(\rho)\|_1 & = \|\rho_\infty+e^{\kappa t}\left(\rho(t)-\rho_\infty\right)\|_1\\
& \leq \|\rho_\infty\|_1+e^{\kappa t}\|\rho(t)-\rho_\infty\|_1\leq \chi+1,
\end{split}
\ee
where we used the contraction property in  Eq.~\eqref{eq:contraction_lR}.

\reviewtwo{Most of our results concern the special case of \textit{global decoder models},} which is a special case of the model in Eq.~\eqref{eq:full_model} where we additionally remove any partial recovery, i.e.\ we assume $\mathcal K_t \equiv \mathcal I$.
In this case, the Lindladian for the recovery process is given by
\be \label{eq:noiseless_lind}
\begin{split}
&\mathcal L_R = \kappa\bigl(\mathcal R(\rho)-\rho\bigl),\\
&\rho(t) = e^{-\kappa t}\rho_0+(1-e^{-\kappa t})\mathcal R(\rho_0).
\end{split}
\ee
In this model, the recovery dynamics return error states directly to the codespace (see Fig.~\ref{fig:main_fig}b).
The recovery \(\cal R\) may act non-locally in order to recover information into the codespace in one step.
The definition of the dissipative gap
remains the same for both general and such global decoders (see Fig.~\ref{fig:main_fig}c). \reviewtwo{Note that since in this case $\|\mathcal K_t\| = \|\mathcal I\|= 1$, the results for the global decoder can be compared to the general case above by effectively taking $\chi = 0$.}

The primary distinction between global decoders and the broader category of decoders satisfying Eq.~\eqref{eq:assumptions} lies in the applicability to real-world decoders. This broader class of decoders may include decoders for bosonic codes (in the case of an ideal implementation), as implemented in Ref.~ \cite{sivak_real-time_2023}. At the same time, it may be difficult to construct a decoder for qubit systems that is both local and consistent with Eq.~\eqref{eq:assumptions}. However, local decoders, such as those based on the sweep rule  \cite{kubica2019cellular}, must satisfy these conditions with an error that decreases rapidly with code size. The rationale is that even though the sweep rule may locally produce more physical errors, the likelihood that these operations will increase the total number of errors is exponentially small in the number of physical qubits. To prove such a fact, it is necessary to exploit the properties of a particular decoder. This aspect allows a direct comparison between the general case and local decoders. In contrast, global decoders are fundamentally different from local decoders in their approach: they map error states to the codespace in a single step. Although this one-step requirement complicates the practical implementation of global decoders, the simple mathematical framework of this model provides broader possibilities for analytical exploration.\\

\textbf{Logical errors and critical error rate.} The combination of the recovery process and the error process results in a dynamical equation that describes the autonomous quantum memory:
\be\label{eq:markov_process}
\frac d{dt} \rho = \mathcal L(\rho) := \mathcal L_R(\rho)+\Delta\mathcal L_E(\rho).
\ee
Our goal is to explore the performance of such a quantum memory. In particular, we aim to find a regime where, in the presence of noise ($\Delta>0$), the probability of a logical error after recovery for a family of codes with increasing $\ell$ vanishes polynomially or exponentially in the limit $\ell\to\infty$. 

As a measure to quantify logical errors, we consider the trace distance as a function of time
between two initially orthogonal logical pure states in the codespace. Let $G$ be the set of all pairs of orthogonal states in codespace, $G =\{(|\psi_0\>,|\psi_1\>): |\psi_0\>,|\psi_1\>\in\mathsf C, \<\psi_0|\psi_1\> = 0\}$. Given this setup, we proceed to define an error measure as
\be\label{eqs:trace_dis}
\delta(t) := 1-\min_{|\psi_0\>,\, |\psi_1\>\in G}T\Bigl(\exp(\mathcal Lt)[P_0], \exp(\mathcal Lt)[P_1]\Bigl),
\ee
where $P_0 = |\psi_0\>\<\psi_0|$, $P_1 = |\psi_1\>\<\psi_1|$, and $T(\rho,\sigma)$ is the trace distance.
\footnote{The trace distance is the maximum probability of distinguishing between two quantum states and is expressed as $T(\rho,\sigma) = \frac 12 \|\rho-\sigma\|_1$, where $\|A\|_1: = \Tr \sqrt{A^\dag A}$ is the trace norm.}

The error measure $\delta(t)$ vanishes if and only if there exists a recovery map that always returns the logical qubit to its initial, error-free configuration.
However, such a map may be complex and a priori unknown. As an alternative, we define a simpler logical error measure that quantifies our ability to recover information using the recovery map $\mathcal R$ based on the fidelity of recovery starting from a pure initial state:\review{
\be\label{eq:overlap_log_error}
\begin{split}
\epsilon(t) := 1-\min_{|\psi_0\>\in \mathsf C}\Tr\Bigl(|\psi_0\>\<\psi_0|\mathcal R\exp(\mathcal L t)|\psi_0\>\<\psi_0|\Bigl),
\end{split}
\ee
where the minimum is taken over pure states in the codespace.}

For a single logical qubit, the two error measures $\delta(t)$ and $\epsilon(t)$ are related by (see \refx[s]{appendix_error_measures})
\be\label{eq:error_ineq}
\delta(t)\leq 2\epsilon(t)~.
\ee
We will focus below on the measure $\epsilon(t)$ in Eq.~\eqref{eq:overlap_log_error}. However, some results also apply to $\delta(t)$ (see Theorems~\ref{generic_bound} and \ref{upper_bound_early_time} below). \review{Strictly speaking, both measures $\delta(t)$ and $\epsilon(t)$ are functions of the recovery map $\mathcal{R}$ and the noise process, as well as the corresponding rates. For simplicity, we omit explicit dependence on these parameters in our notations. 
}\\

\section{General bound for strictly error-reducing decoders}
\label{sec:general}

In this section, we present a result that shows that autonomous memories always exhibit exponential memory lifetimes for sufficiently small noise rate. To do this, we must first introduce a parameter that quantifies the strength of the noise.  We first introduce the correctable error space ${\mathsf E}$ consisting of the range of error states $|\psi_{w\boldsymbol\mu}\> \propto K_{\boldsymbol\mu}|w\>$, where $|w\>\in \mathsf C$ denotes the set of $q$ mutually orthogonal codewords arbitrarily chosen in the codespace $\mathsf C$, $q$ is the dimension of the codespace, and $K_{\boldsymbol\mu}$ are errors of weight less or equal to $\ell$ that satisfy the condition in Eq.~\eqref{eq:knill_laflamme}. This subspace satisfies $\mathsf C \subset {\mathsf E}\subset \mathsf H$. Then we can introduce $\|\mathcal L_E\|_{1\to1,\mathsf E}$ as the contraction norm of the superoperator $\mathcal L_E$ over the subspace ${\mathsf E}$, i.e. 
 \be\label{eq:contr_norm_E_def}
 \|\mathcal L_E\|_{1\to1,\mathsf E} : = \max_{O\in \mathsf {\rm End}(\mathsf E)\oplus 0} \frac{\|\mathcal L_E(O)\|_1}{\|O\|_1},
 \ee
 where ${\rm End}(\mathsf E) \oplus 0$ denotes the space of operators with range and support in the subspace $\mathsf E\subset \mathsf H$ and $\|\cdot\|_1$ is the 
 trace norm.
By definition, for finite systems this norm is bounded by the full norm, i.e. $\|\mathcal L_E\|_{1\to1,\mathsf E}\leq \|\mathcal L_E\|_{1\to1,\mathsf H} \equiv \|\mathcal L_E\|_{1\to1}$. The norm $\|\mathcal L_E\|_{1\to1,\mathsf E}$ could remain well-defined even if the norm $\|\mathcal L_E\|_{1\to1}$ becomes unbounded for systems with a formally infinite Hilbert space dimension, e.g.~in the bosonic codes. However, for systems with bounded local Hilbert space, such as qubits on a lattice, the difference between these norms is less significant, so one could replace $\|\mathcal L_E\|_{1\to 1,\mathsf E}$ by $\|\mathcal L_E\|_{1\to 1}$ below without losing the essence of the result. In particular, for single-qubit noise, both norms are expected to be bounded and to have the same linear scaling with the number of physical qubits $n$.

Then we can write our first result as the following theorem: 
\begin{thm} \label{generic_bound} For an arbitrary error model in Eq.~\eqref{eq:noisy_L} and recovery map in Eq.~\eqref{eq:full_model} satisfying the \reviewtwo{ assumptions outlined in Eqs.~\eqref{eq:contraction_lR}, \eqref{eq:assumptions} and \eqref{eq:good_correction}, the error rate is bounded by}
\be\label{eq:general_theorem}
\begin{split}
\epsilon(t),\delta(t)&\leq\frac 1{\chi+1}\left(\frac{(\chi+1)\Delta\|{\cal L}_{E}\|_{1\to1,\mathsf E}}{\kappa}\right)^{\ell+1}F_{\ell}\bigl(\kappa t\bigl)\,\\
&\leq\frac 1{\chi+1}\left(\frac{(\chi+1)\Delta\|{\cal L}_{E}\|_{1\to1,\mathsf E}}{\kappa}\right)^{\ell+1}\kappa t\\
&\equiv\frac 1{\chi+1}\eta^{\ell+1}\kappa t~,
\end{split}
\ee
where \reviewtwo{$\eta := (\chi+1)\Delta\|{\mathcal L}_{E}\|_{1\to 1,\mathsf E}/\kappa$}, $F_\ell(z) = zg(\ell,z)-\ell g(\ell+1,z)\leq z$ and $g(\ell,z)$ is the regularized lower incomplete Gamma function  \cite{NIST:DLMF}.
\end{thm}

\begin{figure*}[t!]
    \centering
    \includegraphics[width=1\textwidth]{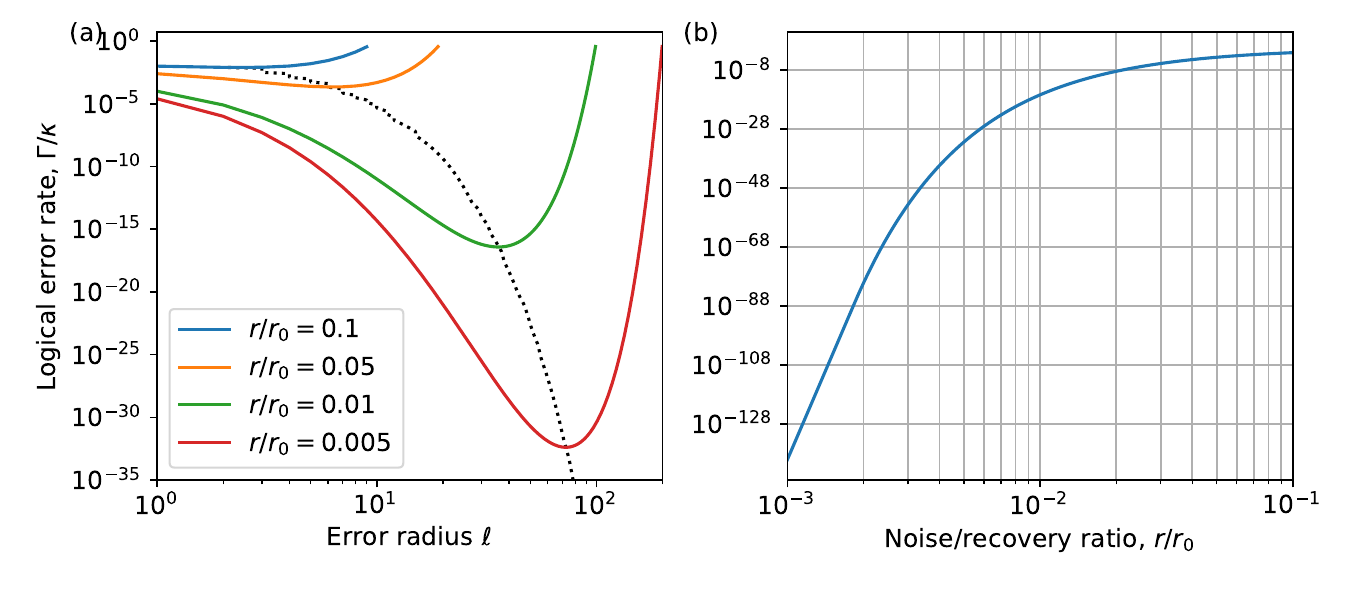}
    \caption{\textbf{Logical rate for generic qubit-based models.} (a) The logical rate as a function of the error radius for different values of the noise-to-recovery ratio $r$. The function has a minimum for $\ell_{\rm min} = O(r_0/r)$, the positions of the minima are shown as the dotted black curve. (b) The minimum logical error rate ratio as a function of the noise-to-recovery ratio, see Eq.~\eqref{eq:gamma_min}.}
    \label{fig:soft_threshold}
\end{figure*}

 To prove this result, we use Dyson's perturbative expansion of the evolution superoperator with $\Delta$ as a small parameter. We show that the first $\ell$ perturbative terms in the series vanish, and use a re-summation of the remaining terms to obtain a non-perturbative expression. The proof can be found in \refx[s]{appendix_generic_bound}. Note that this result holds not only for finite-dimensional systems, but also for certain (infinite-dimensional) bosonic systems where ${\mathsf E}$ is a finite subspace. This bound can be seen as a generalization of the model-specific bound proposed in  \cite{lebreuilly2021autonomous}, which essentially describes the special case in Eq.~\eqref{eq:noiseless_lind} (see next section for better bound for this case). In contrast, the bound in Eq.~\eqref{eq:general_theorem} holds for any recovery dissipation that satisfies the criteria listed in the previous section. It also captures nonlinear early-time error scaling.

For short times, defined as $\kappa t\ll1$, the special function has polynomial scaling $F_\ell(\kappa t)\sim(\kappa t)^{\ell+1}$, and the error bound scales as $\epsilon(t) = O(((\chi+1)\Delta \|\mathcal L_E\|_{1\to1,\mathsf E} t)^{\ell+1})$. This reflects the fact that, in this limit, the error is described by perturbation theory and the lowest non-vanishing terms are of order $\ell+1$. 
In contrast, in the non-perturbative regime $\kappa t\gg 1$, the logical error bound grows linearly and its rate is proportional to $\eta^{\ell+1} $, where \reviewtwo{$\eta = (\chi+1)\Delta\|{\mathcal L}_{E}\|_{1\to 1,\mathsf E}/\kappa$.} Therefore, this rate is exponentially suppressed in $\ell$ if the noise is small enough that $\eta<1$.

\review{Let us analyze this suppression for certain multi-qubit codes for constant recovery rate $\kappa$ and \reviewtwo{prefactor $\chi$}. In multi-qubit systems, the Lindblad operators \reviewtwo{are expected to satisfy the symptotic equivalence $\|\mathcal L_E\|_{1\to1,\mathsf E}\leq \|\mathcal L_E\|_{1\to1} \sim cn$ in the limit of large $n$, where} the dimensionless multiplication factor $c$ reflects the rate of local error processes. One can also find QLDPC codes, such as those defined on expander graphs  \cite{lin2022good,leverrier2022quantum,Panteleev2022asymptotically}, whose radius is proportional to the number of qubits, such that $\ell \geq \alpha n$ for some $\alpha<1$.} As a result, the logical error rate is $\Gamma :=\delta(t)/t\leq (\ell r/r_0)^{\ell+1}\kappa$, where $r = \Delta/\kappa$ is a renormalized noise to recovery ratio and \reviewtwo{$r_0 \sim \alpha/(\chi+1)c$ in the limit of large $n$}. Unlike in active error correction, the logical error rate bound does not become arbitrarily small for ever-increasing $\ell$. In fact, for small $r\ll 1$, the bound has its minimum at $\ell \simeq r_0/er$, where $e$ is the base of the natural logarithm. This dependency is shown in Fig.~\ref{fig:soft_threshold}(a). The minimum logical error rate bound is
\be\label{eq:gamma_min}
\Gamma_{\rm min} = O\left(\kappa e^{-r_0/r}\right), \qquad r = \Delta/\kappa.
\ee
 This is the minimum error rate bound for fixed $\kappa$. Therefore, we show that, for constant $\kappa$, there exists a universal upper bound on the error rate that is sufficiently smaller than the original error rate $\Delta$ once $r<r_0$. In this regime, a small improvement in the error rate $\Delta$ yields, at least, an exponential  improvement in the logical error rate bound, see Fig.~\ref{fig:soft_threshold}(b). This is a universal ``soft threshold'' result applicable to autonomous error-correcting codes based on linear-error-radius QLDPC codes.

Of course, we know that this bound is not tight: for some autonomous codes we can have arbitrarily small errors for some $\Delta/\kappa$ below the threshold  \cite{Alicki:2010,pastawski:2011,liu2023dissipative}. However, our result suggests that, in general, to get such unbounded error reduction, we must have \reviewtwo{$\kappa\propto \|\mathcal L_E\|_{1\to1, \mathsf E}$}, i.e., growing with system size. In fact, by considering a particular example below in Section~\ref{sec:lower}, we show that this scaling condition cannot be relaxed in general.

\section{Global decoders and a Poissonian noise}
\label{sec:poisson_error_models}

In this section, we derive tighter bounds than the generic-decoders bound in Eq.~\eqref{eq:general_theorem} by focusing on global decoders in Eq.~\eqref{eq:noiseless_lind} and specializing to unitary error jump operators (including Pauli noise).
Such bounds allow for a better grasp of the capabilities of a recovery rate $\kappa$ that is sublinear in the number of qubits.
 
\begin{dfn}[\textbf{Poissonian noise}]\label{dfn:pois_noise} The error model is Poissonian if the error operators in Eq.~\eqref{eq:noisy_L} \review{are unitary, i.e. they} satisfy
$
E^\dag_\mu E_\mu = I
$,
where $I$ is the identity matrix. 
\end{dfn}

We refer to this dynamics as \textit{Poissonian} as it maps directly to a Poissonian point process  \cite{kossakowski1972quantum}. 
While this restriction still covers generic Pauli noise for qubit, \eczoohref[modular-qudit]{qudits_into_qudits}  \cite{gottesman1999}, and \eczoohref[Galois-qudit]{galois_into_galois} codes  \cite{bierbrauer2000quantum,ketkar2006nonbinary}, it does not include processes such as photon loss or additive Gaussian white noise applicable to bosonic codes.

Assuming Poissonian errors, we can rewrite the Lindblad equation in Eq.~\eqref{eq:markov_process} as
\be
\frac{d}{dt}\rho=\mathcal{L} (\rho) = \gamma\sum_{\mu=0}^{N}p_{\mu}\bigl(\mathcal{E}_{\mu}(\rho)-\rho\bigl)\,,\label{eq:standard_{f}orm}
\ee
where 
\be{\cal E}_{\mu}(\rho)=\begin{cases}
\mathcal{R}(\rho) & \mu=0\\
E_{\mu}\rho E_{\mu}^{\dag} & \mu>0
\end{cases}\nonumber
\ee
and $\gamma = \kappa+ N\Delta$,
$p_0 = \kappa/\gamma$, and $p_{\mu>0} = \lambda_\mu\Delta/\gamma$.
Notably, parameters $p_\mu$ are positive, satisfy the normalization condition $\sum_\mu p_\mu = 1$, and can therefore  be treated as probabilities. 

The analytical solution of Eq.~\eqref{eq:standard_{f}orm} can be obtained from the exponentiation of $\mathcal{L}$ and the consequent decomposition of the exponent using Taylor series, which takes the form of a sum of multiple stochastic trajectories:
\begin{equation}
\label{eq:poisson_form}
\begin{split}
\exp(\mathcal{L}t)&=e^{-\gamma t}\exp\left(\gamma t\sum_{\mu=0}^{N}p_{\mu}\mathcal{E}_{\mu}\right)\\ 
&=\sum_{\boldsymbol{\mu}\in F}p(\boldsymbol{\mu},t)\mathcal{E}_{\boldsymbol{\mu}},
\end{split}
\end{equation}
where the set $F$ includes individual 
\textit{trajectories} $\boldsymbol \mu$ of any length, including consequent errors and recoveries, and 
$\mathcal{E}_{\boldsymbol{\mu}} = \mathcal{E}_{\mu_k}\circ\dots \circ\mathcal{E}_{\mu_1}$ is the trajectory map, where $\circ$ denotes the composition of superoperators (we omit it below). The probability of a trajectory of error weight $|\boldsymbol\mu|=k$ occurring at time $t$ has the form

\be\label{eq:stoch_prob}
p(\boldsymbol \mu,t) = \frac{1}{k!}(\gamma t)^k e^{-\gamma t}p_{\mu_1}\dots p_{\mu_k},
\ee
where probabilities $p_{\mu}$ are defined below Eq.~\eqref{eq:standard_{f}orm}.
It is easy to confirm that, for a given $t$, the probabilities $p(\boldsymbol \mu,t)$ for any time $t$ sum to one. 
As a result, we can interpret the dynamics as a homogeneous Poisson point process, in which error and recovery events occur at random times following a Poisson distribution with average spacing $\gamma^{-1}$.

\review{\section{Upper bounds for global decoder with Poissonian noise}\label{sec:upper_bound1}}

\review{In this section, we first present a rigorous upper bound on the error measure $\epsilon(t)$ for $n$-qubit codes subject to Poissonian noise (Theorem~\ref{lem:all_time1}). We then present two other upper bounds that give a better estimate of the logical error rate in more specific settings (Theorem~\ref{upper_bound_sl1} and Theorem~\ref{upper_bound_early_time}). } 

\review{
Many of the quantum error-correcting codes can still reliably store the logical information when they are subjected to much larger noise than allowed by the Knill-Laflamme condition in Eq.~\eqref{eq:knill_laflamme}. This is due to the existence of a \textit{error threshold} of the codes. To account for the improved error tolerance beyond the Knill-Laflamme condition, we introduce the \textit{tolerable error weight}.    
}

\review{\begin{dfn}[\textbf{Tolerable error weight}] \label{dfn:tol_error_rate}Consider a family of $n$-qubit error-correcting codes for increasing $n$. Each code has a codespace $\mathsf{C}$, code distance $d = d(n)$, an error channel $\mathcal{E}$ and a recovery map $\mathcal R$. We say that the code family has a tolerable error weight $h$ if $h = h(n)$ is an integer-valued function such that, 
for any $|\psi\rangle\in\mathsf{C}$ and non-negative integer $k\leq h$, the following inequality holds:
\be \label{sm:eq:dfn_threshold}
\mathcal R\mathcal E^k\bigl(|\psi\>\<\psi|\bigl)- (1-\xi)|\psi\>\<\psi|\geq 0,
\ee
where $\xi = 2^{-\Omega(d)}$ is independent of $\ket{\psi}$. Let $G$ be the set of all possible constants $f$ such that $h(n) = \lfloor fn\rfloor$ is a tolerable error weight. We say that the code family has a threshold $f_c$ if $f_c = \sup G >0$.\label{sm:def:threshold}
\end{dfn}
The inequality in Eq.~\eqref{sm:eq:dfn_threshold} means that the eigenvalues of the operator on the left side of the inequality are non-negative. Since $\mathcal R\mathcal E^k$ is a quantum channel, this inequality is equivalent to the statement $\mathcal R\mathcal E^k|\psi\>\<\psi|= (1-\xi)|\psi\>\<\psi|+\sum_i p_i|\psi_i\>\<\psi_i|$ for some $|\psi_i\>\in\mathsf{C}$ and $p_i\geq 0$ that satisfy $\sum_i p_i\leq \xi$. In other words, for any initial logical state, the error occurs with $\xi$-small probability. By this definition, an $n$-qubit code with a threshold can recover the encoded quantum information with arbitrary precision, even if it is subjected to a large number of $k\leq h$ error channel rounds. Examples of error-correcting codes that have a threshold against single-qubit Pauli noise include the repetition code and many quantum stabilizer codes such as the toric code and color codes.
For all codes, error weights below error radius $\ell$ are by definition tolerable. }

\review{With the concept of tolerable error weight, we are now ready to state the upper bound for systems subject to Poissonian noise.}
\begin{thm}\label{lem:all_time1}
Consider a family of $n$-qubit error-correcting codes for increasing $n$. Each code has a codespace $\mathsf{C}$ with code distance $d = d(n)$, a Poissonian noise model $\{E_\mu\}$, a recovery map $\mathcal R$ and a tolerable error weight $h = h(n)$ \review{with respect to the total error channel $\mathcal{E}=\frac{1}{N}\sum_{\mu=1}^N\mathcal{E}_{\mu}$}. Then there exists a small parameter $\xi=2^{-\Omega(d)}$ such that the logical error probability for the global decoder in Eq.~\eqref{eq:poisson_lemma}, for any $t\geq 0$, satisfies
\begin{align}\label{eq:late_time_bbound_sl}
\begin{split}
 \epsilon(t)\leq 
 1-\exp\Bigl(-(1-\xi)&N\Delta\left(\frac{N\Delta}{\kappa+N\Delta}\right)^{h}t\\
 &-\xi(\kappa+N\Delta)t\Bigl).
 \end{split}
\end{align}
\end{thm}

We provide a proof of the theorem in \refx[s]{appendix_b}. 
Similar to the asymptotic estimate in Section~\ref{sec:asymptotic}, the proof utilizes the Poissonian picture.\footnote{It is worth noting that this bound holds for a more generic class of errors: instead of requiring $E^{\dag}_{\mu}E_{\mu} = 1$ for all $\mu$, the jump operators only have to satisfy $\sum_{\mu}E^{\dag}_{\mu}E_{\mu}=N$. This class of errors also has the same convenient properties as the Poissonian model~\eqref{eq:poisson_form}. Thus, the bound also applies in situations involving non-unitary errors such as those described by Pauli ladder operators $\sigma^\pm=(X\pm iY)/2$, provided $\sum_{\mu}E^{\dag}_{\mu}E_{\mu}=N$.} 
It uses the fact that the dynamics of the system is an ensemble average over trajectories where the single-shot recovery and the errors happen stochastically. 
Along a given trajectory, the occurrence of a recovery event resets the system back to the codespace. If no more than $h$ errors take place between any such consecutive resets, the recovery is almost guaranteed to send the system back to the correct codeword (up to a small failure rate $\xi$). 
We can therefore obtain an upper bound for the logical error probability by lower bounding the probability of trajectories consisting of only such faithful resets. 

Consider an ideal recovery map that corrects only errors within the error radius, i.e. a tolerable error weight $h = \ell$ and $\xi = 0$. It therefore follows that
\be\label{eq:upper_bound_1}
\epsilon(t)
\leq
1-\exp\left(-\frac{N \Delta t}{(1+\kappa/N\Delta)^\ell}\right).
\ee
This result can be compared to Eq.(3) in  \cite{lebreuilly2021autonomous}. A notable difference is that the current logical error bound in Eq.~\eqref{eq:upper_bound_1} remains meaningful even for $N\Delta/\kappa>1$ in the regime $\ell\gg1$.  This upper bound is also shown in Fig.~\ref{fig:fig1} along with the others.

\begin{figure*}[t!]
    \centering
    \includegraphics[width=1\textwidth]{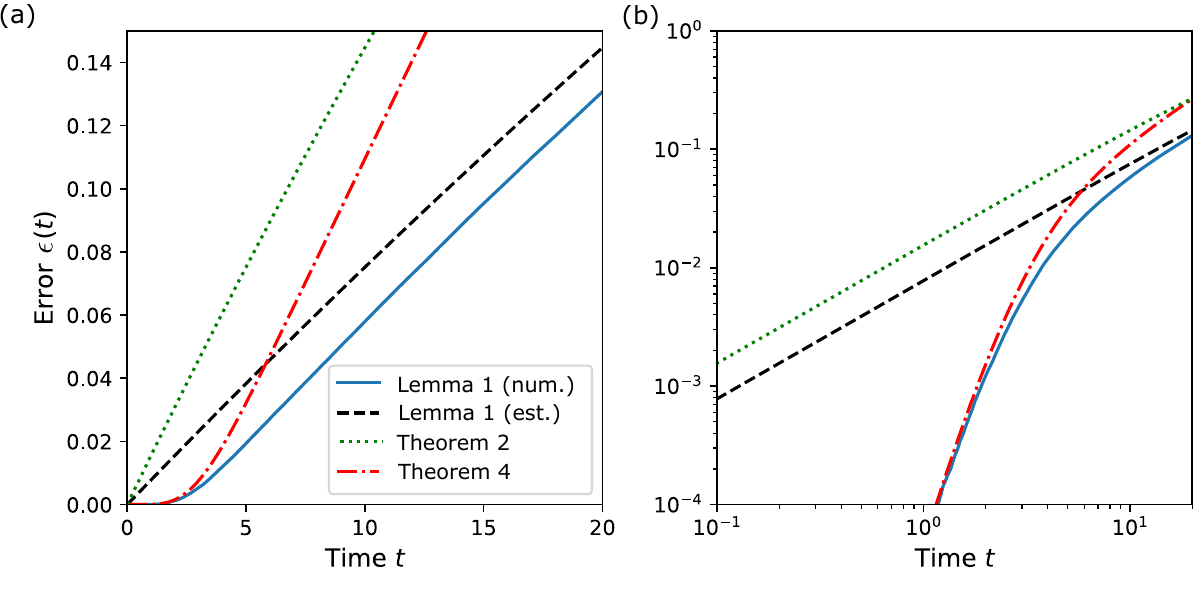}
    \caption{\review{\textbf{The results for a Poissonian error model}. 
The bounds and estimates for the logical error $\epsilon(t)$ for a code with $\ell = 6$ and $\kappa = N\Delta = 1$. The plots illustrate: the bound in Eq.~\eqref{eq:upper_bound_1} is shown as a dashed green curve, the bound in Eq.~\eqref{eq:early-bound} is shown as a dotted red curve, the bound in Eq.~\eqref{eq:poisson_lemma} with $p(t)$ evaluated numerically is shown as a blue curve, and its estimate in Eq.~\eqref{eq:asymptotic dynamics} is shown as a dashed black curve. Panel (a) plots the linear scale, while panel (b) plots the logarithmic scale.
}}
    \label{fig:fig1}
\end{figure*}

\review{Below, we consider two additional bounds for Poissonian noise that work for more specific settings. In the first bound, we assume that the noise is Pauli noise. The second bound is less tight but applies to all types of Poissonian noise and works better for early times.}

\review{\subsection{Upper bound for uniform Pauli noise}
\label{sec:upper_bound1_2}}

\review{The upper bound can be further improved for the system consisting of many qubits. In these systems, a stronger result can be obtained if we restrict the errors to Pauli noise (see definition below) and assume that the rate is uniform. This model is one of the standard error models considered in practical quantum computation.}

\review{\begin{dfn}[\textbf{Pauli noise}]\label{dfn:pauli_noise} Consider a system consisting of qubits. The channel $\mathcal E(\cdot)=\sum_\mu E_\mu \cdot E_\mu$ is defined as Pauli noise if $E_\mu \subset \mathsf P$, where $\mathsf  P$ is the set of generalized multi-qubit Pauli operators.
\end{dfn}}

\review{An example of a Pauli noise model is depolarizing noise, where the errors at each location $i$ are described by three Pauli jump operators, $E_{3i+1} = X_i,E_{3i+2} = Y_i$, and $E_{3i+3} = Z_i$, where $i=0,\dots,n-1$. Another example is dehasing $E_i = Z_i$.}

\review{To incorporate the error threshold while exploiting the simple structure of Pauli noise, we consider a subset of codes subject to uniform Pauli noise, i.e., in Eq.~\eqref{eq:noisy_L} $\lambda_{\mu} = \lambda_{\mu'}$ for all $\mu,\mu'$, and have a tolerable error weight as defined below. }

\review{\begin{dfn}[\textbf{Tolerable error weight for uniform Pauli noise}]\label{dfn:twpn} Consider the conditions of Definition~\ref{sm:def:threshold} and additionally assume that error channel $\mathcal E$ is Pauli noise. Then we say that the code family has a tolerable error weight $h$ for uniform Pauli noise if $h = h(n)$ is an integer-valued function such that, 
for any $|\psi\rangle\in\mathsf{C}$ and non-negative integer $k\leq h$, the following inequality holds:
\be \label{sm:eq:dfn_threshold_pauli}
\mathcal{R}\mathcal{Q}_k\bigl(|\psi\>\<\psi|\bigl)- (1-\xi)|\psi\>\<\psi|\geq 0,
\ee
where $\xi = 2^{-\Omega(d)}$ is independent of $\ket{\psi}$. Here, $\mathcal{Q}_0 = \mathcal{I}$ is the identity channel and $\mathcal{Q}_k(\cdot) = \frac{1}{|\textbf{S}_k|}\sum_{\{\mu\}_k\in S_k} E_{\{\mu\}_k}(\cdot) E_{\{\mu\}_k}^{\dag}$, where $\textbf{S}_k$ is the set of all the possible $k$ distinct error indices and $E_{\{\mu\}}=\prod_{\mu\in\{\mu\}}E_\mu$.
\end{dfn}}

\review{ We expect many common error-correcting codes, such as quantum stabilizer codes, to have a tolerable error weight for uniform single-qubit Pauli noise. As a consistency check, if a code family is subject to uniform Pauli noise and has a tolerable error weight for uniform Pauli noise, then it has a tolerable error weight according to Definition~\ref{sm:def:threshold}. }

Since Pauli operators mutually commute or anticommute and their square is identity, two identical errors in the sequence cancel each other out. This fact means that some of the physical error sequences with length greater than $h$ may not contribute to the logical error due to such cancellation. This allows us to improve the upper bound, which leads to the following theorem.
\begin{thm} \label{upper_bound_sl1}
Under the conditions of Theorem~\ref{lem:all_time1} and consider a code family that is \review{subject to Pauli noise model with uniform rates, i.e. in Eq.~\eqref{eq:noisy_L} $\lambda_{\mu} = \lambda_{\mu'}$ for all $\mu,\mu'$, and has a tolerable error weight $h(n)$ for uniform Pauli noise}, the logical error in Eq.~\eqref{eq:overlap_log_error} satisfies
\begin{equation}\label{eq:ineq0}
    \epsilon(t)
    \leq 1-
   \exp\Bigl(-(1-\xi)N\Delta s_1t-\xi(\kappa+N\Delta)t\Bigl),
\end{equation}
where $s_1$ is the solution to the recurrence relation
\begin{equation}\label{eq:recurrence}
     s_v = \frac{v}{N}p_1s_{v-1} + \left(1-\frac{v}{N}\right)p_1s_{v+1},\ s_0 = 0,\ s_{h+1} = 1,
\end{equation}
with $p_1 = N\Delta/(\kappa+N\Delta)$ and $v\in\{0,1,2,\cdots, h+1\}$.
\end{thm}
The proof is similar to that of Theorem~\ref{lem:all_time1} and is given in \refx[s]{appendix_c}.
The recurrence relation 
corresponds to a classical random walk, where a left or right move corresponds to an application of an error operator that increases or reduces the weight of the resulting 
total error\footnote{
Here, we 
count the error weight 
after cancelling all the repeating errors in a jump trajectory described by a string of elementary errors (Eq.~\eqref{eq:elem_err}).}. 
Note that the recurrence relation Eq.~\eqref{eq:recurrence} always has a solution. This can be obtained by first initializing the recurrence with $s_0 = 0$ and $s_1 = 1$, and then dividing the resulting sequence $\{s_v\}$ by $s_{h+1}$ (to satisfy the boundary condition $s_{h+1}=1$).

We analyze the recurrence relation in Eq.~\eqref{eq:recurrence} numerically. In particular, we compute $s_1$ for different $\kappa$ and $\Delta$, also varying $N$ up to $10^7$ (see \refx[s]{appendix_c} for the numerical results). We have the following empirical observations:\footnote{Here we use $\sim$ to denote asymptotic scaling with respect to a small/large parameter, and $\propto$ as a standard notion of proportionality (i.e., difference by a constant factor).} \\
\begin{enumerate}
    \item For $0< h/N<1/2$, we find that $\log s_1 \sim -\frac{\kappa}{\Delta}$ for $N\gg 1$.\\

    \item For $h/N = 1/2$, we find that $\log s_1 \sim -\frac{\kappa}{4\Delta}\log N$ for $N\gg 1$.
\end{enumerate}
\vspace{0.25cm}

For generic codes satisfying $h/N <1/2$, the Pauli-noise bound in Theorem~\ref{upper_bound_sl1} yields an error rate lower than that of the general-noise bound in Theorem~\ref{lem:all_time1}. In the special case when $h/N = 1/2$, Theorem~\ref{upper_bound_sl1} predicts a memory lifetime that increases as $N^{\kappa/4\Delta-1}$ when $\kappa>4\Delta$. This case applies, for example, to the classical repetition code subject to the single-qubit bit-flip noise or the surface code subject to only qubit-erasure noise~ \cite{stace:2009} or only Pauli-$Y$ noise~ \cite{tuckett:2019}. While the bound in Theorem 1 fails to capture it, Theorem 2 predicts an unbounded lifetime as $N\to\infty$ when $\kappa>4\Delta$. 

\review{The result of Theorem~\ref{upper_bound_sl1} also holds when the jump operators are not Pauli operators but \textit{Pauli-type}, i.e., they only satisfy $E^2_\mu = I$ and $E_\mu E_{\mu'} = \pm E_{\mu'}E_\mu$ for all $\mu,\mu^{\prime}$, but not necessarily $E_\mu\in P$. For example, we can consider a noise model where at each qubit the error is described by the same single Pauli jump operator, i.e.~for a site $i$ the error operator is $E_i = c_xX_i+c_yY_i+c_zZ_i$ where $c_x,c_y,c_z\in\mathbf{R}$ and $c_x^2+c_y^2+c_z^2=1$.}

\review{\subsection{Upper bound tight for early times}}
\label{sec:early_time}

Applying the bound we found in Eq.~\eqref{eq:late_time_bbound_sl} for early times indicates that $\epsilon(t)= O(t)$. However, numerical simulations shown in Fig.~\ref{fig:5q_and_toric_code} suggest that linear scaling is only relevant at late times. Below, we establish a complementary bound $\epsilon(t)= O(t^{\ell+1})$ that confirms the slower-than-linear growth at early times.
This bound also applies to the trace distance $\delta(t)$.\\

\begin{thm} \label{upper_bound_early_time}
\review{Assuming global decoder in Eq.~\eqref{eq:noiseless_lind} and any Poissonian noise model} in Eq.~\eqref{eq:poisson_form}, the logical error rate is bounded as
\be\label{eq:early-bound}
\epsilon(t),\delta(t) \leq  \frac{1}{(1+\kappa/N \Delta)^{\ell+1}}F_\ell\Bigl((\kappa +N \Delta)t\Bigl),
\ee
where $F_\ell(x)$ is defined in Theorem \ref{generic_bound}.
\end{thm}

The proof of this result can be found in \refx[s]{appendix_generic_boundarly_time_diagrams}.
It follows the same steps as the proof of Theorem~\ref{generic_bound}, with a slightly different resummation procedure made possible by the Poissonian error assumption. 
Since $\|\mathcal L_E\|_{1\to1}\leq N$ for $\mathcal L_E$ that consists of $N$ independent processes with unitary jumps, this yields a bound that improves on Theorem~\ref{generic_bound} by a factor at most $\bigl((\chi+1)(1+N\Delta/\kappa)\bigl)^{\ell}$ for the case of a Poissonian error model.

The result from Theorem \ref{upper_bound_early_time} provides an accurate scaling for the logical error at early times, while also capturing the error rate at later times. In the limit $x\to0$, we observe the scaling $F_\ell(x)\sim x^{\ell+1}$. On the other hand, when $x\gg1$, this function behaves as $F_\ell(x) \sim x$. This means that the error rate at early times grows as $\epsilon(t) \sim  (tN\Delta)^{\ell+1}$ and saturates to a linear rate.  In the late-time regime, the logical error satisfies
\be
\begin{split}
\epsilon(t) &\leq \left(\frac{N \Delta}{\kappa+N \Delta}\right)^{\ell+1}(\kappa+N\Delta) t = \frac{N \Delta t}{(1+\kappa/N \Delta)^{\ell}}.
\end{split}\label{eq:short_t_loose}
\ee
Thus, this bound is consistent with the Taylor expansion of the bound in Eq.~\eqref{eq:upper_bound_1} if the logical error rate is small. It is also illustrated in Fig.~\ref{fig:fig1} along with the bounds we derived previously.

\section{Lower bound: stabilizer global decoder with Pauli noise}\label{sec:lower}

\begin{figure*}[t!]
\centering
\includegraphics[width=1\textwidth]{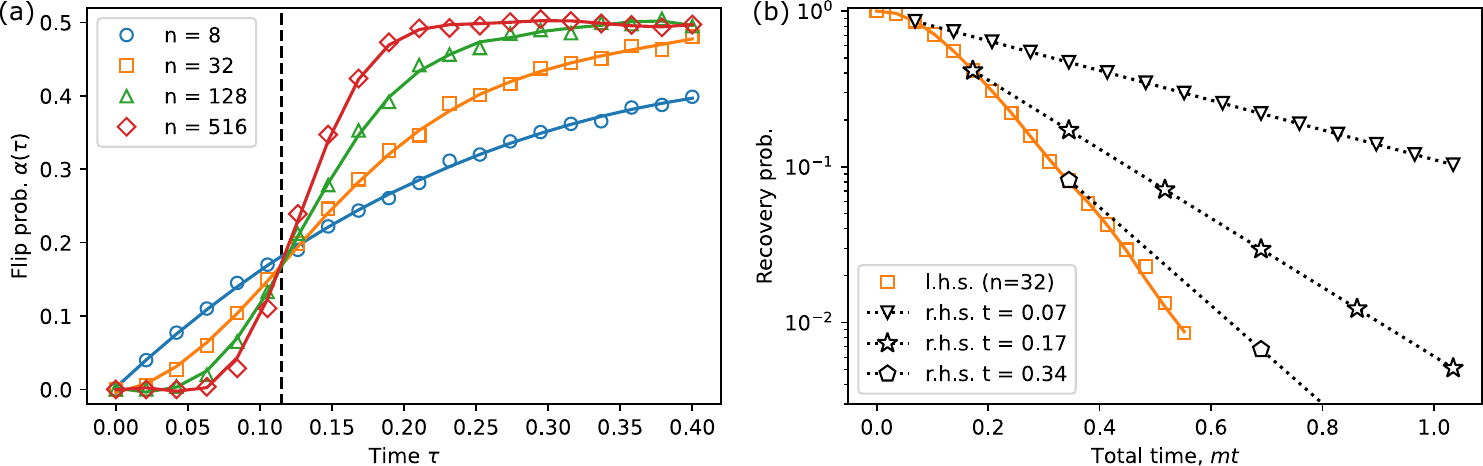}
\caption{ \textbf{Logical error after recovery for 2D toric code.} Here, we put $\Delta=1$ and consider only bit-flip errors $E_\mu \equiv X_\mu$, where $\mu$ enumerates the physical qubits. We utilize the recovery map $\mathcal R$, based on the minimum-weight-matching 
algorithm (see Section~\ref{sec:examples} for details). (a) The flip probability $\alpha_n(\tau)$ in Eq.~\eqref{eq:flip_prob} in the absence of recovery ($\kappa = 0$) for a different number of spins on a square lattice. Dots represent numerical data for the $L\times L$ lattice (the total number of qubits $n = 2L^2$ given in the legend), lines are smooth interpolations. For times larger than $\tau_c \approx 0.115$, marked by a dashed line, the logical flip probability is always nonzero, approaching the value of $0.5$ for large codes. (b) Comparison of the recovery probabilities for $n=32$ qubits \review{and $\kappa = 0$} in two cases: (i) a single recovery at the end, as given by the left-hand side of Eq.~\eqref{eq:inequality_cond} (orange dots), and (ii) repeated recoveries, as given by the right-hand side, for different times $t$ (dotted curves) and different numbers of recoveries $m$ (white triangles, stars, and pentagons). For all parameters, repeated recoveries in case (ii) perform better than applying the final recovery only in case (i). \review{We also numerically verified, but not included here, Assumption 2 for cases when $\kappa\neq 0$.}}
\label{fig:tc_check}
\end{figure*}

The above upper bounds on the logical error rate scale no faster than $N\exp(-c\kappa/\Delta)$ in the recovery rate $\kappa$. \review{However, this does not rule out the possibility that global decoders may be able to suppress errors more efficiently. Using the example of \eczoohref[qubit stabilizer codes]{qubit_stabilizer} \cite{PhysRevLett.78.405,gottes1997} subject to Pauli errors, we show below that this is generally not the case. }
 We derive a \textit{lower bound} on the error rate that decreases exponentially with the recovery-to-noise ratio $\kappa/\Delta$, and is independent of the number of qubits $n$. 
In other words, it is impossible to reduce the logical error rate to zero in the $n\to\infty$ limit while maintaining a constant recovery rate. This result shows that the upper bounds we derive for general and global decoders are tight: to improve them, we need to add conditions that would exclude global decoders, or at least exclude the assumptions we use below.

\review{To prove the lower bound, we make the following two additional natural assumptions: } \\

\begin{itemize}
 \item 
 
   \review{\textbf{Assumption 1}. \textit{In the absence of recovery, the noise process $\mathcal L_E$ generates a nonzero probability of a logical flip, i.e.
   \begin{equation}\label{eq:flip_prob}
        \alpha_n(\tau):=\Tr\Bigl(\ket{ 1}\bra{ 1}\mathcal{R}e^{\hL_{E}\tau}\ket{ 0}\bra{ 0}\Bigl)>a,
    \end{equation}
\review{where $\tau> \tau_c \propto 1/\Delta$ and $a\in(0,1)$ are size independent parameters, $\mathcal L_{E}$ is the noise generator, $|w\>$ are the logical states. }}}\\

\review{This assumption ensures that the noise model is sufficiently powerful.  The threshold time $\tau_c$ determines the time after which the logical information is no longer perfectly recoverable. In Fig.~\ref{fig:tc_check}(a) we illustrate this property by plotting $\alpha_n(\tau)$ for different system sizes of the 2D toric code (see Section~\ref{sec:examples} for a definition). In this plot, one can clearly observe the threshold time $\tau_c$ after which the logical state appears highly mixed after the recovery map for any code size.}\\

    \item \review{\textbf{Assumption 2}. The recovery map $\mathcal R$ and total Lindblad operator $\mathcal{L}=\mathcal L_R+\Delta\mathcal L_E$ satisfy
    \begin{equation}\label{eq:inequality_cond}
    \begin{split}
        \Tr[\ket{ 0}\bra{ 0}\mathcal{R}&e^{\hL mt}\ket{ 0}\bra{ 0}]\\
        &\leq  \Tr[\ket{ 0}\bra{ 0}\left( \mathcal{R} e^{\hL t}\right)^m\ket{ 0}\bra{ 0}]
    \end{split}
    \end{equation}
for any time $t\geq 0$ and integer $m\geq0$.} \\

This assumption states that interleaving the noisy evolution with more recovery operations is more effective than not doing it. We illustrate this property in Fig.~\ref{fig:tc_check}(b) for the 2D toric code. Specifically, we plot both the right-hand side and the left-hand side of this inequality for different times $t$ and integers $m$. Notably, this inequality holds even at times $t>\tau_c$, when the average probability of errors for each physical qubit exceeds the code threshold. While we do not generally expect good performance from the decoder in this case, the toric code demonstrates some improvement even above the threshold.
\end{itemize}

\vspace{0.25cm}

Given these assumptions, we can show the existence of the minimal logical error rate. 
\review{
\begin{thm} \label{them:lower_bound} Consider an $n$-qubit quantum stabilized code, the noise process $\mathcal L_E$ in Eq.~\eqref{eq:noisy_L} with rate $\Delta$ where $\{E_\mu\}$ represent Pauli noise according to Definition~\ref{dfn:pauli_noise}, and the recovery process $\mathcal L_R$ in Eq.~\eqref{eq:noiseless_lind} generated by the recovery map $\mathcal R$ with rate $\kappa\geq 0$, which together satisfy Assumptions 1 and 2. Then the logical error at $t\geq 0$ satisfies
\be
\begin{split}
&\epsilon(t) \geq \frac{1}{2}\Bigl(1-\exp(-\Delta_{\rm eff}t)\Bigl),\\
&\Delta_{\rm eff} = -\frac{\kappa \log(1-ae^{-\kappa\tau_c})}{\kappa \tau_c+\log 2}.
\end{split}
\ee
\end{thm}
A proof for the bound is given in Appendix. Note that $\tau_c\propto 1/\Delta$ by Assumption 1. In the limit $\Delta \ll \kappa$, we have
\begin{equation}
\Delta_{\rm eff} \simeq \frac{a}{\tau_c}e^{-\kappa\tau_c}.
\end{equation}
Suppose $\tau_c=c/\Delta$, for some constant $c>0$, the effective error rate is $\Delta_{\text{eff}}=O(\Delta e^{-c\kappa/\Delta})$ as $\Delta\to 0$. This result shows that the lower bound on the logical error rate decreases exponentially with the ratio of the recovery rate to the error rate. However, it also follows that under the assumptions in Eqs.\ (\ref{eq:flip_prob},\ref{eq:inequality_cond}) (which we have verified for the 2D toric code, but expect to hold more generally), it is impossible to obtain a quantum memory with either exponential or polynomial lifetime using only a constant recovery rate. Some careful readers may notice that Theorem~\ref{upper_bound_sl1} suggests that the repetition code subject to a single type of Pauli noise has a memory lifetime that grows polynomially with system size. This is not inconsistent with the lower bound. One can verify numerically that assumption 1 is violated by the repetition code, i.e., the timescale $\tau_c$ for logical information corruption under purely noisy dynamics grows with system size. }

\begin{figure*}[t!]
    \centering
    \includegraphics[width=1\textwidth]{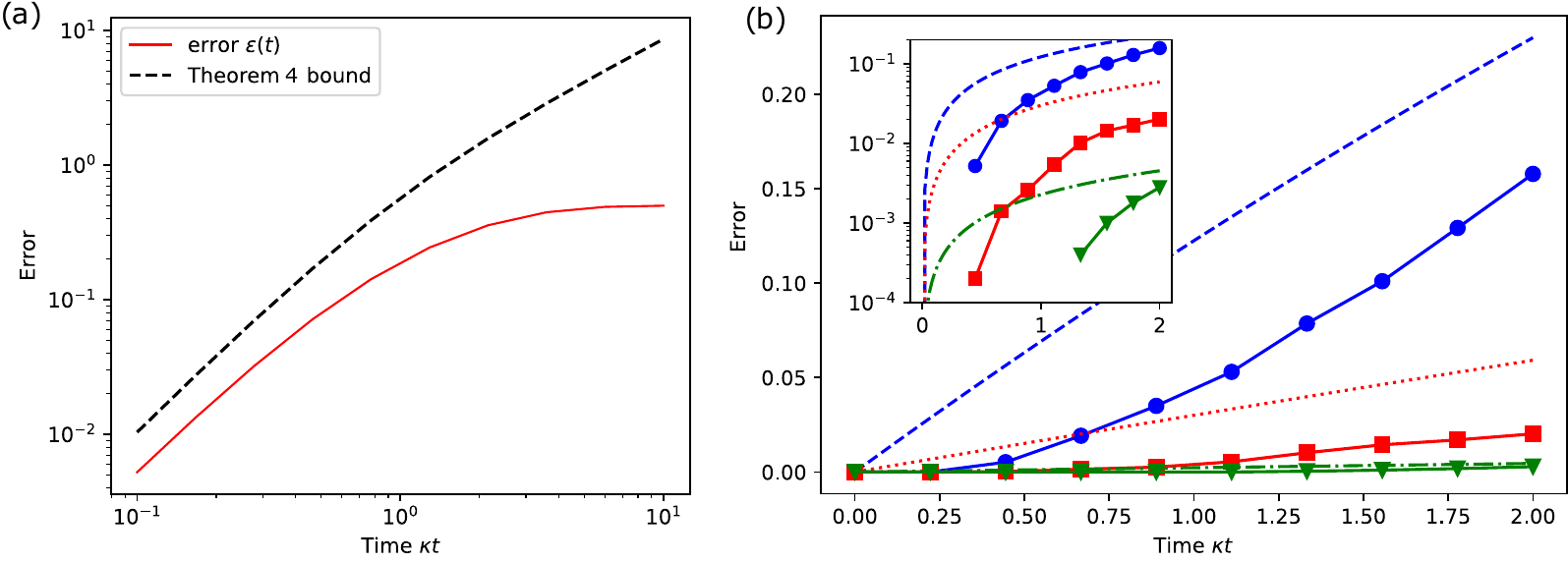}
    \caption{\textbf{Stabilizer codes.} (a) Logical error measure $\epsilon(t)$ for the five-qubit code and the corresponding upper bound in Eq.~\eqref{eq:early-bound}. (b) The same logical error measure for the two-dimensional toric code on a $L\times L$ lattice of linear sizes $L = 4$ (blue), $6$ (red), and $8$ (green), assuming the minimum-weight-matching algorithm has the threshold $h \approx 0.1031 \, n$  \cite{wang2003confinement}. The recovery rate grows linearly with the number of qubits, $\kappa = 0.1\, n$. The inset shows the same plot in logarithmic scale on the y-axis.}
    \label{fig:5q_and_toric_code}
\end{figure*}

\begin{figure*}[t!]
    \centering
    \includegraphics[width=1\textwidth]{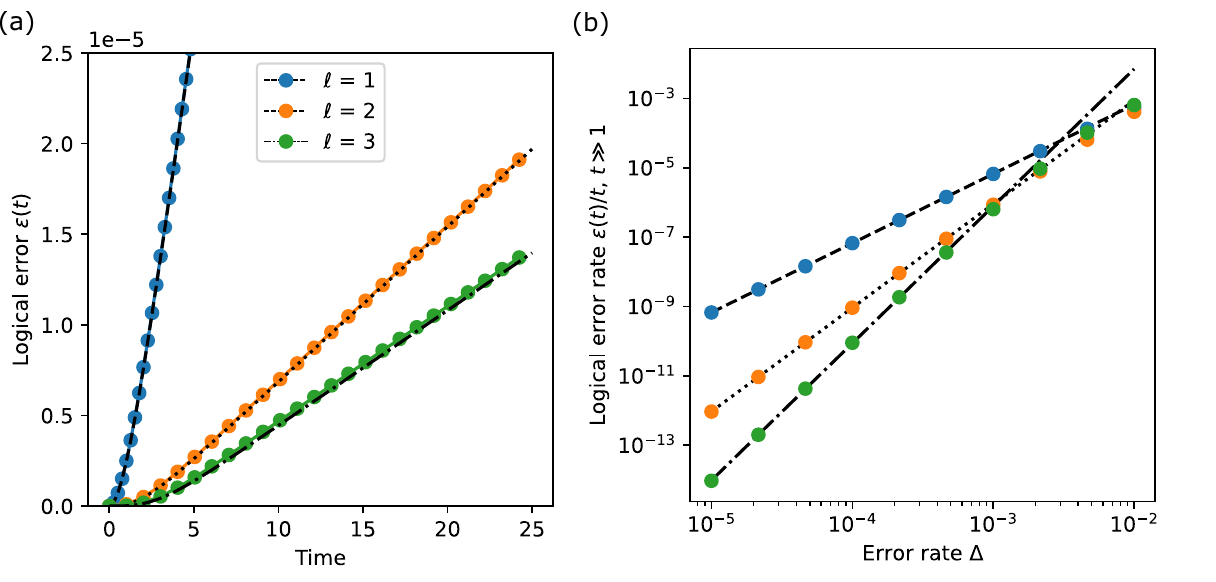}
    \caption{\textbf{Binomial code.} (a) The error probability $\epsilon(t)$ in Eq.~\eqref{eq:overlap_log_error} for the binomial code in Eq.~\eqref{eq:binomial_code} with different values of $\ell=1,2,3$ using the parameters $\kappa=1$, $\Delta = 10^{-3}$. The dots show numerical results (obtained for finite-size Hilbert space approximation), the curves show the fit of the form $\epsilon(t) = (c_\ell\Delta)^{\ell+1} F(\kappa t)$ (compare to the result in Theorem~\ref{generic_bound}), where $c_1 = 2.57$, $c_2 = 9.50$, and $c_3 = 28.2$. (b) The dots show the saturated error rate for different values of $\ell$ as a function of $\Delta$. The curves show the asymptotics of the form $\epsilon(t) = (c'_\ell\Delta)^{\ell+1}t$, where $c'_1 = 2.57$, $c'_2 = 9.51$, and $c'_3 = 28.54$. 
    This agrees with the theoretical bound in Eq.~\eqref{eq:general_theorem} when taking into 
    account that $\|\mathcal{L}_E\|_{1\to1,\mathsf{E}} \geq C_\ell > c_\ell, c'_\ell$,
    where $C_\ell := \sqrt{
    {\rm Tr} 
    \left( |0\rangle \langle 0| \mathcal{L}_E^\dag \mathcal{L}_E |0\rangle \langle 0| \right)
    }$,
    $|0\rangle$ is an $\ell$-dependent codeword in Eq.~\eqref{eq:binomial_code}, and $C_1 \approx 6.50$, $C_2 \approx 24.5$, and $C_3 \approx 61.0$.
    }
    \label{fig:fig3}
\end{figure*}

 \review{\section{Asymptotic error estimate: global decoders with Poissonian noise}}
\label{sec:asymptotic}

\review{The derivations in the previous sections established several general upper and lower bounds for the performance of global decoders combined with Pauli noise. This naturally leads to the question: Is it possible to distill the essence of these derivations into a single approximate formula describing the logical error of global decoders? This section aims to provide the reader with such a simplified expression, valid for times much longer than the characteristic recovery time. We use the Poisson picture of many ``faithful'' trajectories and estimate the number of trajectories that avoid logical errors. From this estimate we derive a practical expression for the rate of logical errors.}

\review{First, we specify such trajectories using the following definition:}

\begin{dfn}\label{def_unrecoverable}
We say that a trajectory $\boldsymbol \mu =(\mu_1,\dots,\mu_q)$ of weight $q$ is faithful if it contains no error subsequence of length $m>\ell$, i.e.~a subsequence $\{\mu_{k},\dots, \mu_{k+m-1}\}$ \review{satisfying $\mu_{k+i}>0$ for all $1\leq k\leq q-m+1$}.
\end{dfn}
 In other words, faithful trajectories contain no uninterrupted error sequences of length greater than $\ell$. 
 The total probability of faithful trajectories provides a lower bound on the probability that no logical error occurs (see the lemma below). 
 
 In passing, it is important to note that in general not all trajectories that are not faithful contribute to the logical error. For example, a sequence of more than $\ell$ Pauli errors does not contribute to a logical error if it can be reduced to a weight less than $\ell$ by canceling identical errors. However, since we are only aiming for an \review{the bound on} the logical error, one can exclude these trajectories from consideration and still keep the bound valid.

Let the full set of faithful trajectories be denoted by $G$. Then the logical error probability is bounded as:

\begin{lem}\label{stoch_bound}
For a Poissonian error process in Eq.~\eqref{eq:poisson_form}, the error in Eq.~\eqref{eq:overlap_log_error} satisfies
\be\label{eq:poisson_lemma}
\epsilon(t)\leq p(t) := \sum_{\boldsymbol \mu\in F\setminus G}p(\boldsymbol \mu,t)~,
\ee
where $F\setminus G$ stands for trajectories that are not faithful.
\end{lem}

The proof of this lemma can be found in \refx[s]{proof_operator_bound}. \review{This lemma allows us to derive the upper bound of the logical error Eq.~\eqref{eq:poisson_lemma} by the fraction of faithful trajectories. For example, for small codes, one could evaluate $p(t)$ numerically by sampling the trajectories and counting the faithful ones directly. For example, the result for a code with $\ell = 6$ is shown in Fig.~\ref{fig:fig1}. More generally, an approximate analytical formula for this result is given below. }

We exploit the fact that errors and recoveries are independent stochastic processes. Then, for any trajectory of time $t$, the probability of $m$ recoveries is equal to $R(m,\kappa t)$, where \review{$R(m,x) := x^m e^{-x}/m!$} (see Eq.~\eqref{eq:stoch_prob}). \review{Let us condition on the event that $m$ corrections occur during the time period $t$. The probability that a logical error occurs between the start of the dynamics and the first correction is given by
\begin{equation}\label{eqs:p1_expression_seq}
    p_1 = \int_{0}^{t} dt_1 \pi(t_1 | t) s(t_1),
\end{equation}
where $\pi(t_1 | t)$ is the conditional probability that the time of the first recovery event in a random sequence is $t_1$, conditioned on the total evolution time $t$, and $s(t_1)$ is the probability that the error process during the interval $[0, t_1]$ does not result in an error weight greater than $\ell$. These probabilities are defined as
\begin{equation}
\begin{split}
    &\pi(t_1 | t) := \frac{\kappa \exp(-\kappa t_1)}{1 - \exp(-\kappa t)}, \\
    &s(t_1) := \sum_{m=0}^{\ell} R(m, N\Delta t_1).
\end{split}
\end{equation}
Similarly, the probability of a logical error occurring during the gap between the $(i-1)$-th and $i$-th corrections, given that previous corrections occurred at times $t_1, \dots, t_{i-1}$, is
\begin{equation}
    p_i = \int_{0}^{t - \sum_{j=1}^{i-1} t_j} dt_i \pi(t_i | t) s(t_i),
\end{equation}
Finally, the probability that a faithfulness violation will occur during the time between the last correction and the end of the dynamics is simply
\begin{equation}
    p_{m+1} = s\left(t - \sum_{j=1}^{m} t_j\right).
\end{equation}
Since the events described above are independent, the total probability of obtaining a non-faithful trajectory is then given by
\begin{widetext}
\be\label{exact_pt}
\begin{split}
p(t) = 1 - \sum_{m=0}^{\infty} R(m, \kappa t) \int_{0}^{t} dt_1 \pi(t_1 | t) s(t_1) \int_{0}^{t - t_1} dt_2 \pi(t_2 | t - t_1) s(t_2) \\
\ldots \int_{0}^{t - \sum_{j=1}^{m-1} t_j} dt_m \pi\left(t_m \Biggl| t - \sum_{j=1}^{m} t_j\right) s(t_m) s\left(t - \sum_{j=1}^{m} t_j\right),
\end{split}
\ee
\end{widetext}
In the asymptotic limit, where $\kappa t \gg 1$, a heuristic estimate can be made for the probability outlined in Eq.~\eqref{exact_pt}. First, we extend the upper limits of all integrals to infinity. This extension is justified because the function $\pi(t_i|\tau)$ vanishes quickly for $\kappa t_i \gg 1$, making the effect of the finite upper limit of the integration effectively negligible. For similar reasons, it is permissible to set the conditioned times $\tau\to\infty$ in $\pi(t_i|\tau)$, since this function tends to saturate for large values of $\kappa\tau$. Finally, we ignore the relatively small contribution of $p_{m+1}$ compared to the cumulative contribution of the previous terms. This approach results in the expression
\be\label{eq:asymptotic dynamics}
\begin{split}
p(t) &\approx 1- \sum_{m=0}^\infty R(m,t) \Biggl(\int_0^\infty d\tau \pi(\tau|\infty) s(\tau)\Biggl)^m \\
&= 1-\exp(-\Delta_{\rm eff} t),
\end{split}
\ee
where the effective logical error rate is
\be\label{eq:effective_rate}
\Delta_{\rm eff} = \frac{\kappa}{\left(1+\frac{\kappa}{N \Delta}\right)^{\ell+1}}.
\ee}

Let us examine the above expression from the standpoint of memory lifetime for multi-qubit codes.
Assuming that each qubit is subject to at least one type of error, the number of elementary error processes grows with the number of qubits $n$, i.e.~$N = \Theta(n)$.\footnote{Here and below we use the family of ``big-O" Bachmann-Landau notations: $o(f(x))$ (dominated by $f(x)$), $O(f(x))$ (bounded from above by $f(x)$), $\Theta(f(x))$ (bounded from below and above by $f(x)$), $\Omega(f(x))$ (bounded from below by $f(x)$), and $\omega(f(x))$ (dominates $f(x)$). More about definitions can be found in Ref.~ \cite{cormen2001introduction}.} 

Using the asymptotic behavior in Eq.~\eqref{eq:effective_rate}, we then get the scaling (for a constant $\Delta$):
\be\label{eq:scaling_radius}
\Delta_{\rm eff} =
\begin{cases}
\kappa \exp(- \Theta(\kappa \ell/n\Delta)) \qquad &\kappa = o(n),\\
\kappa (\frac{\kappa}{n\Delta})^{- \Theta(\ell)} \qquad &\kappa = \Omega(n),
\end{cases}
\ee
where all parameters, i.e.~$\kappa = \kappa(n)$, $\ell = \ell(n)$, must be treated as functions of the number of qubits $n$.
For the case of the two-dimensional \eczoohref[toric code]{surface}  \cite{kitaev2006anyons}, the radius satisfies $\ell= \Theta(\sqrt{n})$, so the effective rate is 
suppressed in number of qubits if the recovery rate scaling satisfies $\kappa=\omega(\sqrt{n})$.

The result
we presented above can be improved for certain families of codes and recovery maps defined on $n$ qubits. 
\review{In this case we consider that the code has a tolerable error weight $h = h(n)$, see Definition~1. This means that the subsequence of errors of length $|\boldsymbol \mu|\leq h$ is at most $2^{-\Omega(d)}$, where $d$ is the code distance.}
 This allows us to improve the upper bound in Eq.~\eqref{eq:scaling_radius} by replacing $\ell$ with $h$
 and taking into account the exponentially small contribution in $d$ coming from trajectories \review{that contain logical errors},
\be
s(t) = \sum_{m=0}^h R(m,N\Delta t) -2^{-\Omega(d)}.
\ee
This improvement grants us a stronger estimate for the logical error in the form
\be\label{eq:effective_rate_threshold}
\Delta_{\rm eff} = \frac{\kappa}{\left(1+\frac{\kappa}{N \Delta}\right)^{h+1}}+2^{-\Omega(d)}.
\ee

In a similar fashion to the analysis in Eq.~\eqref{eq:scaling_radius}, if a code family has a tolerable weight $h>\ell$, we can use Eq.~\eqref{eq:effective_rate_threshold} to derive the bound
\be
\Delta_{\rm eff} = \kappa \exp\Bigl(- \Theta(\kappa h/n\Delta)\Bigl)+2^{-\Omega(d)}, \quad \kappa = o(n).
\ee
For example, the toric code has a threshold error weight $h=\Theta(n)$, which leads us to 
the error rate $\Delta_{\rm eff} \sim \exp(-\Theta(\kappa/\Delta))$. 
Thus, a constant $\kappa$ yields a constant memory time,
while $\kappa\sim \log(n)$ yields a polynomial memory time.

These results are non-perturbative in $\Delta$, as perturbation theory would predict that the logical error rate scales as 
order $ O(\Delta^w)$ for some $w\geq 1$. 
We demonstrate this difference between our treatment and perturbation theory using the example of the toric code in Section~\ref{sec:examples}.

The difference in performance between autonomous global decoders and conventional active quantum error correction, which relies on syndrome measurement and subsequent correction  \cite{terhal2015quantum}, may seem unexpected. In the traditional method, it is usually sufficient to maintain a constant ratio between the single-qubit error rate $\Delta$ and a fixed inverse time scale $T^{-1}$ between recoveries to achieve exponential lifetime. 
Indeed, if $n T \Delta <h$, the probability of accumulating more than $h$ errors becomes exponentially small in $n$. 
The main difference between the two cases, stemming from our interpretation of autonomous recovery as a stochastic process, is that the time $T$ between two consecutive recoveries is \textit{not} fixed for the autonomous case and is instead determined by the Poisson distribution. Due to this fact, even for large $\kappa$, the probability that $n \Delta T>h$ is constant as a function of system size, although it is exponentially small in $\kappa/\Delta$. 
Therefore, increasing the system size alone does not increase the lifetime of the logical qubit.
\section{Examples of autonomous codes}
\label{sec:examples}

Finally, we provide a few examples of autonomous codes generated from global decoders of existing quantum codes. We start with qubit stabilizer codes. Our \textit{global} recoveries are different from local decoders  \cite{pastawski:2011,dengis2014optimal,Ticozzi2019, liu2023dissipative} in that jump operators of the latter apply recovery steps only on geometrically restricted regions.

The codespace $\mathsf C$ of an $[[n,1,d]]$ stabilizer code is formed by the $+1$ eigenstates of $n-1$ mutually commuting Pauli operators $S_\alpha$ that satisfy $S_\alpha^2 = I$, $[S_\alpha,S_\beta] = 0$ for all $\alpha, \beta$. 
The traditional recovery map includes two steps. In the first step, we measure all stabilizer generators $S_\alpha$, projecting the state into a subspace of mutual eigenstates with corresponding eigenvalues $s_\alpha = \pm 1$. 
This procedure is equivalent to applying a projection operator
\be
P(\vec s) = \prod_{\alpha=1}^{n-1} \frac 12\Bigl(1+s_\alpha S_\alpha\Bigl).
\ee
Next, we apply the corresponding recovery unitary $C(\boldsymbol s)$, which is a product of individual Pauli operators, depending on the $(n-1)$-dimensional vector of outcomes $\boldsymbol s = \{s_\alpha\}$. We can make a decision on the recovery using an algorithm or simply a lookup table that pairs every stabilizer configuration with its corresponding recovery.

In the autonomous regime, we propose to implement these recoveries using the continuous process, which combines both procedures:
\be
\mathcal R(\rho) = \sum_{\vec s\in \mathbb Z_2^{n-1}} A_{\vec s}\rho A^\dag_{\vec s},
\ee
where the jump operators are defined as
\be
\begin{split}
A_{\vec s} &= C({\vec s})P({\vec s}) = P C({\vec s}).
\end{split}
\ee
Here $P$ is the projector onto the codespace, and the last equality follows from the fact that $C({\vec s})$ commutes (anticommutes) with the stabilizer $S_\alpha$ if $s_\alpha = 1$ ($-1$):
\be
\begin{split}
C({\vec s})\prod_{\alpha=1}^{n-1} \frac 12\Bigl(1+s_\alpha S_\alpha\Bigl) &= \prod_{\alpha=1}^{n-1} \frac 12\Bigl(1+ S_\alpha\Bigl)C({\vec s}) \\
&=  P C({\vec s}).
\end{split}
\ee
For these stabilizer recovery models, we use simplistic Pauli error operators $E_\mu \in \{X_i,Y_i,Z_i\}$ in Eq.~\eqref{eq:noisy_L}, which act with the same rate $\Delta$ on each qubit $i$ (i.e., we set all $\lambda_\mu=1$).

In the simplest example, we consider an autonomous stabilizer decoder based on a five-qubit code. This distance $d=3$ code protects one logical qubit using five physical qubits and is stabilized by four operators $S_\alpha \in \{XZZXI, IXZZX, XIXZZ, ZXIXZ\}$, where $I$, $X$, and $Z$ are respectively the identity and the $X$- and $Z$-Pauli operators acting on the corresponding qubit of the system. We illustrate the performance of this code in Fig.~\ref{fig:5q_and_toric_code}(a). As can be seen from the figure, the upper bound in Eq.~\eqref{eq:early-bound} accurately describes the error rate in such a model.

Another relevant example is the two-dimensional toric code. 
This code is defined on a two-dimensional square lattice with $L\times L$ plaquettes and periodic boundary conditions, where physical qubits are situated on the edges. The stabilizers are divided into two groups. One group includes all products of four $Z$ operators acting on edges $s$ adjacent to a vertex (``stars"), which we denote as $A_s = \prod_{i\in s} Z_i$. The other group consists of all products of $X$ operators acting around a square $p$ (``plaquettes"), which we denote as $B_p$.
The codespace consists of the ground states of the operator
\begin{equation}
H = -\sum_s A_s - \sum_p B_p. \label{HTC}
\end{equation}
Using measurements of each group of stabilizers separately, it is possible to independently correct errors in the $X$ and $Z$ bases even if both of them are present in the system. To construct the recovery operator $\mathcal R$, we use the minimum-weight-matching 
algorithm  \cite{wang2003confinement}, which suggests the recovery unitary $C(\vec s)$ for each vector of measurement outcomes $\vec s$.
We compare the rate of logical error with the prediction given by the upper bound from Theorem~\ref{lem:all_time1}. In particular, Fig.~\ref{fig:5q_and_toric_code}(b) shows how the logical error depends on the linear size $L$ of the lattice. It can be seen that the upper bound correctly predicts the performance of the code.

Additionally, we compare our results with the predictions of perturbation theory for the autonomous toric code model. First, we find the solution of the spectral problem exactly for a Lindbladian with no noise ($\Delta = 0$). This solution has a $4^2$-dimensional steady-state manifold that is separated by a dissipative gap $\kappa$ from the rest of the eigenstates. The steady states are superpositions of four toric-code ground states. The rest of the eigenstates have the same eigenvalue $\kappa$. Using this exact solution, we use perturbation theory to determine how the eigenvalues of steady states are perturbed by noise. The real part of the lowest-order perturbation can be used as an estimate of the logical error rate.

Notably, as we show in \refx[s]{appendix_h}, when $\kappa$ is a system-size independent constant, the leading-order contribution from perturbation theory diverges as $L$ approaches infinity.
If the recovery rate scales with $L$ as $\kappa = \kappa_0L$ for some constant $\kappa_0>0$, the leading-order contribution from perturbation theory scales as
\be
\epsilon(t) = O\left(\kappa_0 t L^2\left(\frac{2\Delta}{e\kappa_0}\right)^{L/2}\right)
\ee
in the limit $L\to\infty$. This still provides a better estimate than that for the general recovery process in Theorem~\ref{generic_bound}, which requires $\kappa\sim L^2$ to ensure exponential suppression of the logical error rate. As a comparison, we can apply the asymptotic result obtained in Section~\ref{sec:asymptotic} to the autonomous toric code by setting $N = n=2L^2$ and $h = 2fL^2$ for some constant $f>0$ that indicates the finite threshold of the toric code. The error rate given by Eq.~\eqref{eq:effective_rate_threshold} is
\be
\epsilon(t)= O\left(\kappa_0 t Le^{-f\kappa_0L/\Delta}\right),
\ee
which suggests a non-perturbative contribution at $\Delta = 0$. Indeed, we see that, although the perturbation result does capture the exponential suppression of the error rate as $L$ approaches infinity, for small $\Delta$ it overestimates the error rate compared to the asymptotic behaviour.
This example highlights the importance of non-perturbative approaches in estimating the memory lifetime for an autonomous error-correcting code.  

Finally, we consider an example of a code that cannot be understood in terms of Pauli errors. An example of such a code is the \eczoohref[binomial code]{binomial}  \cite{michael206new} defined for the space of a quantum harmonic oscillator, $\mathsf H = \{|n\>_B, n\geq0\}$. The transitions between quantized oscillator levels are induced by the creation operator $a^\dag$ and the annihilation operator $a$ such that $a^\dag|n\>_B = \sqrt{n+1}|n+1\>_B$ and $a|n\>_B = \sqrt{n}|n-1\>_B$.  
The codewords of the binomial code of distance $d=2\ell+1$ are 
\be\label{eq:binomial_code}
|0\rangle,|1\rangle = \frac 1 {2^\ell} \sum_{s\in {\rm even, odd}}^{[0,2\ell+1]}\sqrt{\binom{2\ell+1}{s}}|s(2\ell+1)\rangle_B.
\ee
The binomial code tolerates single-photon processes as well as dephasing, with elementary errors generated from the set $E_\mu \in \{a,a^\dag, a^\dag a\}$, where weights $\lambda_\mu = 1/3$ are the same for each channel type in Eq.~\eqref{eq:noisy_L}. The recovery map $\mathcal{R}$ is defined using the procedure in \refx[s]{appendix_error_measures}. Fig.~\ref{fig:fig3} shows that the logical error rate of the code decreases exponentially with the code radius, for different values of $\ell$.

\section{Summary and outlook}\label{sec:discussion}

We derived the universal dependence of the logical error of a global quantum decoder on error model parameters. Under general assumptions, we found that global decoders provide viable error suppression. We also developed criteria under which the lifetime of the memory can be extended indefinitely by increasing the system size. To achieve this, decoders must operate at a rate that grows with system size. While this growth can be mild---polynomial suppression can be achieved with logarithmic rates---it shows that a constant dissipative gap of the recovery procedure is not sufficient to ensure a quantum memory whose lifetime grows indefinitely with system size. It also means, contrary to what one might naively imagine, that autonomous decoders cannot be perceived or constructed as stochastic versions of traditional error correction protocols. In fact, the structure of the correction map (e.g., represented by $\mathcal K_t$ in Eq.~\eqref{eq:full_model}) plays an important role in existing autonomous decoders  \cite{pastawski:2011}.

Another motivation for studying non-local dissipative processes is to see if they exhibit threshold-like behavior. While we do not generally observe sharp features like this in our analytical analysis of Poissonian models, it is possible that a transition could occur for a more general type of noise model that becomes weaker as system size increases.

In the future, we could use similar techniques to study local decoders. For example, we could try to prove analytically that there is a threshold for autonomous models based on existing cellular-automata decoders such as sweep-rule decoders  \cite{kubica2019cellular} or local decoders motivated by the thermalization of physical Hamiltonians~ \cite{liu2023dissipative}. \review{Another interesting direction is to improve the upper bound derived in Theorem~\ref{generic_bound} by adding conditions that exclude the case of global decoders. This would allow us to avoid the limitations imposed by the lower bound in Theorem~\ref{them:lower_bound}. Finally, we hope that the techniques can yield even stronger results for more concrete models such as bosonic codes. Using the structure of the codes, e.g.\ thresholds of sweep decoders  \cite{kubica2019cellular}, one can derive stronger bounds that will demonstrate the ultimate scalability of autonomous quantum memories.}

\section*{Acknowledgments}
Y.-J.L.~acknowledges support from the Max Planck Gesellschaft (MPG) through the International Max Planck Research School for Quantum Science and Technology (IMPRS-QST) and the Munich Quantum Valley, which is supported by the Bavarian state government with funds from the Hightech Agenda Bayern Plus.
A.V.G.~was supported in part by 
NSF QLCI (award No.~OMA-2120757), 
NQVL:QSTD:Pilot:FTL, 
DoE ASCR Quantum Testbed Pathfinder program (awards No.~DE-SC0019040 and No.~DE-SC0024220), 
NSF STAQ program, 
AFOSR, 
ARO MURI, 
AFOSR MURI, 
DARPA SAVaNT ADVENT, and 
ARL (W911NF-24-2-0107). 
A.V.G.~also acknowledges support from the U.S.~Department of Energy, Office of Science, 
National Quantum Information Science Research Centers, 
Quantum Systems Accelerator (QSA), 
and 
from the U.S.~Department of Energy, Office of Science, Accelerated Research in Quantum Computing, Fundamental Algorithmic Research toward Quantum Utility (FAR-Qu).
V.V.A.~acknowledges support from NSF QLCI grant OMA-2120757.
V.V.A.~thanks Olga Albert and Ryhor Kandratsenia for providing daycare support throughout this work.

\onecolumngrid

\section*{Appendices}

\subsection*{Appendix A: Recovery map $\mathcal R$ and error measures}
\labelx[Appendix A]{appendix_error_measures}

\review{Let us explicitly construct the recovery map and proving some of the code properties.}
To do so, we first consider the error operators $F_\alpha = \sum_\bsnu u^*_{\alpha \bsnu}K_\bsnu$, where $u_{\alpha\bsnu}$ are matrix elements of the unitary $u$ that diagonalizes the matrix $C$
in Eq.~\eqref{eq:knill_laflamme}, i.e. $C = u^\dag\hat C u$, where $\hat C = {\rm diag}\{d_\alpha\}$ \review{and $d_\alpha$
are eigenvalues of $C$.} The action of these operators is orthogonal in the codespace, i.e.
\be\label{eq:FF_knill_laflamme}
P F^\dag_\alpha F_\beta P = d_\alpha \delta_{\alpha\beta} P,
\ee
where $P$ is the projector to the codespace $\mathsf C$. Next, we can use the polar decomposition
\be
F_\alpha P = U_\alpha \sqrt{PF^\dag_\alpha F_\alpha P} =\sqrt{d_\alpha}U_\alpha P,
\ee
where $U_\alpha$ are unitary operators.

\review{Next, our goal is to write down the explicit expression for the recovery operator $\mathcal R$ and connect it to the properties of the error operators it guarantees to correct. To do so, we follow the steps of the standard textbook procedure, e.g. see Ref.~ \cite{nielsen_quantum_2010} chapter 10.3. We start with the most general form
\be\label{eqs:R_gen_form}
\mathcal R(\rho) = \sum_{\alpha=1}^{D^2-1} R_\alpha \rho R^\dag_\alpha,
\ee
where $R_\alpha$ are Kraus operators satisfying $\sum_\alpha R^\dag_\alpha R_\alpha = I$ and $D$ is the dimension of the Hilbert space. The condition in Eq.~\eqref{eq:assumptions} implies that for any $\beta$, $\beta'$ we have
\be
\mathcal R(F_{\beta}\rho F^\dag_{\beta'}) = \sum u_{\alpha\bsmu}u^*_{\beta'\bsmu'}\mathcal R(K_{\bsmu}\rho K^\dag_{\bsmu'})\propto \rho.
\ee
At the same time, the right hand side of this expression takes the form
\be
\mathcal R(F_{\beta}\rho F^\dag_{\beta'}) = \sum_{\alpha=1}^{D^2-1} R_\alpha F_{\beta}\rho  F^\dag_{\beta'} R^\dag_\alpha = \sqrt{d_\beta d_{\beta'}}\sum_{\alpha=1}^{D^2-1} P R_\alpha U_{\beta}P\rho  PU^\dag_{\beta'} R^\dag_\alpha P\propto \rho.
\ee
This condition implies that for any $\beta$ such that $d_\beta>0$, there exist complex numbers $c_{\alpha\beta}$ such that
\be
P R_\alpha U_{\beta}P = c_{\alpha\beta} P.
\ee
For the case of degenerte spectrum $d_\beta$, there is more than once choice of $R_\alpha$ that satisfies this condition. Since all these choices generate the same map $\mathcal R$, without the loss of generality, we can choose
\be\label{eqs:def_ra_operator}
R_\alpha = PU^\dag_\alpha, \qquad \alpha = 1,\dots, N_C,
\ee
where \( R_\alpha \) is defined only for operators \( F_\alpha \) with \( d_\alpha > 0 \). Here \( N_C \) represents the total number of such $F_\alpha$. 
Then, we can construct the recovery map as
\be\label{eq:recovery_map_kraus}
\mathcal R(\rho) = \sum_{\alpha=1}^{N_C} R_\alpha \rho R^\dag_\alpha+\frac 1q \Tr(\rho P_{\perp})\sigma(\rho),
\ee}
where $\sigma(\rho) = P\sigma(\rho)P$ is certain density function that may depend on $\rho$, \review{ $q$ is the codespace dimension,} and \( P_{\perp} = I - \sum_\alpha R^\dagger_\alpha R_\alpha \) defines the projector on the subspace of ``undecidable" error states. These error states are created by errors that do not obey the Knill-Laflamme condition. 

\review{Using the relation in Eq.~\eqref{eqs:rr_property}, let us show first that $\mathcal R^2 = \mathcal R$. In order to do this, we express
\be\label{eq:fstep_nlp}
\begin{split}
\forall \rho:\quad \mathcal R^2(\rho) = &\sum_{\alpha,\beta=1}^{D_0}  R_\alpha R_\beta \rho R^\dag_\beta R^\dag_\alpha+\frac 1q \sum_{\alpha}^{D_0}\Tr(R_\alpha\rho R^\dag_\alpha P_{\perp})P \\
&+ \frac 1q\Tr(\rho P_{\perp})\sum_{\alpha=1}^{D_0}  R_\alpha P R^\dag_\alpha+\frac 12 \Tr(\rho P_{\perp})\Tr(P P_{\perp})P.
\end{split}
\ee
The first term in Eq.~\eqref{eq:fstep_nlp} can be simplified using that correction Kraus operators satisfy
\be\label{eqs:rr_property}
\begin{split}
R_\alpha R_\beta = PU^\dag_\alpha PU^\dag_\beta = \frac{1}{\sqrt{d_\alpha}}PF^\dag_\alpha PU^\dag_\beta= \delta_{\alpha 0}R_\beta,
\end{split}
\ee
where we use index 0 to enumerate the ``trivial'' error, i.e. $F_0 = I$ and $d_0 = 1$.
The last term vanishes due to the orthogonality of the spaces represented by the operators $P$ and $P_\perp$, i.e. $PP_\perp = 0$. This property follows from the transformation
\be
P_\perp P = P - \sum_\alpha R_\alpha^\dag R_\alpha P = P - \sum_\alpha U_\alpha P U^\dag_\alpha P = P - \sum_\alpha \frac1{d_\alpha} F_\alpha P F^\dag_\alpha P = P-\frac1{d_0}F_0P = 0,
\ee
In a similar way, we express the second term as
\be
\Tr(R_\alpha\rho R^\dag_\alpha P_{\perp}) = \Tr(U^\dag_\alpha\rho U_\alpha PP_{\perp}P) = 0.
\ee
Finally, the third term can be simplified using the expression
\be
\sum_{\alpha=1}^{D_0}  R_\alpha P R^\dag_\alpha = \sum_{\alpha=1}^{D_0}  PU_\alpha^\dag  P U_\alpha P = \sum_{\alpha=1}^{D_0}  \frac1{d_\alpha}PF_\alpha^\dag  P F_\alpha P = P.
\ee
Thus, we arrive at
\be\label{eqs:idempotent_dd}
\begin{split}
\forall \rho:\quad \mathcal R^2(\rho) &= \sum_{\alpha\beta=1}^{D_0}  R_\alpha R_\beta \rho R^\dag_\beta R^\dag_\alpha +\frac 1q \Tr(\rho P_{\perp})P  = \mathcal R(\rho).
\end{split}
\ee
showing that the recovery map is an idempotent operation.}

Finally, using the structure of the Kraus operators in Eq.~\eqref{eq:recovery_map_kraus}, we can derive its action on error states as
\be
\forall \rho \in \mathsf{L}(\mathsf C):\quad \mathcal R(F_{\alpha}\rho F^\dag_{\beta}) = \sum_{\gamma=1}^{D_0}  PU^\dag_\gamma F_{\alpha}\rho F^\dag_{\beta}
U_\gamma P =  \sum_{\gamma=1}^{D_0}  \frac{1}{d_\gamma}  PF^\dag_\gamma F_{\alpha}P\rho P F^\dag_{\beta}
F_\gamma P,
\ee
where $\mathsf{L}(\mathsf C)$ is the space of linear operators on the codespace $\mathsf C$.
Using the Knill-Laflamme condition, we get the expression
\be
\forall \rho \in \mathsf{L}(\mathsf C):\quad\mathcal R(F_{\alpha}\rho F^\dag_{\beta}) = \sum_{\gamma=1}^{D_0}  d_{\gamma}\delta_{\alpha\gamma}\delta_{\beta\gamma} \rho   = d_\alpha \delta_{\alpha\beta} \rho.
\ee
Transforming back to the error basis consisting of individual errors, we get
\be\label{eq:r-invariance_exact}
\forall \rho \in \mathsf{L}(\mathsf C):\quad \mathcal R(K_{\bsmu}\rho K^\dag_{\bsnu}) = C_{\bsnu\bsmu} \rho.
\ee
This relation is important: we will use it for proving the properties of the recovery map in \refx[s]{appendix_generic_bound}.

Next, we present a proof of the relationship between the trace distance and the fidelity measures of logical error, defined in Eqs.~\eqref{eqs:trace_dis} and \eqref{eq:overlap_log_error}. In particular, we show that
\be
\delta(t)\leq 2\epsilon(t).
\ee
The first step is to utilize Holder's inequality, namely
\be\label{eqs:delta_bound}
T\Bigl((\exp(\mathcal Lt)|\psi_0\>\<\psi_0|, \exp(\mathcal Lt)|\psi_1\>\<\psi_1|\Bigl)\geq \frac 12\Tr(Q\exp(\mathcal Lt)\delta\rho_0),
\ee
where $\delta\rho_0 := |\psi_0\>\<\psi_0|-|\psi_1\>\<\psi_1|$, and $Q$ is any Hermitian operator of unit spectral norm. It is convenient to choose
$
Q = \mathcal R^\dag\delta\rho_0,
$
where $\mathcal R^\dag$ is the adjoint to the recovery operator $\mathcal R$. Then
\review{\be\label{eqs:trac_ineq_34}
\begin{split}
T\Bigl(\exp(\mathcal Lt)|\psi_0\>\<\psi_0|, \exp(\mathcal Lt)|\psi_1\>\<\psi_1|\Bigl)&\geq \frac 12\Tr\Bigl(\delta\rho_0\mathcal R\exp(\mathcal Lt)\delta\rho_0\Bigl).
\end{split}
\ee
Next, using the fact that $|\psi_0\>$ and $|\psi_1\>$ are orthogonal states in the codespace $\mathsf C$, we can rewrite the r.h.s. of this inequality using $|\psi_0\>\<\psi_0|+|\psi_1\>\<\psi_1| \leq I$, where $I$ is the identity operator. 
This leads us to
\be\label{eqs:inequality_trdis}
\begin{split}
T\Bigl(\exp(\mathcal Lt)|\psi_0\>\<\psi_0|, \exp(\mathcal Lt)|\psi_1\>\<\psi_1|\Bigl)\geq\Tr\Bigl(&|\psi_0\>\<\psi_0|\mathcal R\exp(\mathcal Lt)|\psi_0\>\<\psi_0|\Bigl)\\
&+\Tr\Bigl(|\psi_1\>\<\psi_1|\mathcal R\exp(\mathcal Lt)|\psi_1\>\<\psi_1|\Bigl)-1.
\end{split}
\ee
where we used the condition $\Tr(\mathcal R \exp(\mathcal Lt) \rho) = \Tr\rho = 1$ valid for any density matrix $\rho$ since both $\exp(\mathcal Lt)$ and $\mathcal R$ are trace preserving maps.
Incorporating this inequality into the definition of $\delta(t)$, we get
\be\label{eq:rep_inequality}
\begin{split}
\delta(t)&\leq 2-\min_{|\psi_0\>,|\psi_1\>\in G}\Biggl(\Tr\Bigl(|\psi_0\>\<\psi_0|\mathcal R\exp(\mathcal Lt)|\psi_0\>\<\psi_0|)+\Tr\Bigl(|\psi_1\>\<\psi_1|\mathcal R\exp(\mathcal Lt)|\psi_1\>\<\psi_1|\Bigl)\Biggl)\\
&\leq 2\Bigl(1-\min_{|\psi_0\>\in \mathsf C}\Tr(|\psi_0\>\<\psi_0|\mathcal R\exp(\mathcal Lt)|\psi_0\>\<\psi_0|)\Bigl)=2\epsilon(t).
\end{split}
\ee
This concludes our proof.}

\subsection*{Appendix B: Proof of Theorem \ref{generic_bound}}
\labelx[Appendix B]{appendix_generic_bound}

\review{This appendix contains the proof of Theorem~\ref{generic_bound}. In the first step of the proof, we show that the error measures in Eqs.~\eqref{eq:overlap_log_error} satisfy
\be\label{eq:second_error_measure}
\epsilon(t), \delta(t) \leq 1-\frac 12\min_{|\psi_0\>,|\psi_1\>\in G}\Tr(Q\exp(\mathcal L t)\delta\rho_0),
\ee
where $\delta \rho_0 := |\psi_0\>\<\psi_0|-|\psi_1\>\<\psi_1|$ and
$
Q = \mathcal R^\dag (|\psi_0\>\<\psi_0|-|\psi_1\>\<\psi_1|)
$, similar to the notation we used in \refx[s]{appendix_error_measures}. The inequality for $\delta(t)$ follows directly from its definition in Eq.~\eqref{eqs:trace_dis} and the property in Eq.~\eqref{eqs:delta_bound}. To prove this inequality for $\epsilon(t)$, we notice that
\be
\begin{split}
\frac 12\min_{|\psi_0\>,|\psi_1\>\in G}\Tr(Q\exp(\mathcal L t)\delta\rho_0) &= \min_{|\psi_0\>,|\psi_1\>\in G}\Bigl(\Tr(|\psi_0\>\<\psi_0|\mathcal R\exp(\mathcal Lt)|\psi_0\>\<\psi_0|)\\
& \qquad\qquad\qquad+\Tr(|\psi_1\>\<\psi_1|\mathcal R\exp(\mathcal Lt)|\psi_1\>\<\psi_1|)\Bigl)-1\\
&\leq \min_{|\psi_0\>\in \mathsf C}\Tr\Bigl(|\psi_0\>\<\psi_0|\mathcal R\exp(\mathcal Lt)|\psi_0\>\<\psi_0|\Bigl) = 1-\epsilon(t),
\end{split}
\ee
where we used the fact that $\Tr(|\psi_1\>\<\psi_1|\mathcal R\exp(\mathcal Lt)|\psi_1\>\<\psi_1|)\leq 1$. This expression leads us to the inequality in Eq.~\eqref{eq:second_error_measure} for $\epsilon(t)$.}

Next, it is convenient to switch to the imaginary frequency space $t\to s$ and write this inequality using the inverse Laplace transform $\mathscr{L}^{-1}$ as
\be\label{eqs:error_laplace}
\epsilon(t), \delta(t)\leq 1- \frac 12\min_{\rho_0}\mathscr{L}^{-1}\Biggl[\Tr\left(Q\frac{1}{s-\mathcal L}\delta\rho_0\right)\Biggl].
\ee
where $(s-\mathcal L)^{-1}$ is called the resolvent of $\mathcal L$. To further analyze this expression, we use the decomposition $\mathcal L = \kappa\mathcal L_R+\Delta \mathcal L_E$, where $\mathcal L_E$ and $\mathcal L_R$ are defined in Eqs.~\eqref{eq:noisy_L} and \eqref{eq:recovery_lindblad}, respectively. With error rate $\Delta$ as a small parameter, we are using Dyson's series
\be\label{eqs:dyson}
\frac{1}{s-\mathcal L} = \frac{1}{s-\mathcal L_R}\sum_{r=0}^\infty\Bigl(\Delta\mathcal L_E \frac{1}{s-\mathcal L_R}\Bigl)^r.
\ee
To simplify calculations, we can use diagrammatic notation to represent different superoperators. We introduce the following notation:
\be\label{eq:diagram_notations}
Q\otimes I \equiv \; \sq, \qquad
 \frac{1}{s-\mathcal L} \equiv \; \sss, \qquad \frac{1}{s-\mathcal L_R} \equiv \zzz\; , \qquad \Delta\mathcal L_E \equiv \x, \qquad \frac 12\Tr(\mathcal O\delta\rho_0)\equiv \<\mathcal O\>.
\ee\label{dagchange3}
The notation $A\otimes B^*$ stands for the matrix representation of the superoperator, which acts as $A$ from the left and as $B^\dag$ from the right, i.e. $(A\otimes B^*)(\rho) \equiv A\rho B^\dag$. Consequently, $Q\otimes I$ corresponds to the left multiplication with the operator $Q$.
Using this notation, the Dyson's series in Eq.~\eqref{eqs:dyson} can be expressed as an infinite sum of diagrams,
\be\label{eqs:dyson_diagrams}
\sss = \zzz+\zzz\x\zzz+\zzz\x\zzz\x\zzz+\zzz\x\zzz\x\zzz\x\zzz+\dots.
\ee
At the same time, the error expression in Eq.~\eqref{eqs:error_laplace} takes the diagrammatic form
\be
\epsilon(t),\delta(t) \leq 1-\frac 12\mathscr{L}^{-1}\,\Tr\left(Q\frac{1}{s-\mathcal L}\delta\rho_0\right) = 1-\mathscr{L}^{-1}\bigl\<\sq\sss\bigl\>.
\ee
Using Dyson's expansion and the diagrammatic representation, we can now rewrite the term on the right as
\be
\bigl\<\sq\sss\bigl\> = \bigl\<\sq\zzz\bigl\>+\bigl\<\sq\zzz\x\zzz\bigl\>+\bigl\<\sq\zzz\x\zzz\x\zzz\bigl\>+\bigl\<\sq\zzz\x\zzz\x\zzz\x\zzz\bigl\>+\dots.
\ee
It can be further simplified once we take into account the fact that recovery dynamics preserves the states in the codespace, which means $\mathcal L_R\delta\rho_0 = 0$. Therefore,
\be
\frac{1}{s-\mathcal L_R}\delta\rho_0 = \frac {1}{s}\delta\rho_0.
\ee
This property allows us to rewrite
\be
\bigl\<\sq\sss\bigl\> = \frac 1s+\frac 1s\bigl\<\sq\zzz\x\bigl\>+\frac 1s\bigl\<\sq\zzz\x\zzz\x\bigl\>+\frac 1s\bigl\<\sq\zzz\x\zzz\x\zzz\x\bigl\>+\dots.
\ee
Next, we use the decomposition
\be\label{eq:decoder_evo_expl_form}
\begin{split}
\exp(\mathcal L_Rt) = \mathcal W_t+\bigl(1-e^{-\kappa t}\bigl)&\mathcal R,
\end{split}
\ee
where we defined $\mathcal W_t := e^{-\kappa t}\mathcal K_t$. \review{This operator, similar to $\exp(\mathcal L_R t)$, satisfies the property in Eq.~\eqref{eq:good_correction}, i.e.
\be\label{eqs:cond_on_Wt}
\mathcal W_t(K_{\bsmu}\rho_0K^\dag_{\bsnu}) = \sum_{\bsmu'\bsnu'}b_{\bsmu\bsnu,\bsmu'\bsnu'}(t)K_{\bsmu'}\rho_0K^\dag_{\bsnu'},
\ee
where $b_{\bsmu\bsnu,\bsmu'\bsnu'}(t)=0$ if $|\bsmu'|>|\bsmu|$ or $|\bsnu'|>|\bsnu|$. Indeed,
\be\label{eqs:rewritten_K_arg}
\begin{split}
\mathcal W_t(K_{\bsmu}\rho_0K^\dag_{\bsnu}) &= \mathcal \exp(\mathcal L_R t)(K_{\bsmu}\rho_0K^\dag_{\bsnu})-(1-e^{-\kappa t})\mathcal R(K_{\bsmu}\rho_0K^\dag_{\bsnu})\\
& = \sum_{\bsmu'\bsnu'}(a_{\bsmu\bsnu,\bsmu'\bsnu'}(t)-C_{\mathcal \bsnu\bsmu}\delta_{\bsmu'\{\emptyset\}}\delta_{\bsnu'\{\emptyset\}})K_{\bsmu'}\rho_0K^\dag_{\bsnu'},
\end{split}
\ee
where we used Eq.~\eqref{eq:r-invariance_exact}.}
In the space of imaginary frequencies, the same expression takes the form
\be\label{eq:decoder_evo_expl_form_laplace}
\begin{split}
\frac1{s- \mathcal L_R} = \widehat{\mathcal  W_s}+\frac{\kappa}{s(s+\kappa)}\mathcal R, \qquad \widehat{\mathcal W}_s = \mathscr{L}\mathcal W_t.
\end{split}
\ee
This expression can also be written in diagrammatic form
\be\label{eqs:V_decomposition}
\zzz = \vvv+\nnn, \qquad \frac{\kappa}{s(s+\kappa)}\mathcal R= \vvv,\qquad
\widehat{\mathcal W}_s= \nnn.
\ee
Now let us consider the terms from the second to the $\ell$th (containing $k$ operators $\mathcal L_E$, where  $1 \leq k \leq \ell$), and rewrite it by expanding the last term
\be
\begin{split}
\bigl\<\sq\zzz\underbrace{\x\zzz\x\dots\x\zzz\x}_\textrm{$0<k\leq \ell$}\bigl\> =  \<\sq\zzz\underbrace{\x\zzz\x\dots\x\zzz\x}_\textrm{$k-1$}\nnn \x\> + \<\sq\zzz\underbrace{\x\zzz\x\dots\x\zzz\x}_\textrm{$k-1$}\vvv \x\>.
\end{split}
\ee
Then we take the first term on the right and decompose the double line next to the arrow again, which gives us
\be
\begin{split}
\bigl\<\sq\zzz\underbrace{\x\zzz\x\dots\x\zzz\x}_\textrm{$0<k\leq \ell$}\bigl\> =  \<\sq\zzz\underbrace{\x\zzz\x\dots\x\zzz\x}_\textrm{$k-2$}\nnn \x\nnn \x\> &+ \<\sq\zzz\underbrace{\x\zzz\x\dots\x\zzz\x}_\textrm{$k-2$}\vvv \x\nnn \x\>\\
&+ \<\sq\zzz\underbrace{\x\zzz\x\dots\x\zzz\x}_\textrm{$k-1$}\vvv \x\>.
\end{split}
\ee
Repeating this procedure with the first term several times until all double lines are decomposed, we get
\be\label{eqs:first_terms_full}
\begin{split}
\bigl\<\sq\zzz\underbrace{\x\zzz\x\dots\x\zzz\x}_\textrm{$0<k\leq \ell$}\bigl\> =& \<\sq\nnn\underbrace{\x\nnn\x\dots\x\nnn\x}_\textrm{$k$}\>+\<\sq\vvv \underbrace{\x\nnn\x\dots\x\nnn\x}_\textrm{$k$}\>\\
&+\sum_{m=1}^k \<\sq\zzz\underbrace{\x\zzz\x\dots\x\zzz\x}_\textrm{$k-m$}\vvv\underbrace{\x\nnn\x\dots\x\nnn\x}_\textrm{$m$}\>.
\end{split}
\ee
The right-hand side of this equation contains $k+2$ terms.  Our goal is to show that the right-hand side of this expression vanishes. We do so by using the following Lemma.

\begin{lem} \label{lem-diag-canc} For any $k$ satisfying $0<k\leq \ell$ and any superoperator $\mathcal O$, we have
\be\begin{split}
&\Tr(\mathcal O \mathcal R\mathcal L_E(\widehat{\mathcal W}_s\mathcal L_E)^{k-1} \delta\rho_0) = 0,\\
&\Tr(\mathcal O \mathcal R\mathcal (\widehat{\mathcal W}_s\mathcal L_E)^{k} \delta\rho_0) = 0.
\end{split}
\ee
In diagrammatic form this is
\be
\begin{split}
\bigl\<\mathcal O\vvv\underbrace{\x\nnn\x\dots\x\nnn\x}_\textrm{$k$}\> = 0,\qquad \bigl\<\mathcal O\vvv\nnn\underbrace{\x\nnn\x\dots\x\nnn\x}_\textrm{$k$}\> = 0.
\label{knill-lafl-lemma}
\end{split}
\ee
\end{lem}
\begin{proof}.
\review{
We express the action of the Lindblad generators in a more explicit form:
\be\label{eqs:L-change_err98}
\begin{split}
\forall \bsmu,\bsmu'\in M_k: \quad \mathcal L_E \left[K_{\boldsymbol \mu} \rho K^\dag_{\boldsymbol \mu'}\right] &= \sum_\mu \Bigl(E_\mu K_{\boldsymbol \mu} \rho K^\dag_{\boldsymbol \mu'} E^\dag_\mu -\frac 12\{E^\dag_\mu E_\mu, K_{\boldsymbol \mu} \rho K^\dag_{\boldsymbol \mu'}\}\Bigl) \\
&= \sum_{\bsnu , \bsnu' \in M_{k+1}} L_{\bsmu\bsmu',\bsnu\bsnu'} K_{\boldsymbol \nu} \rho K^\dag_{\boldsymbol \nu'},
\end{split}
\ee
where $M_k = \{\bsmu: |\bsmu|\leq k\}$ is a set of trejectories of length smaller than $k$ and $L_{\bsmu\bsmu',\bsnu\bsnu'}$ are certain real coefficients. At the same time, according to the strict error-reducing condition in Eq.~\eqref{eq:good_correction}, we get
\be
\forall \bsmu,\bsmu'\in M_k:  \quad \exp(\mathcal L_R t)\left[K_{\boldsymbol \mu} \rho K^\dag_{\boldsymbol \mu'} \right]= \sum_{\bsnu , \bsnu' \in M_k} R_{\bsmu\bsmu',\bsnu\bsnu'} (t) K_{\boldsymbol \nu} \rho K^\dag_{\boldsymbol \nu'}.
\ee
Next, we use the definition of $\widehat{\mathcal W}_s$ in Eq.~\eqref{eq:decoder_evo_expl_form_laplace} and $\mathcal W_t$ in Eq.~\eqref{eq:decoder_evo_expl_form}, as well as the property in Eq.~\eqref{eq:r-invariance_exact} to show that
\be\label{eqs:Ws-action-non-reducing}
\begin{split}
\forall \bsmu,\bsmu'\in M_C:\qquad \widehat{\mathcal W}_s\left[K_{\boldsymbol \mu} \rho K^\dag_{\boldsymbol \mu'} \right] &= \mathscr{L}\Biggl(\exp(\mathcal L_R t)\left[K_{\boldsymbol \mu} \rho K^\dag_{\boldsymbol \mu'} \right] - (1-e^{-\kappa t})\mathcal R\left[K_{\boldsymbol \mu} \rho K^\dag_{\boldsymbol \mu'} \right]\Biggl) \\
  &= \sum_{\bsnu , \bsnu' \in M_k} \mathscr{L}\Bigl(R_{\bsmu\bsmu',\bsnu\bsnu'} (t)-(1-e^{-\kappa t})C_{\bsmu'\bsmu}\delta_{\bsmu\bsnu}\delta_{\bsmu'\bsnu'}\Bigl) K_{\boldsymbol \nu} \rho K^\dag_{\boldsymbol \nu'} \\
  &= \sum_{\bsnu , \bsnu' \in M_k} W_{\bsmu\bsmu',\bsnu\bsnu'}(s) K_{\boldsymbol \nu} \rho K^\dag_{\boldsymbol \nu'},
 \end{split}
\ee
where $W_{\bsmu\bsmu',\bsnu\bsnu'}(s)$ are elements of a given real-valued matrix function of $s$. Using Eqs.~\eqref{eqs:L-change_err98} and \eqref{eqs:Ws-action-non-reducing} we get for any $\rho$ in the codespace we have
\be\label{eq:eguig46}
\begin{split}
&\mathcal L_E(\widehat{\mathcal W}_s\mathcal L_E)^{k-1} \rho =  \sum_{\bsnu , \bsnu' \in M_k}A_{\bsmu\bsmu',\bsnu\bsnu'}(s) K_{\boldsymbol \mu} \rho K^\dag_{\boldsymbol \mu'},\\
&\widehat{\mathcal W}_s\mathcal L_E(\widehat{\mathcal W}_s\mathcal L_E)^{k-1} \rho =  \sum_{\bsnu , \bsnu' \in M_k}B_{\bsmu\bsmu',\bsnu\bsnu'}(s) K_{\boldsymbol \mu} \rho K^\dag_{\boldsymbol \mu'},
\end{split}
\ee
where we use the notations for the matrices $A(s):= L(W(s)L)^{k-1}$ and $B(s) := (W(s)L)^{k}$. It is important to note that, as the action of any Linbladian operator must return traceless operators, these operators must satisfy
\be\label{eqs:traceless_lindb}
\forall \rho, k>0: \quad \Tr(\mathcal L_E(\widehat{\mathcal W}_s\mathcal L_E)^{k-1} \rho) = 0, \quad \Tr(\widehat{\mathcal W}_s\mathcal L_E(\widehat{\mathcal W}_s\mathcal L_E)^{k-1} \rho) = 0.
\ee
The second inequality follows from the first, if we consider that $\mathcal W_t = e^{-\kappa t}\mathcal K_t$, $\mathcal K_t$ preserves the trace, and that for any traceless operator $O$ we have
\be\label{eqs:specified_eq_tra_cons}
{\rm Tr}[\mathscr{L}\, \mathcal W_t(O)] = \mathscr{L}\, {\rm Tr}[\mathcal W_t(O)] =\mathscr{L}\, e^{-\kappa t} {\rm Tr}[\mathcal K_t(O)] = \mathscr{L}\, e^{-\kappa t} {\rm Tr}[O] = 0.
\ee
Thus, for any $\rho$, the conditions in Eq.~\eqref{eq:eguig46}  satisfy
\be
\sum_{\bsmu,\bsmu'\in M_C}A_{\bsmu\bsmu'}\Tr (K_{\boldsymbol \mu'}\rho K_{\boldsymbol \mu}^\dag)=0, \qquad
\sum_{\bsmu,\bsmu'\in M_C}B_{\bsmu\bsmu'}\Tr (K_{\boldsymbol \mu'}\rho K^\dag_{\boldsymbol \mu})=0.
\ee
where we dropped the functional dependendency on $s$ for simplicity.
Now we can express the portion of the expression inside the trace in Eq.~\eqref{eq:r_kills_liouvillian} using the property of the recovery operator in Eq.~\eqref{eq:r-invariance_exact}:
\be
\begin{split}
\mathcal R\mathcal L_E(\widehat{\mathcal W}_s\mathcal L_E)^{k-1} \delta\rho_0 & =\sum_{\bsmu,\bsmu'\in M_C} A_{\bsmu\bsmu'} \mathcal R(K_{\boldsymbol \mu} \delta\rho_0 K^\dag_{\boldsymbol \mu'})\\
&= \delta\rho_0 \sum_{\bsmu,\bsmu'\in M_C} A_{\bsmu\bsmu'} C_{\bsmu'\bsmu} = \delta\rho_0 \sum_{\bsmu,\bsmu'\in M_C} A_{\bsmu\bsmu'} \<0|K^\dag_{\boldsymbol \mu'}K_{\boldsymbol \mu}|0\> = 0\\
& = \delta\rho_0 \sum_{\bsmu,\bsmu'\in M_C} A_{\bsmu\bsmu'} \Tr(K_{\boldsymbol \mu'}|0\>\<0|K^\dag_{\boldsymbol \mu}) = 0..
\end{split}
\ee
Similarly,
\be
\begin{split}
\mathcal R\widehat{\mathcal W}_s\mathcal L_E(\widehat{\mathcal W}_s\mathcal L_E)^{k-1} \delta\rho_0 & =\sum_{\bsmu,\bsmu'\in M_C} B_{\bsmu\bsmu'} \mathcal R(K_{\boldsymbol \mu} \delta\rho_0 K^\dag_{\boldsymbol \mu'})\\
&= \delta\rho_0 \sum_{\bsmu,\bsmu'\in M_C} B_{\bsmu\bsmu'} C_{\bsmu'\bsmu} = \delta\rho_0 \sum_{\bsmu,\bsmu'\in M_C} B_{\bsmu\bsmu'} \<0|K^\dag_{\boldsymbol \mu'}K_{\boldsymbol \mu}|0\> = 0\\
& = \delta\rho_0 \sum_{\bsmu,\bsmu'\in M_C} B_{\bsmu\bsmu'} \Tr(K_{\boldsymbol \mu'}|0\>\<0|K^\dag_{\boldsymbol \mu}) = 0.
\end{split}
\ee
Thus, both expressions in Eq.~\eqref{eq:r_kills_liouvillian} vanish and this leads to the statement of the Lemma.
To prove the last statement, we first write the diagram in symbolic form:
\be\label{eq:r_kills_liouvillian}
\begin{split}
&\bigl\<\mathcal O\vvv\underbrace{\x\nnn\x\dots\x\nnn\x}_\textrm{$k$}\bigl\> \equiv \frac {\Delta^k}2 \Tr(\mathcal O \mathcal R\mathcal L_E(\widehat{\mathcal W}_s\mathcal L_E)^{k-1} \delta\rho_0),\\
&\bigl\<\mathcal O\vvv\nnn\underbrace{\x\nnn\x\dots\x\nnn\x}_\textrm{$k$}\bigl\> \equiv \frac {\Delta^k}2 \Tr(\mathcal O \mathcal R\mathcal (\widehat{\mathcal W}_s\mathcal L_E)^{k} \delta\rho_0).
\end{split}
\ee
This expression concludes out proof.}
\end{proof}.
\\

As a result of this Lemma, all terms on the right-hand side of Eq.~\eqref{eqs:first_terms_full} vanish immediately. To show that the first term must vanish, we can use the definition of the operator $Q$ below Eq.~\eqref{eq:second_error_measure} and rewrite
\be
\<\sq\,\nnn\x^k\> = \<(\delta\rho_0 \otimes I)\vvv\nnn\underbrace{\x\nnn\x\dots\x\nnn\x}_\textrm{$k$}\>.
\ee
Therefore we have
\be\label{eqs:KL_cancelling}
\begin{split}
\bigl\<\sq\zzz\underbrace{\x\zzz\x\dots\x\zzz\x}_\textrm{$0<k\leq \ell$}\bigl\> = 0.
\end{split}
\ee
Following the removal of the vanishing terms, we obtain the series
\be
\bigl\<\sq\sss\bigl\> = \frac 1s+\frac 1s\bigl\<\sq\zzz\underbrace{\x\zzz\x\dots\x\zzz\x}_\textrm{$\ell+1$}\bigl\>+\frac 1s\bigl\<\sq\zzz\underbrace{\x\zzz\x\dots\x\zzz\x}_\textrm{$\ell+2$}\bigl\>+\frac 1s\bigl\<\sq\zzz\underbrace{\x\zzz\x\dots\x\zzz\x}_\textrm{$\ell+3$}\bigl\>+\dots.
\ee
This expression can be compacted again using Eq.~\eqref{eqs:dyson_diagrams} to obtain
\be\label{eq:dyson_with_KL}
\bigl\<\sq\sss\bigl\> = \frac 1s+ \frac 1s \bigl\<\sq\sss\underbrace{\x\zzz\x\dots\x\zzz\x}_\textrm{$\ell+1$}\bigl\>.
\ee
The rest of the proof uses this expression iteratively. As the first step, we rewrite
Eq.~\eqref{eq:dyson_with_KL} as
 \be\label{eqs:basic_decomp0}
\bigl\<\sq\sss\bigl\> = \frac 1s + \frac 1s\bigl\<\sq\sss\underbrace{\x\zzz\x\dots\x\zzz\x}_\textrm{$\ell$}\nnn\x\bigl\>+\frac 1s\bigl\<\sq\sss\underbrace{\x\zzz\x\dots\x\zzz\x}_\textrm{$\ell$}\vvv\x\bigl\>.
\ee
The last term vanishes as a result of Lemma~\ref{lem-diag-canc}. \review{To see this, we set
\be\label{eqs:O_lem_decsr}
\mathcal O = \sq\sss\underbrace{\x\zzz\x\dots\x\zzz\x}_\textrm{$\ell$}
\ee
and apply Lemma~\ref{lem-diag-canc} for the case $k=1$.  Then, we} rewrite the remaining terms as
\be\label{eqs:basic_decomp}
\bigl\<\sq\sss\bigl\> = \frac 1s + \frac 1{s}\bigl\<\sq\sss\underbrace{\x\zzz\dots\zzz\x}_\textrm{$\ell$}\nnn\x\bigl\>.
\ee
\review{Next we repeat the procedure in Eq.~\eqref{eqs:basic_decomp0}: we apply the decomposition $\zzz = \nnn+\vvv$ for the rightmost resolvent,
\be
\bigl\<\sq\sss\bigl\> = \frac 1s + \frac 1{s}\bigl\<\sq\sss\underbrace{\x\zzz\dots\zzz\x}_\textrm{$\ell-1$}\nnn\x\nnn\x\bigl\> + \frac 1{s}\bigl\<\sq\sss\underbrace{\x\zzz\dots\zzz\x}_\textrm{$\ell-1$}\vvv\x\nnn\x\bigl\>.
\ee
and, using Lemma~\ref{knill-lafl-lemma} for $k=2$, show that the last term vanishes, therefore
\be
\bigl\<\sq\sss\bigl\> = \frac 1s + \frac 1{s}\bigl\<\sq\sss\underbrace{\x\zzz\dots\zzz\x}_\textrm{$\ell-1$}\nnn\x\nnn\x\bigl\>.
\ee}
We repeat these two steps $\ell-2$ more times to get
\be\label{eqs:final_step_123}
\bigl\<\sq\sss\bigl\> = \frac 1s + \frac 1{s}\bigl\<\sq\sss\underbrace{\x\nnn\x\dots\x\nnn\x}_\textrm{$\ell+1$}\bigl\>.
\ee
As the next step, we take the inverse Laplace transform of the last term using the property that applies to any two analytic functions/operators $g_1$ and $g_2$ \reviewtwo{and their Laplace transforms $\hat g_1$ and $\hat g_2$},
\be\label{eqs:laplace_101}
\mathscr{L}^{-1} [\reviewtwo{\hat g_1(s) \hat g_2(s)}] = \int_0^t dt' g_1(t')g_2(t-t').
\ee
Applying this formula $\ell$ times, we get generalized expression for $\ell$ operators $g_1,\dots,g_{\ell}$ \reviewtwo{and their Laplace transforms $\hat g_1 \dots \hat g_{\ell+1}$},
\be
\mathscr{L}^{-1} [\reviewtwo{\hat g_1(s) \dots  \hat g_{\ell+1}(s)}] = \mathcal T\int dt_1\dots dt_{\ell} g_1(t_1)g_2(t_2-t_1)\dots g_{\ell +1} (t-t_{\ell}).
\ee
where we defined the time-ordered exponential as
\be\label{eqs:tm_ordrd}
\mathcal T\int dt_1\dots dt_k = \int_0^tdt_1\int_{0}^{t_1}dt_2\dots \int_0^{t_{k-1}}dt_k.
\ee
Using this definition, we translate the diagrammatic terms in the second term of Eq.~\eqref{eqs:final_step_123} back into mathematical operator notation, yielding
\be\label{eqs:Linf_conv}
\begin{split}
\mathscr{L}^{-1}\frac 1{s}\bigl\<\sq\sss &\underbrace{\x\nnn\x \dots\x\nnn\x}_{\textrm{$\ell+1$}}\bigl\>  \equiv \frac 12\Delta^{\ell+1}\mathcal T\int dt_0dt_1\dots dt_{\ell} \Tr \Bigl( Q e^{\mathcal L(t_1-t_0)}\mathcal L_E \mathcal W_{t_2-t_1} \dots \mathcal W_{t-t_{\ell}} \mathcal L_E\delta\rho_0\Bigl)\\
&\geq -\frac 12\Delta^{\ell+1}\mathcal T\int dt_0dt_1\dots dt_{\ell}\|e^{\mathcal L^\dag (t_1-t_0)}(Q)\| \|\mathcal L_E \mathcal W_{t_2-t_1} \dots \mathcal W_{t-t_{\ell}} \mathcal L_E\delta\rho_0\|_1\\
&\geq -\frac 12\Delta^{\ell+1}\mathcal T\int dt_0dt_1\dots dt_{\ell} \|\mathcal L_E \mathcal W_{t_2-t_1} \dots \mathcal W_{t-t_{\ell}} \mathcal L_E\delta\rho_0\|_1.
\end{split}
\ee 
Our next goal is to find an appropriate bound for the expression under the time-ordered integrals. Let 
$\mathsf{N}_0 = \{\rho : \operatorname{Tr}(\rho) = 1, \, \rho \geq 0, \rho \in {\rm End}(\mathsf C)\oplus 0\}$ be the set of density operators in the codespace.  Then we can prove the following lemma.
\begin{lem}
For all $\rho_1, \rho_2 \in \mathsf N_0$ and $\tau_j \geq 0$ for all $j \in \{1, \dots, \ell\}$, we have
\be 
\left\|\mathcal L_E\mathcal W_{\tau_\ell}\mathcal L_E \dots \mathcal W_{\tau_1} \mathcal L_E(\rho_1-\rho_2) \right\|_1 \leq 2\exp\left(-\kappa \sum_{j=1}^k\tau_j\right) (\chi+1)^{\ell+1} \|\mathcal L_E\|^\ell_{1\to1, \mathsf E}.
\ee
\end{lem}
\begin{proof}  We first introduce the set of density operators with an error weight of at most $j$. For any integer $j\geq 0$, this set is defined as
\be\label{eqs;NEk_def}
\mathsf{N}_j := \left\{ \rho = \sum_{\boldsymbol{\mu}, \boldsymbol{\nu} \in S_j} a_{\boldsymbol{\mu}\boldsymbol{\nu}} K_{\boldsymbol{\mu}} \rho_0 K_{\boldsymbol{\nu}} : \rho \geq 0, \, \operatorname{Tr}(\rho) = 1, \, a_{\boldsymbol{\mu}\boldsymbol{\nu}} \in \mathbb{C}, \, \rho_0 \in \operatorname{End}(\mathsf{C}) \oplus 0 \right\},
\ee
where, as before, $K_{\bsmu} = \prod_{\mu \in \bsmu} K_{\mu}$ represents the composite error, $S_j := \{\bsmu : |\bsmu| \leq j, \forall \mu \in\bsmu: K_{\mu} \in E\}$ denotes the set of errors with length at most $j$, and $E := \{E_\alpha, E_\alpha^\dagger E_\alpha : \alpha \in \{1, \dots, N\}\}$ is the set of error operators. Naturally, for all $j > 0$, the inclusions $\mathsf{N}_0 \subset \mathsf{N}_j \subset \mathsf{N}$ hold.

Then, to establish the result of the Lemma, we first demonstrate that
\be\label{eqs:recurrence_KL}
\forall \rho_1, \rho_2 \in \mathsf{N}_{j}, \quad \exists C_j, \rho'_1, \rho'_2 \in \mathsf{N}_{j+1} : \quad \mathcal{W}_\tau \mathcal{L}_E(\rho_1 - \rho_2) = e^{-\kappa\tau}C_j(\rho'_1 - \rho'_2),
\ee
where $C_j \leq (\chi+1)\|\mathcal{L}_E\|_{1 \to 1, \mathsf{E}}$.

 Next, we define similar sets of general Hermitian operators that are not subject to any constraints on their eigenvalues,
\be\label{eqs:sdcsd-pllo}
\mathsf{O}_j := \left\{ O = \sum_{\boldsymbol{\mu}, \boldsymbol{\nu} \in S_j} a_{\boldsymbol{\mu}\boldsymbol{\nu}} K_{\boldsymbol{\mu}} M K_{\boldsymbol{\nu}} : \,O = O^\dagger, \, \operatorname{Tr}(O) = 0,\, a_{\boldsymbol{\mu}\boldsymbol{\nu}} \in \mathbb{C}, \, M \in \operatorname{End}(\mathsf{C}) \oplus 0 \right\}.
\ee

Using the definition in Eq.~\eqref{eqs:sdcsd-pllo} and the fact that Lindbladian $\mathcal L_E$ is trace preserving, we express the action of the Lindbladian as
\be\label{eqs:tr_rep_Lrhodiff}
\begin{split}
 \mathcal{L}_E(\rho_1 - \rho_2) &= \sum_{\bsmu, \bsnu \in S_j} a_{\bsmu \bsnu} \sum_\alpha \Bigl( E_\alpha K_\bsmu M K^\dag_\bsnu E_\alpha^\dag - \frac{1}{2} E_\alpha^\dag E_\alpha K_\bsmu M K^\dag_\bsnu - \frac{1}{2} K_\bsmu M K^\dag_\bsnu E_\alpha^\dag E_\alpha \Bigr) \\
 &= \sum_{\bsmu, \bsnu \in S_{j+1}} b_{\bsmu \bsnu} K_{\bsmu} M K^\dag_{\bsnu} \in \mathsf O_{j+1},
\end{split}
\ee
where $M\in \mathsf O_0$ and $b_{\bsmu \bsnu}$ are certain coefficients that linearly depend on $a_{\bsmu \bsnu}$.

Then we note that, for any traceless Hermitian operator from $\mathsf O_j$, we can express it in the form
\be\label{eqs:trace_rep_anyO}
\forall O\in \mathsf O_{j}, \qquad \exists \rho_1, \rho_2 \in \mathsf{N}_j : \quad O = \frac{1}{2} \|O\|_1 (\rho_1 - \rho_2).
\ee
In fact, any Hermitian operator can be written as \( O = A_+ - A_- \), where \( A_{\pm} \geq 0 \) are positive semidefinite operators. Therefore, the trace norm of the operator \( O \) can be written as
\be
\|O\|_1 = \Tr(A_+) + \Tr(A_-).
\ee
Thus, the zero trace condition leads to $\Tr(A_+) = \Tr(A_-) = \frac{1}{2} \|O\|_1$ and, consequently, to Eq.~\eqref{eqs:trace_rep_anyO}, where we define $\rho_{1,2} := A_\pm/\Tr(A_\pm)$. Since $\Tr \rho_{1,2} = 1$ and $\rho_{1,2}\geq 0$, they belong to $\mathsf N_j$.

Thus, using Eq.~\eqref{eqs:trace_rep_anyO}, we get 
\be
\forall \rho_1, \rho_2 \in \mathsf{N}_j \quad \exists \rho'_1, \rho'_2 \in \mathsf{N}_{j+1} : \qquad \mathcal{L}_E(\rho_1 - \rho_2) = \frac{1}{2} \|\mathcal{L}_E(\rho_1 - \rho_2)\|_1 (\rho'_1 - \rho'_2).
\ee
A similar statement holds for the action of the considered pair of superoperators
\be
\begin{split}
\forall \rho_1,\rho_2\in \mathsf N_j: \qquad   \mathcal{W}_\tau (\rho_1 - \rho_2)& = e^{-\kappa\tau}\sum_{\bsmu, \bsnu \in S_{j}} \mathcal{K}_\tau (K_\bsmu M K^\dag_\bsnu)\\
&= e^{-\kappa\tau}\sum_{\bsmu, \bsnu \in S_{j}} c_{\bsmu \bsnu} K_{\bsmu} M K^\dag_{\bsnu} \in \mathsf O_{j},
\end{split}
\ee
where we used the definition $\mathcal W_\tau = e^{-\kappa\tau}\mathcal K_\tau$ and the property in Eq.~\eqref{eq:good_correction}; here $c_{\bsmu \bsnu}$ are also certain coefficients that linearly depend on $b_{\bsmu \bsnu}$. Combining the results of Eqs.~\eqref{eqs:trace_rep_anyO} and \eqref{eqs:tr_rep_Lrhodiff}, we obtain that, for all $\rho_1,\rho_2\in \mathsf N_j$, there exist pairs $\rho_1',\rho_2'\in \mathsf N_{j+1}$ and $\rho_1'',\rho_2''\in \mathsf N_{j+1}$ such that
\be\label{eqs:iter_C_expls}
\begin{split}
 \mathcal W_\tau\mathcal L_E(\rho_1-\rho_2) =  \frac 12\|\mathcal L_E(\rho_1-\rho_2)\|_1 \mathcal K_\tau(\rho'_1-\rho'_2)& =\frac 12\|\mathcal L_E(\rho_1-\rho_2)\|_1 \mathcal W_\tau(\rho'_1-\rho'_2) \\
 &= \frac 14e^{-\kappa \tau}\|\mathcal L_E(\rho_1-\rho_2)\|_1 \|\mathcal K_\tau(\rho'_1-\rho'_2)\|_1 (\rho''_1-\rho''_2). 
 \end{split}
\ee
To bound this parameter, we note that, for all $\rho_1,\rho_2\in \mathsf N_j$, we have $\rho_1-\rho_2\in\mathsf O_j \subset {\rm End}(\mathsf E)\oplus 0$ for all $k\leq\ell$, where $\mathsf E$ is the subspace of correctaqble error states (see definition at the beginning of Sec.~\ref{sec:general}). Therefore, by definition of error-space-restricted contraction norm $\|\cdot\|_{1\to1, \mathsf E}$ in Eq.~\eqref{eq:contr_norm_E_def}, we get
\be\label{eqs:L_contr_trn}
\forall j\leq \ell,\quad  \rho_1,\rho_2\in \mathsf N_j: \quad \|\mathcal L_E(\rho_1-\rho_2)\|_1\leq 2\|\mathcal L_E\|_{1\to1, \mathsf E}.
\ee
Using the definition of the map $\mathcal{K}_t$ below Eq.~\eqref{eq:full_model}, we obtain
\be
\begin{split}
\mathcal{K}_t(\rho_1 - \rho_2) &= \mathcal{R}(\rho_1 - \rho_2) + e^{\kappa t} \left( e^{\mathcal{L}_R t} (\rho_1 - \rho_2) - \mathcal{R}(\rho_1 - \rho_2) \right) \\
&= \rho_1^\infty - \rho_2^\infty + e^{\kappa t} (\rho_1(t) - \rho_1^\infty) - e^{\kappa t} (\rho_2(t) - \rho_2^\infty),
\end{split}
\ee
where, as before, we denote $\rho_m^\infty := \mathcal{R}(\rho_m)$. Using the main assumption in Eq.~\eqref{eq:contraction_lR}, we arrive at the following bound:
\be\label{eqs:K_contr_trn}
\begin{split}
\forall \rho_1, \rho_2 \in \mathsf{N}: \quad \|\mathcal{K}_t(\rho_1 - \rho_2)\|_1 &\leq \|\rho_1^\infty - \rho_2^\infty\|_1 + e^{\kappa t} \|\rho_1(t) - \rho_1^\infty\|_1 + e^{\kappa t} \|\rho_2(t) - \rho_2^\infty\|_1 \\
&\leq \|\rho_1^\infty - \rho_2^\infty\|_1 + 2\chi \leq 2(\chi+1).
\end{split}
\ee
Combining Eq.~\eqref{eqs:iter_C_expls} with Eqs.~\eqref{eqs:L_contr_trn} and \eqref{eqs:K_contr_trn} we get the statement in Eq.~\eqref{eqs:recurrence_KL}.

By applying Eq.~\eqref{eqs:recurrence_KL} sequentially $\ell$ times while 
using values $\tau_j$ for $j = \{1, \dots, \ell\}$ instead of $\tau$, we 
find that there exist $\rho_1',\rho_2'\in \mathsf N_\ell$ such that
\be\label{eq:almost_final_77yh}
\mathcal W_{\tau_\ell}\mathcal L_E \dots \mathcal W_{\tau_1} \mathcal L_E(\rho_1-\rho_2) = \exp\left(-\kappa \sum_{j=1}^k\tau_j\right)C^k (\rho_1'-\rho_2'),
\ee
where $C := (C_0\dots C_{k-1})^{1/k}\leq (\chi+1) \|\mathcal{L}_E\|_{1 \to 1, \mathsf{E}}$. Finally, by applying $\mathcal L_E$ on both sides of Eq.~\eqref{eq:almost_final_77yh}, taking 1-norm, and using Eq.~\eqref{eqs:L_contr_trn}, we arrive at the statement of the Lemma.
\end{proof}
From this Lemma, we conclude that there exist $\rho'_1,\rho'_2\in  {\mathsf N}$ such that
\be
\begin{split}
\mathscr{L}^{-1}\frac 1{s}\bigl\<\sq\sss\underbrace{\x\nnn\x\dots\x\nnn\x}_{\textrm{$\ell+1$}}\bigl\>\geq -\frac 1{\chi+1}\Bigl(\Delta(\chi+1) \|\mathcal L_E\|_{1\to1, \mathsf E}\Bigl)^{\ell+1} \mathscr{L}^{-1}\Biggl\{\frac {1}{s^2(s+\kappa)^{\ell}}\Biggl\},
\end{split}
\ee
where the last term represents the convolution of $\ell$ exponential functions with two identities, which can be expressed using the special functions mentioned in the statement of this Theorem,
\be
\mathscr{L}^{-1}\Biggl\{\frac {1}{s^2(s+\kappa)^{\ell}}\Biggl\} = F_\ell\Bigl(\kappa t\Bigl).
\ee
Finally, combining Eqs.~\eqref{eqs:error_laplace} and \eqref{eqs:final_step_123}, we arrive at 
\be\label{eqs:thm1_final}
\epsilon(t),\delta(t) \leq  1-\min_{\rho_0}\mathscr{L}^{-1}\bigl\<\sq\sss\bigl\> \leq  \frac 1{\chi+1}\left(\frac{(\chi+1)\|\mathcal L_E\|_{1\to1, \mathsf E}\Delta}{\kappa}\right)^{\ell+1}F_\ell\Bigl(\kappa t\Bigl).
\ee
This expression concludes our proof.

\subsection*{Appendix C: Proof of Lemma~\ref{stoch_bound}}
\labelx[Appendix C]{proof_operator_bound}

This appendix contains the proof of Lemma~\ref{stoch_bound}. To express the logical error, we first consider the trajectory decomposition in Eq.~\eqref{eq:poisson_form}. This allows us to derive the bound
\review{\be\label{eqs:bound_theta}
\begin{split}
\epsilon(t) =1-\min_{\rho_0} \sum_{\boldsymbol \mu\in F} p(\boldsymbol \mu,t)=\max_{\rho_0} \sum_{\boldsymbol \mu\in F} p(\boldsymbol \mu,t)\Bigl(1-
\Tr\bigl(\rho_0 \mathcal R \mathcal E_{\boldsymbol \mu} \rho_0\bigl)\Bigl)
\end{split}
\ee
where the minimization/maximization is over pure states $\rho_0$ in the codespace $\mathsf C$.}
\review{Next, we can express any faithful trajectory $\boldsymbol \mu\in G$ as
\be
\mathcal R\mathcal E_{\bsmu}\mathcal R = \mathcal R\mathcal E_{\bsmu_1}\mathcal R\mathcal E_{\bsmu_2}\dots \mathcal E_{\bsmu_m}\mathcal R,
\ee
where $\bsmu_1$ have length at most $\ell$ and contains only errors, i.e. $(\bsmu_k)_i > 0$ and $|\bsmu_k|\leq \ell$. Then
\be\label{eqs:train_decomp33}
\forall \boldsymbol \mu\in G: \quad \mathcal R\mathcal E_{\boldsymbol \mu}\rho_0 = \mathcal R\mathcal E_{\boldsymbol \mu} \mathcal R\rho_0 = \mathcal R\mathcal E_{\boldsymbol \mu}\mathcal R\rho_0 = (\mathcal R\mathcal E_{\bsmu_1}\mathcal R)(\mathcal R\mathcal E_{\bsmu_2}\mathcal R)\dots \mathcal (\mathcal RE_{\bsmu_m}\mathcal R)(\rho_0),
\ee
where $\boldsymbol \mu_m$ are error subsequences.} For any state $\rho$, we have
\be
\mathcal R \mathcal E_{ \boldsymbol \mu_m}\mathcal R(\rho) = \mathcal R(E_{\boldsymbol \mu_m}\mathcal R(\rho)E^\dag_{\boldsymbol \mu_m}) = C_{\boldsymbol \mu_m \boldsymbol \mu_m}\mathcal R(\rho) =  \<0|E^\dag_{\boldsymbol \mu_m}E_{\boldsymbol \mu_m}|0\>\mathcal R(\rho) = \mathcal R(\rho),
\ee
where we used Eq.~\eqref{eq:r-invariance_exact} and the fact that the sequence $\boldsymbol \mu_m$ has weight smaller than the radius $\ell$ and, thus, satisfies Knill-Laflamme condition. Thus, for any faithful trajectory we have
\be
\mathcal R\mathcal E_{\boldsymbol \mu}\rho_0 = \mathcal R(\rho_0) = \rho_0.
\ee
Finally, using the fact that \review{maximum} over $\rho_0$ is taken over a pure states, we get 
\be\label{eqs:sec_term_vanish345}
\forall \boldsymbol \mu\in G: \quad \Tr\bigl(\rho_0  \mathcal R\mathcal E_{\boldsymbol \mu} \rho_0 \bigl) = \Tr \rho_0^2 = 1.
\ee

Combining this expression and the definition of the Heaviside function, we find from Eq.~\eqref{eqs:bound_theta} that
\review{\be
\begin{split}
\epsilon(t) \leq \max_{\rho_0}\sum_{\boldsymbol \mu\in F\setminus G}  p(\boldsymbol \mu,t)\Bigl(1-
\Tr\bigl(\rho_0  \mathcal E_{\boldsymbol \mu} \rho_0\bigl)\Bigl)+ \max_{\rho_0}\sum_{\boldsymbol \mu\in G}  p(\boldsymbol \mu,t)\Bigl(1-
\Tr\bigl(\rho_0  \mathcal E_{\boldsymbol \mu} \rho_0\bigl)\Bigl).
\end{split}
\ee
The last term vanishes due to Eq.~\eqref{eqs:sec_term_vanish345}.
Taking into account the fact that $1-
\Tr\bigl(\rho_0  \mathcal E_{\boldsymbol \mu} \rho_0\bigl)\leq 1$, we get}
\be
\epsilon(t) \leq \sum_{\boldsymbol \mu\in F\setminus G}  p(\boldsymbol \mu,t).
\ee
This concludes our proof.

\subsection*{Appendix D: Proof of Theorem~\ref{lem:all_time1}}
\labelx[Appendix D]{appendix_b}

This appendix contains the definition of the threshold and the proof of Theorem~\ref{lem:all_time1}. 

Recall that, for a Poissonian noise channel, the error jump operators satisfy $\sum_{\mu=1}^N E^{\dag}_{\mu}E_{\mu} = NI$. The Lindbladian can be written in a similar form as Eq.~\eqref{eq:poisson_form}
\begin{align}\label{sm:eq:lindblad}
    \mathcal{L}\rho = \frac{d}{dt}\rho &= \kappa\left(\mathcal{R}(\rho)-\rho\right)+\Delta\sum_{\mu=1}^N\left(E_{\mu}\rho E_{\mu}^{\dag}-\frac{1}{2}E_{\mu}^{\dag}E_{\mu}\rho-\frac{1}{2}\rho E_{\mu}^{\dag}E_{\mu}\right) \nonumber\\
    &= (\kappa+N\Delta)\left(p_0\left(\mathcal{E}_0(\rho)-\rho\right)+p_1\left(\mathcal{E}_1(\rho)-\rho\right)\right),
\end{align}
where the recovery $\mathcal{E}_0=\mathcal{R}$ and the error process $\mathcal{E}_{1}(\cdot)=\frac{1}{N}\sum_{\mu}E_{\mu}(\cdot) E_{\mu}^{\dag}$ have probabilities $p_0 = \frac{\kappa}{\kappa+N\Delta}$ and $p_1 = \frac{N\Delta}{\kappa+N\Delta}$, respectively. Using a Taylor expansion, we get
\begin{equation}
    \exp(\mathcal L t)= e^{-(\kappa+N\Delta) t}\exp\left((\kappa+N\Delta) t\sum_{\mu\in\{0,1\}}p_\mu \mathcal E_\mu \right)= \sum_{k=0}^{\infty} \mathcal{P}(k)\frac{(t(\kappa + N\Delta))^k}{k!}e^{-t(\kappa +N\Delta)},
\end{equation}
where the channel $\mathcal{P}(k)$ is defined as $\mathcal{P}(k) = (p_0\mathcal{E}_0+p_1\mathcal{E}_1)^k=\sum_{\{\textbf{a}\}} p_{a_1}p_{a_2}\cdots p_{a_k}\mathcal{E}_{a_1}\mathcal{E}_{a_2}\cdots\mathcal{E}_{a_k}$ and $\{\bf a\}$ is a set of all possible binary (jump) sequences of length $k$ with $a_i\in\{0, 1\}$ labeling the recovery and error channels. This is an alternative form of Eq.~\eqref{eq:poisson_form}. Using the definition of the error measure $\epsilon(t)$ [Eq.~\eqref{eq:overlap_log_error}], we find that, for an initial pure state $\rho_0$ in the codespace, we have
\begin{equation}
     \epsilon(t) =\sum_{k=0}^\infty p_e(k) \frac{(t(\kappa + N\Delta))^k}{k!}e^{-t(\kappa +N\Delta)}, \label{eq:PnC}
\end{equation}
where $p_e(k) = 1-\min_{\rho_0}\sum_{\{\textbf{a}\}} p_{a_1}p_{a_2}\cdots p_{a_k}\Tr[\rho_0\mathcal{R}\mathcal{E}_{a_1}\mathcal{E}_{a_2}\cdots\mathcal{E}_{a_k}\rho_0]$.
For integers $k\geq 1,h\geq 1$, a sequence $\textbf{a} =(a_1,\dots,a_k)$ has an error weight of at most $h$ if it contains no more than $h$ consecutive error jumps. Let $\{\textbf{a}_h\}$ be the set of all jump sequences of length $k$ with an error weight of at most $h$. Since $\{\textbf{a}_h\}\subseteq \{\textbf{a}\}$, it follows that
\begin{equation}\label{sm:eq:traj_prob}
\min_{\rho_0}\sum_{\{\textbf{a}_h\}} p_{a_1}p_{a_2}\cdots p_{a_k}\Tr[\rho_0\mathcal{R}\mathcal{E}_{a_1}\mathcal{E}_{a_2}\cdots\mathcal{E}_{a_k}\rho_0]\leq
    \min_{\rho_0}\sum_{\{\textbf{a}\}} p_{a_1}p_{a_2}\cdots p_{a_k}\Tr[\rho_0\mathcal{R}\mathcal{E}_{a_1}\mathcal{E}_{a_2}\cdots\mathcal{E}_{a_k}\rho_0].
\end{equation}
\review{where here and below we assume that we are minimizing over $\rho_0 = |\psi_0\>\<\psi_0|$, which are pure states in the codespace, $|\psi_0\>\in \mathsf C$.}.
To establish an upper bound on $\epsilon(t)$, let $h$ be a tolerable error weight for the code satisfying Eq.~\eqref{sm:eq:dfn_threshold}. 
Let $z \ge 0$ denote the number of zeros in the sequence $\textbf{a}_h = (a_1,a_2,\cdots,a_k)$. Then
\begin{equation}  \mathcal{E}_{a_1}\mathcal{E}_{a_2}\cdots\mathcal{E}_{a_k}=
\mathcal{E}_1^{m_0}\mathcal{R}\mathcal{E}_1^{m_1}\cdots \mathcal{R}\mathcal{E}_1^{m_{z}}, \label{sm:eq:induction}
\end{equation}
where $0 \le m_i \le h$ for $i = 0, \dots, z$ and $z+\sum_{i=0}^zm_i =k$.
Let us rewrite
\begin{align}
\Tr[\rho_0\mathcal{R}\mathcal{E}_1^{m_0}\cdots \mathcal{R}\mathcal{E}_1^{m_{z-1}}\mathcal{R}\mathcal{E}_1^{m_{z}}\rho_0]
&=\Tr[\rho_0\mathcal{R}\mathcal{E}_1^{m_0}\cdots \mathcal{R}\mathcal{E}_1^{m_{z-1}}\left(\mathcal{R}\mathcal{E}_1^{m_{z}}\rho_0-(1-\xi)\rho_0\right)]
\nonumber \\
&+(1-\xi)\Tr[\rho_0\mathcal{R}\mathcal{E}_1^{m_0}\cdots \mathcal{R}\mathcal{E}_1^{m_{z-1}}\rho_0 ]. \label{eq:xi}
\end{align}
By the definition of threshold, it follows from Eq.~\eqref{sm:eq:dfn_threshold} that the first term on the right-hand side of Eq.~(\ref{eq:xi}) is non-negative. Therefore, if $m_z>0$, then
\be
\Tr[\rho_0\mathcal{R}\mathcal{E}_1^{m_0}\cdots \mathcal{R}\mathcal{E}_1^{m_{z-1}}\mathcal{R}\mathcal{E}_1^{m_{z}}\rho_0] \geq
(1-\xi)\Tr[\rho_0\mathcal{R}\mathcal{E}_1^{m_0}\cdots \mathcal{R}\mathcal{E}_1^{m_{z-1}}\rho_0 ].
\ee
If $m_z = 0$, we have instead
\be
\Tr[\rho_0\mathcal{R}\mathcal{E}_1^{m_0}\cdots \mathcal{R}\mathcal{E}_1^{m_{z-1}}\mathcal{R}\mathcal{E}_1^{m_{z}}\rho_0] = \Tr[\rho_0\mathcal{R}\mathcal{E}_1^{m_0}\cdots \mathcal{R}\mathcal{E}_1^{m_{z-1}}\rho_0 ].
\ee
Inductively, we know that after $z+1$ steps, the exponent of the factor $(1-\xi)$ will be equal to the number of non-zero $m_i$'s in the jump sequence. We therefore find that
\begin{equation} \Tr[\rho_0\mathcal{R}\mathcal{E}_1^{m_0}\cdots \mathcal{R}\mathcal{E}_1^{m_{z-1}}\mathcal{R}\mathcal{E}_1^{m_{z}}\rho_0]
\geq (1-\xi)^{(k+1)/2} \geq (1-\xi)^k,
\label{sm:eq:recover_prob}
\end{equation}
where we obtained an upper bound on the number of non-zero $m_i$'s by setting $m_i = 1$ for all $i$ in the relation $z+\sum_{i=0}^zm_i =k$. The number of non-zero $m_i$'s is thus less than or equal to $(k+1)/2$.

Now we are ready to prove the following Lemma:
\begin{lem}\label{sm:lem1}
Assume an $n$-qubit code family with increasing $n$ and a tolerable error weight $h = h(n)$ with respect to the error channel $\mathcal E(\cdot) = \frac{1}{N} \sum_\mu E_\mu(\cdot) E_\mu^\dag$. Then there exists $\xi = 2^{-\Omega(d)}$ such that
\begin{equation}\label{eq:ineq}
     p_e(k)\leq 1-
   \left[(1-\xi)\left(1-p_1^{h+1}\right)\right]^k,
\end{equation}
where $p_1 = \frac{N\Delta}{\kappa+N\Delta}$ and $k\geq 0$.
\end{lem}

\begin{proof}

The equality trivially holds for $k = 0$. We therefore consider the case $k\geq 1$.
To establish an upper bound on this probability, we use Eq.~\eqref{sm:eq:traj_prob} and Eq.~\eqref{sm:eq:recover_prob} to obtain
\begin{equation}\label{eq:target}
     p_{e}(k)\leq 1- \min_{\rho_0}\sum_{\{\textbf{a}_h\}} p_{a_1}p_{a_2}\cdots p_{a_k}\Tr[\rho_0\mathcal{R}\mathcal{E}_{a_1}\mathcal{E}_{a_2}\cdots\mathcal{E}_{a_k}\rho_0]\leq 1- (1-\xi)^{k}\sum_{\{\textbf{a}_h\}} p_{a_1}p_{a_2}\cdots p_{a_k}.
\end{equation}
We want to show that
\begin{equation}\label{eq:induct0}
   \left(1-p_1^{h+1}\right)^k\leq \sum_{\{\textbf{a}_h\}} p_{a_1}p_{a_2}\cdots p_{a_k}.
\end{equation}
To show this, we note that, for a given jump sequence $\textbf{a}_h$, its probability in Eq.~\eqref{eq:induct0} takes the form
\begin{equation}
     p_{a_1}p_{a_2}\cdots p_{a_k}=p_1^{m_0}p_0 p_1^{m_2}p_0
     \cdots
     p_1^{m_{z}},
\end{equation}
where $z+\sum_{i=0}^{z}m_i = k$ and each integer satisfies $0\leq m_i\leq h$. Each jump sequence $\textbf{a}_h$ is uniquely labelled by $\textbf{m}(\textbf{a}_h) = (m_0,m_1,\cdots,m_{z})$. 
We can establish the desired lower bound by noting that, when $p_0\neq 0$,
\begin{align}
(1- p_1^{h+1})^k &= 
\left[\sum_{i=0}^{h}p_1^{i}p_0\right]^k = \sum_{i_1 = 0}^{h}\sum_{i_2 = 0}^{h}\cdots \sum_{i_k = 0}^{h} p_1^{i_1}p_0p_1^{i_2}p_0\cdots p_1^{i_k}p_0
\nonumber \\
&= 
\sum_{z = 0}^k  \sum_{\mathbf{m}}
    \underbrace{p_1^{m_0}p_0 p_1^{m_1}p_0
     \cdots
    p_1^{m_{z}}}_\textrm{first $k$ factors}
     \times \left[\sum_{i'=0}^{h-m_z}p_1^{i'}p_0  \right]\times\left[\sum_{i=0}^{h}p_1^{i}p_0\right]^{k-z-1}\nonumber
     \\
     &\leq
     \sum_{z = 0}^k \sum_{\mathbf{m}}
    p_1^{m_1}p_0 p_1^{m_2}p_0
     \cdots
    p_1^{m_{z}}
     \nonumber \\
     &=
     \sum_{\{\textbf{a}_h\}} p_{a_1}p_{a_2}\cdots p_{a_k}.
     \label{sm:eq:final_lem2}
\end{align}
In the second line, the sum is over $\mathbf{m} = (m_0,m_1,\cdots,m_{z})$, where $0 \leq m_i \leq h$ and $z+\sum_{i=0}^{z}m_i = k$. 
In going from the first line to the second line, we re-write the sum in the first line using the following steps: we fix the first $k$ factors in each term and sum over the rest of the possible factors. Then we take the sum over all possible first $k$ factors (i.e.~$\sum_{z = 0}^k \sum_{\mathbf{m}}$). The inequality in the third line follows from $\left[\sum_{i'=0}^{h-m_z}p_1^{i'}p_0  \right]\left[\sum_{i=0}^{h}p_1^{i}p_0\right]^{k-z-1}\leq 1$, which trivially holds for $z\leq k-1$. When $z = k$ and $p_0\neq 0$, we have $m_i = 0$ for all $i$ due to the constraint $z+\sum_{i=0}^{z}m_i = k$. Therefore, the two factors cancel, and the inequality becomes an equality. When $p_0 = 0$, we can directly verify that Eq.~\eqref{sm:eq:final_lem2} also holds. 
This concludes the proof of the lemma.

\end{proof}

If we substitute Eq.\ (\ref{eq:ineq}) into Eq.~\eqref{eq:PnC}, we get
\begin{align}
    \epsilon(t) &= \sum_k p_{e}(k) \frac{[t(\kappa + N\Delta)]^k}{k!}e^{-t(\kappa +N\Delta)} \nonumber \\
    &\leq  \sum_k\left( 1- \left[(1-\xi)\left(1-\left(\frac{N\Delta}{\kappa+N\Delta}\right)^{h+1}\right)\right]^k\right)\frac{[t(\kappa + N\Delta)]^k}{k!}e^{-t(\kappa +N\Delta)} \nonumber \\
    &= 1-\exp\left(-(1-\xi)(\kappa+N\Delta)\left(\frac{N\Delta}{\kappa+N\Delta}\right)^{h+1}t-\xi(\kappa+N\Delta)t\right) \nonumber\\
    &= 1-\exp\left(-N\Delta(1-\xi)\left(\frac{N\Delta}{\kappa+N\Delta}\right)^{h}t-\xi(\kappa+N\Delta)t\right),
    \label{eq:bound}
\end{align}
which proves Theorem~\ref{lem:all_time1}.
\\
\\
We would like to conclude this Appendix with a side remark on the definition of the tolerable error weight and the threshold. In Definition~\ref{sm:def:threshold}, instead of Eq.~\eqref{sm:eq:dfn_threshold}, we can also consider a natural alternative condition on the measure of fidelity:
\begin{equation}\label{sm:eq:alternative_threshold}
    \mathcal \Tr\left[\rho \mathcal{R}\mathcal E^m \rho\right]\geq 1-\xi,\quad\quad \rho = \ket{\psi}\bra{\psi},
\end{equation}
for $\xi\in [0,1]$ and $0\leq m \leq h$. This condition is necessary but not sufficient for Eq.~\eqref{sm:eq:dfn_threshold} to hold. For instance, a pure logical state $\ket{\psi'}$ with a small logical rotation from the original logical state $\ket{\psi}$ can satisfy Eq.~\eqref{sm:eq:alternative_threshold} but will fail to satisfy Eq.~\eqref{sm:eq:dfn_threshold}. Note that $  \Tr\left[\rho \mathcal{R}\mathcal E^m \rho\right] =  F(\rho,\mathcal{R}\mathcal{E}^{m}\rho)^2$, where the fidelity between two quantum states $\rho,\sigma$ is defined as $F(\rho,\sigma) = \Tr(\sqrt{\rho^{1/2}\sigma\rho^{1/2}})$. For any quantum state $\gamma$, we have~ \cite{nielsen_quantum_2010}
\begin{equation}\label{sm:eq:tria_ineq}
    F(\rho,\sigma)\geq F(\rho,\gamma)F(\gamma,\sigma)-\sqrt{1-F(\rho,\gamma)^2}\sqrt{1-F(\gamma,\sigma)^2}.
\end{equation}
The fidelity measure also satisfies joint concavity~ \cite{nielsen_quantum_2010}
\begin{equation}
    F\left(\sum_i p_i\rho_i, \sum_i p_i\sigma_i\right)\geq \sum_i p_i F(\rho_i, \sigma_i),
\end{equation}
where $p_i$'s form a probability distribution over quantum states $\rho_i, \sigma_i$. For a generic mixed state in the codespace $\rho = \alpha|\psi\>\<\psi|+(1-\alpha)|\phi\>\<\phi|$, with $\alpha\in [0,1]$ and $|\psi\>,|\phi\>\in\mathsf C$, it follows from the joint concavity that
\begin{equation}\label{sm:eq:xi_gen}   F(\rho,\mathcal{R}\mathcal{E}_1^{m}\rho)\geq \alpha F(\rho_{\psi},\mathcal{R}\mathcal{E}_1^{m}\rho_{\psi})+(1-\alpha) F(\rho_{\phi},\mathcal{R}\mathcal{E}_1^{m}\rho_{\phi})\geq \sqrt{1-\xi},
\end{equation}
where $\rho_{\psi} = |\psi\>\<\psi|$ and $\rho_{\phi} = |\phi\>\<\phi|$. The last line follows once the integer $m$ satisfies $0\leq m \leq h$. For convenience, let us denote $\mathcal{M}_{z}=\prod_{i=0}^{z}\mathcal{R}\mathcal{E}_1^{m_i}$. Using  inequality~\eqref{sm:eq:tria_ineq} and the fact that $\sqrt{1-F(\rho_0,\mathcal{M}_{z}\rho_0)^2}\leq 1$, we can deduce that, if $m_z > 0$,
\begin{align}\label{sm:eq:finduction}
    F(\rho_0, \mathcal{M}_{z}\rho_0) &\geq F(\rho_0,\mathcal{M}_{z-1}\rho_0)F(\mathcal{M}_{z-1}\rho_0,\mathcal{M}_{z}\rho_0)-\sqrt{1-F(\mathcal{M}_{z-1}\rho_0,\mathcal{M}_{z}\rho_0)^2}
    \nonumber\\
    &\geq
    \sqrt{1-\xi}F(\rho_0,\mathcal{M}_{z-1}\rho_0)-\sqrt{\xi},
\end{align}
where the last line follows from $F(\mathcal{M}_{z-1}\rho_0,\mathcal{M}_{z}\rho_0)\geq \sqrt{1-\xi}$ when setting $\rho = \mathcal{M}_{z-1}\rho_0$ in Eq.~\eqref{sm:eq:xi_gen}. We also have $ F(\rho_0, \mathcal{M}_{z}\rho_0) = F(\rho_0,\mathcal{M}_{z-1}\rho_0)$ if $m_z = 0$. 
From Eq.~\eqref{sm:eq:finduction}, we note that
\begin{equation}
    \sqrt{1-\xi}F(\rho_0,\mathcal{M}_{z-1}\rho_0)-\sqrt{\xi}\geq \sqrt{1-\xi^{1/4}}F(\rho_0,\mathcal{M}_{z-1}\rho_0)-\sqrt{\xi}
\end{equation}
for $0\leq \xi\leq 1$. Applying the inequality inductively, similarly to the inductive derivation of Eq.~\eqref{sm:eq:recover_prob}, we have
\begin{equation}
     F(\rho_0, \mathcal{M}_{z}\rho_0) \geq \left(1-\xi^{1/4}\right)^{k/2}-\sqrt{\xi}\frac{1-\left(1-\xi^{1/4}\right)^{k/2}}{1-\sqrt{1-\xi^{1/4}}}
     \geq \left(1-\xi^{1/4}\right)^{k/2}
     - 2\xi^{1/4}.
\end{equation}
The second inequality follows from the use of Bernoulli's inequality $\sqrt{1-\xi^{1/4}}\leq 1+\xi^{1/4}/2$.
Since $ F(\rho_0, \mathcal{M}_{z}\rho_0)\geq 0$, this implies
\begin{equation}\label{sm:eq:fidelity_bound}
     F(\rho_0, \mathcal{M}_{z}\rho_0)^2 \geq \left(1-\xi^{1/4}\right)^{k}-4\xi^{1/4}.
\end{equation}
To bound Eq.~\eqref{eq:PnC}, we note that 
\begin{align}
p_e(k) &= 1-\min_{\rho_0}\sum_{\{\textbf{a}\}} p_{a_1}p_{a_2}\cdots p_{a_k}\Tr[\rho_0\mathcal{R}\mathcal{E}_{a_1}\mathcal{E}_{a_2}\cdots\mathcal{E}_{a_k}\rho_0] 
\nonumber \\
&=  1-\min_{\rho_0}\sum_{\{\textbf{a}\}} p_{a_1}p_{a_2}\cdots p_{a_k} F(\rho_0, \mathcal{M}_{z}\rho_0)^2 
\nonumber \\
& \leq 1+4\xi^{1/4}-\left(1-\xi^{1/4}\right)^k\sum_{\{\textbf{a}\}} p_{a_1}p_{a_2}\cdots p_{a_k} 
\nonumber \\
& \leq 1+4\xi^{1/4}-\left(1-\xi^{1/4}\right)^k\left(1-p_1^{h+1}\right)^k,
\end{align}
where $z$ denotes the number of 0's in the jump sequence $\textbf{a}$, and we used Eq.~\eqref{eq:induct0} in the last line.
Substituting this into Eq.~\eqref{eq:PnC}, we get a modified Theorem 2 based on this alternative definition of the tolerable error weight and the threshold:
\begin{equation}
    \epsilon(t)\leq  1+4\xi^{1/4}-\exp\left(-N\Delta(1-\xi^{1/4})\left(\frac{N\Delta}{\kappa+N\Delta}\right)^{h}t-\xi^{1/4}(\kappa+N\Delta)t\right).
\end{equation}
Note that $\xi$, and hence the contribution $\xi^{1/4}$, is exponentially small in the distance of the quantum codes. For codes with a large enough distance, the bound yields qualitatively the same scaling of the error rate as Eq.~\eqref{eq:bound} as the code distance increases.
\subsection*{Appendix E: Proof of Theorem~\ref{upper_bound_sl1}}
\labelx[Appendix E]{appendix_c}

This appendix contains the necessary assumptions and the proof of Theorem~\ref{upper_bound_sl1}. In particular, we consider a noise model where the error jump operators satisfy $E_\mu^2 = I$ and $E_\mu E_{\mu'} = \pm E_{\mu'}E_{\mu}$ for all $\mu,\mu'$. This noise model is a special case of a Poisonnian noise model. We will assume exactly the same setup as in the previous appendix, i.e.\ Eqs.~\eqref{sm:eq:lindblad} and~\eqref{eq:PnC}. We prove an error bound for the class of $n$-qubit error-correcting codes $\mathsf C$ that have a varying distance $d = d(n)$ and a tolerable error weight $h(n)$ tailored for this particular noise model (see Eq.~\eqref{sm:eq:xi_pauli} for a precise definition).

Let $\mathcal{Q}_0 = \mathcal{I}$ be the identity channel. For $k\geq 1$, we also define $\mathcal{Q}_k(\cdot) = \frac{1}{|\textbf{S}_k|}\sum_{\{\mu\}_k\in S_k} E_{\{\mu\}_k}(\cdot) E_{\{\mu\}_k}^{\dag}$, where $E_{\{\mu\}}=\prod_{\mu\in\{\mu\}}E_\mu$ and the set $\textbf{S}_k$ is the set of all possible distinct $k$ error indices. Note that $\mathcal{Q}_1 = \mathcal{E}_1$, where $\mathcal{E}_1(\cdot) = \frac{1}{N}\sum_{\mu}E_{\mu}(\cdot)E_{\mu}^{\dag}$ is the total error channel defined in the Lindbladian in Eq.~\eqref{sm:eq:lindblad} and $N$ is the total number of error jumps in the Lindbladian.

Let us consider the channel $\mathcal{Q}_1^k$. The product expands into a convex combination of terms that contain a different number of $E_{\mu}$'s.  Since the channel is symmetric under permutation of the error indices, all the terms with the same number of $E_{\mu}$'s will have the same coefficient. In particular, this implies that, for any integer $k\geq 0$,
\begin{equation}\label{sm:eq:phys_err_decomp}
    \mathcal{Q}_1^k = \sum_{i = 0}^k r_{ki}\mathcal{Q}_i,
\end{equation}
where $r_{ki}\geq 0$ and $\sum_i r_{ki} = 1$.

We consider the class of $n$-qubit error-correcting codes $\mathsf C$ with a varying distance $d = d(n)$ that satisfy the following: There exists an integer-valued function $h = h(n)>0$ such that for $k\leq h$ and any $|\psi\>\in\mathsf C$, the following holds:
\begin{equation} \label{sm:eq:xi_pauli} \mathcal{R}\mathcal{Q}_k|\psi\>\<\psi|-(1-\xi)|\psi\>\<\psi|\geq 0,
\end{equation}
where $\xi = 2^{-\Omega(d)}$ and is independent of $|\psi\>$. \review{Such code family is said to have a tolerable error weight with uniform Pauli noise}. It follows that
\begin{equation}
    \mathcal{R}\mathcal{Q}_1^k|\psi\>\<\psi| = \sum_{i = 0}^k r_{ki} \mathcal{R}\mathcal{Q}_i|\psi\>\<\psi|\geq (1-\xi)|\psi\>\<\psi|.
\end{equation}
Note that $h(n)$ is a tolerable error weight for $\mathsf{C}$. For example, if we let $h = \ell$, where $\ell$ is the error radius satisfying $d = 2\ell+1$, then we have $\xi = 0$. Moreover, if there exists a constant $f>0$ such that $nf\leq h$ for all $n$, then for $k\leq nf$, the code we consider has a threshold according to Definition~\ref{sm:def:threshold}. 
This set of codes is not very restricted and contains commonly known examples, e.g. the quantum stabilizer codes subject to single-qubit Pauli noise.

We now proceed to prove Theorem~\ref{upper_bound_sl1}. 
We first prove the following lemma:
\begin{lem}
Assume an $n$-qubit code family with increasing $n$ and a tolerable error weight $h = h(n)$ with respect to the error channel defined above. Then $p_e(k)$ in Eq.~\eqref{eq:PnC} satisfies
\begin{equation}\label{eq:ineq2}
     p_e(k)\leq 1-\left[(1-\xi)\left(1- p_1s_1\right)\right]^k,
\end{equation}
where $\xi = 2^{-\Omega(d)}$ and $s_1$ is the solution to the recurrence relation
\begin{equation}
     s_v = \frac{v}{N}p_1s_{v-1} + \left(1-\frac{v}{N}\right)p_1s_{v+1},\ s_0 = 0,\ s_{h+1} = 1,
\end{equation}
with $p_1 = \frac{N\Delta}{\kappa+N\Delta}$.
\end{lem}
\begin{proof}
Our goal is to lower-bound the contribution
$
\min_{\rho_0}\sum_{\{\textbf{a}\}}
    p_{a_1}p_{a_2}\cdots p_{a_k}\Tr[\rho_0\mathcal{R}\mathcal{E}_{a_1}\mathcal{E}_{a_2}\cdots\mathcal{E}_{a_k}\rho_0]$ contained in $p_e(k)$ (defined below Eq.~\eqref{eq:PnC}), where $\{\textbf{a}\}$ is the set of length-$k$ trajectories $\textbf{a} = (a_1,a_2,\cdots,a_k)$ with the labels $a_i\in\{0,1\}$.
For a given jump sequence $\textbf{a}$, we can write
\begin{align}
 p_{a_1}p_{a_2}\cdots p_{a_k}\Tr[\rho_0\mathcal{R}\mathcal{E}_{a_1}\mathcal{E}_{a_2}\cdots\mathcal{E}_{a_k}\rho_0]
 &=
(p_1^{m_0})(p_0p_1^{m_1})\cdots (p_0p_1^{m_{z}})\Tr[\rho_0\mathcal{R}\mathcal{E}_1^{m_0}\cdots \mathcal{R}\mathcal{E}_1^{m_{z}}\rho_0],
\end{align}
where $z$ denotes the number of $0$'s in the sequence $\textbf{a} = (a_1,\cdots, a_k)$ and $z+\sum_{i = 0}^{z}m_i=k$. Note that the sequence $\textbf{a}$ is also uniquely labelled by $\textbf{m}(\textbf{a}) = (m_0,m_1,\cdots,m_z)$. We use Eq.~\eqref{sm:eq:phys_err_decomp} to decompose the error subsequence
\begin{align}
&(p_1^{m_0})(p_0p_1^{m_1})\cdots (p_0p_1^{m_{z}})\Tr[\rho_0\mathcal{R}\mathcal{E}_1^{m_0}\cdots \mathcal{R}\mathcal{E}_1^{m_{z}}\rho_0]
 \nonumber \\
 &=
 \left(p_1^{m_0}\sum_{i_0=0}^{m_0}r_{m_0i_0}\right)\left(p_0p_1^{m_1}\sum_{i_1=0}^{m_1}r_{m_1i_1}\right)\cdots
 \left(p_0p_1^{m_z}\sum_{i_z=0}^{m_z}r_{m_zi_z}\right)
 \Tr\left[\rho_0\mathcal{R}\mathcal{Q}_{i_0}\mathcal{R}\mathcal{Q}_{i_1}\cdots \mathcal{R}\mathcal{Q}_{i_z}\rho_0\right]
 \nonumber\\
 &\geq
 (1-\xi)^k\left(p_1^{m_0}\sum_{i_0=0}^{\min\{h,m_0\}}r_{m_0i_0}\right)\left(p_0p_1^{m_1}\sum_{i_1=0}^{\min\{h,m_1\}}r_{m_1i_1}\right)\cdots
 \left(p_0p_1^{m_z}\sum_{i_z=0}^{\min\{h,m_z\}}r_{m_zi_z}\right)
 \nonumber \\
 &=
 (1-\xi)^k\sum_{i_0=0}^{\min\{h,m_0\}}\sum_{i_1=0}^{\min\{h,m_1\}}\cdots
\sum_{i_z=0}^{\min\{h,m_z\}} \left(p_1^{m_0}r_{m_0i_0}\right)\left(p_0p_1^{m_1}r_{m_1i_1}\right)\cdots
 \left(p_0p_1^{m_z}r_{m_zi_z}\right).
 \label{sm:eq:path_seq2}
\end{align}
In the second line, we restricted the sums and used the fact that
\begin{equation}
    \Tr\left[\rho_0\mathcal{R}\mathcal{Q}_{i_0}\mathcal{R}\mathcal{Q}_{i_1}\cdots \mathcal{R}\mathcal{Q}_{i_z}\rho_0\right]\geq (1-\xi)^k,
\end{equation}
which follows from Eq.~\eqref{sm:eq:xi_pauli} using the same reasoning as the one used to derive Eq.~\eqref{sm:eq:recover_prob}.

To proceed, we note that the decomposition in Eq.~\eqref{sm:eq:phys_err_decomp} can be interpreted as a random-walk process that either creates or annihilates an error at each step. Each jump randomly applies one of the $N$ different errors. The same error cancels with itself if it is triggered an even number of times. Let $v$ be the number of physical errors at the current configuration (i.e.~$E_{\{\mu\}}\ket{\psi}\bra{\psi}E_{\{\mu\}}^{\dag}$ with a subset of error indices satisfying $|\{\mu\}|=v$ and some fixed reference logical state $\ket{\psi}$). The next jump has a probability of $1-v/N$ to create a physical error or a probability of $v/N$ to cancel a physical error. At each step, the number of physical errors will be updated accordingly.
For example,
\begin{equation}
    \mathcal{Q}_1\mathcal{Q}_1 = 1\cdot\frac{1}{N}\cdot\mathcal{Q}_0+1\cdot\frac{N-1}{N}\cdot\mathcal{Q}_2,
\end{equation}
where $\mathcal{Q}_2(\cdot) = \frac{1}{N(N-1)}\sum_\mu\sum_{\mu'\neq \mu}E_\mu E_{\mu'}(\cdot) E_{\mu'}^{\dag}E_\mu^{\dag}$. That is, when applying $\mathcal{Q}_1$ twice, there are two possible paths for the errors: (i) The first jump creates an error, and the second jump cancels the created error; (ii) The first jump creates an error, and the second jump creates another error. Similarly, we have $\mathcal{Q}_2\mathcal{Q}_1 = \frac{2}{N}\mathcal{Q}_1+\frac{N-2}{N}\mathcal{Q}_3$. This leads to three paths:
\begin{equation}\label{sm:eq:e3}
    \mathcal{Q}_1\mathcal{Q}_1\mathcal{Q}_1=1\cdot\frac{1}{N}\cdot 1\cdot\mathcal{Q}_1 + 1\cdot\frac{N-1}{N}\cdot\frac{2}{N}\cdot\mathcal{Q}_1+1\cdot\frac{N-1}{N}\cdot\frac{N-2}{N}\cdot\mathcal{Q}_3.
\end{equation}
The three terms correspond to the three possible paths when applying $\mathcal{Q}_1$ three times. This suggests that each coefficient $r_{ki}$ in Eq.~\eqref{sm:eq:phys_err_decomp} takes the form
\begin{equation}\label{sm:eq:path_decomp}
    r_{mi} = \sum _{\{\textbf{v}(m,i)\}}g(v_0,v_1)\cdots g(v_{m-1},v_m),
\end{equation}
where $\{\textbf{v}(m,i)\}$ is the set of all paths that lead to $i$ physical errors when $\mathcal{Q}_1$ is applied $m$ times. Here $\textbf{v}(m,i) = (v_0,v_1,v_2,\cdots, v_m)$, where $v_j$ is the number of physical errors after the $j$-th jump. By definition, $v_0 = 0$ and $v_m = i$. The coefficient $g(v_{j-1},v_{j})$ is the probability that $v_{j}$ errors are present after the $j$-th jump, conditioned on the presence of $v_{j-1}$ errors after the $(j-1)$-th jump. Explicitly, we have, for $j\geq 1$,
\begin{equation}
    g(v_{j-1},v_j) =
    \begin{cases}
        \frac{v_{j-1}}{N}\quad\text{if }v_{j}-v_{j-1}=-1,
        \\
         1-\frac{v_{j-1}}{N}\quad\text{if }v_{j}-v_{j-1}=1.
    \end{cases}
\end{equation}

The goal is to lower-bound the contribution
$
\min_{\rho_0}\sum_{\{\textbf{a}\}}
    p_{a_1}p_{a_2}\cdots p_{a_k}\Tr[\rho_0\mathcal{R}\mathcal{E}_{a_1}\mathcal{E}_{a_2}\cdots\mathcal{E}_{a_k}\rho_0]$ from Eq~\eqref{eq:PnC}. Summing over all possible length-$k$ trajectories $\{\textbf{a}\}$ is the same as summing over $\{\textbf{m}(\textbf{a})\}$, which labels each trajectory in $\{\textbf{a}\}$ uniquely by the integer sequence $(m_0,\cdots,m_z)$. Summing Eq.~\eqref{sm:eq:path_seq2} over $\{\textbf{m}(\textbf{a})\}$ and  using Eq.~\eqref{sm:eq:path_decomp} leads to
\begin{align}
    &\sum_{\{\textbf{m}(\textbf{a})\}}\sum_{i_0=0}^{\min\{h,m_0\}}\sum_{i_1=0}^{\min\{h,m_1\}}\cdots
\sum_{i_z=0}^{\min\{h,m_z\}}
 \left(p_1^{m_0}r_{m_0i_0}\right)\left(p_0p_1^{m_1}r_{m_1i_1}\right)\cdots
 \left(p_0p_1^{m_z}r_{m_zi_z}\right)
 \nonumber \\
&=\sum_{\{\textbf{m}(\textbf{a})\}}\sum_{i_0=0}^{\min\{h,m_0\}}\sum_{i_1=0}^{\min\{h,m_1\}}\cdots
\sum_{i_z=0}^{\min\{h,m_z\}}
\sum_{\{\textbf{v}(m_0,i_0)\}}\sum_{\{\textbf{v}(m_1,i_1)\}}\cdots \sum_{\{\textbf{v}(m_z,i_z)\}}
\nonumber\\
&\quad \quad
\left(\prod_{j=1}^{m_0}p_1g(v_{j-1},v_j)\right)p_0\left(\prod_{j=1}^{m_1}p_1g(v_{j-1},v_j)\right)p_0\cdots
\left(\prod_{j=1}^{m_z}p_1g(v_{j-1},v_j)\right)
\nonumber\\
&\geq \sum_{\{\textbf{m}(\textbf{a})\}}\sum_{i_0=0}^{\min\{h,m_0\}}\sum_{i_1=0}^{\min\{h,m_1\}}\cdots
\sum_{i_z=0}^{\min\{h,m_z\}}
\sum_{\{\textbf{v}_h(m_0,i_0)\}}\sum_{\{\textbf{v}_h(m_1,i_1)\}}\cdots \sum_{\{\textbf{v}_h(m_z,i_z)\}}
\nonumber\\
&\quad \quad
\left(\prod_{j=1}^{m_0}p_1g(v_{j-1},v_j)\right)p_0\left(\prod_{j=1}^{m_1}p_1g(v_{j-1},v_j)\right)p_0\cdots
\left(\prod_{j=1}^{m_z}p_1g(v_{j-1},v_j)\right),
\label{sm:eq:random_walk}
\end{align}
where in the last inequality we restrict the sum to the subset $\{\textbf{v}_h(m,i)\} \subseteq \{\textbf{v}(m,i)\}$ of paths along which the number of physical errors $v$ always satisfies $v\leq h$.
Each term in the sum is a product of $k$ jump probabilities for a random-walk trajectory specified by the sequence $(\textbf{v}_h(m_0,i_0),\cdots,\textbf{v}_h(m_z,i_z))$. To proceed, let us now define a random walk according to the rules in Eq.~\eqref{sm:eq:path_decomp}: Given a configuration with $v$ physical errors, the random walk updates with one of the three stochastic jumps:
\begin{enumerate}
    \item With probability $p_1(1-\frac{v}{N})$, $v\to v+1$.
    \item With probability $p_1\frac{v}{N}$, $v\to v-1$.
    \item With probability $p_0$, the process terminates and returns success.
\end{enumerate}
After the first jump of the walk, two more termination checks (note that these are not counted as jumps) are done before each future stochastic jump is taken:
\begin{enumerate}
    \item[4.] If $v = 0$, the process terminates and returns success.
    \item[5.]  If $v = h+1$, the process terminates and returns failure. 
\end{enumerate}
This random walk is always initialized in a configuration with zero physical errors, and the walk terminates with probability one (non-terminating trajectories have an infinite number of steps and therefore a probability of zero). To establish a lower bound for Eq.~\eqref{sm:eq:random_walk}, we consider a stochastic process $\Phi(k)$ of $k$ consecutive independent random-walk processes defined above. The stochastic process $\Phi(k)$ returns success if all $k$ random walks return success. By definition, $\Phi(k)$ contains at least $k$ jumps before it is terminated.
Furthermore, we note that the set of all possible first $k$ jumps of a successful $\Phi(k)$ is the same as the set of trajectories in Eq.~\eqref{sm:eq:random_walk}.

Namely, any trajectory specified by $(\textbf{v}_h(m_0,i_0),\cdots,\textbf{v}_h(m_z,i_z))$ in Eq.~\eqref{sm:eq:random_walk} is a valid trajectory for the first $k$ jumps in a successful $\Phi(k)$, and any possible set of first $k$ jumps taken by the process can be specified by some $(\textbf{v}_h(m_0,i_0),\cdots,\textbf{v}_h(m_z,i_z))$. 
To see this, suppose we have a trajectory for the first $k$ jumps for a successful $\Phi(k)$. Since no more than $h$ physical errors can be generated by the trajectory, we can find all the subsequences of jumps separated by a recovery (termination with a success) and define the corresponding $\textbf{v}_h(m,i)$ for each subsequence. Conversely, if we have some $(\textbf{v}_h(m_0,i_0),\cdots,\textbf{v}_h(m_z,i_z))$, we can identify all the subsequences of jumps separated by the termination steps 3 and 4 of $\Phi(k)$, with each subsequence belonging to one of the independent random walks in $\Phi(k)$. We therefore have a one-to-one map between the two sets of trajectories. It follows that they are the same set.

This identification leads to the lower bound
\begin{align}
&\sum_{\{\textbf{m}(\textbf{a})\}}\sum_{i_0=0}^{\min\{h,m_0\}}\sum_{i_1=0}^{\min\{h,m_1\}}\cdots
\sum_{i_z=0}^{\min\{h,m_z\}}
\sum_{\{\textbf{v}_h(m_0,i_0)\}}\sum_{\{\textbf{v}_h(m_1,i_1)\}}\cdots \sum_{\{\textbf{v}_h(m_z,i_z)\}}
\nonumber\\
&\quad \quad
\left(\prod_{j=1}^{m_0}p_1g(v_{j-1},v_j)\right)p_0\left(\prod_{j=1}^{m_1}p_1g(v_{j-1},v_j)\right)p_0\cdots
\left(\prod_{j=1}^{m_z}p_1g(v_{j-1},v_j)\right)
\nonumber \\
&\geq\sum_{\{\textbf{m}(\textbf{a})\}}\sum_{i_0=0}^{\min\{h,m_0\}}\sum_{i_1=0}^{\min\{h,m_1\}}\cdots
\sum_{i_z=0}^{\min\{h,m_z\}}
\sum_{\{\textbf{v}_h(m_0,i_0)\}}\sum_{\{\textbf{v}_h(m_1,i_1)\}}\cdots \sum_{\{\textbf{v}_h(m_z,i_z)\}}
\nonumber \\
&\quad\quad\quad\quad \Pr\left(\Phi(k)\text{ returns success}|\text{First $k$ jumps}\right)
\nonumber\\
&\quad \quad\quad\quad\quad\quad\times
\left(\prod_{j=1}^{m_0}p_1g(v_{j-1},v_j)\right)p_0\left(\prod_{j=1}^{m_1}p_1g(v_{j-1},v_j)\right)p_0\cdots
\left(\prod_{j=1}^{m_z}p_1g(v_{j-1},v_j)\right)
\nonumber\\
&=\Pr\left(\Phi(k)\text{ returns success}\right), \label{sm:eq:stochastic_bound}
\end{align}
where the conditional probabilities are the success probability of $\Phi(k)$ conditioned on its first $k$ stochastic jumps.

To solve for $\Pr\left(\Phi(k)\text{ returns success}\right)$, we only need to solve for the success probability of each random walk. Let us now consider a generalized random walk that follows the same rules but can start from initial configurations with any number of physical errors. Let  $s_v$ be the conditional probability that, when initialized in a configuration with $v\in [1, h]$ physical errors, the random walk terminates with failure. For mathematical convenience, we also define $s_0 =0$ and  $s_{h+1} = 1$. The probability $s_v$ for $v\in[1,h]$ can be solved from a recurrence relation
\begin{equation}
    s_v = \frac{v}{N}p_1s_{v-1} + \left(1-\frac{v}{N}\right)p_1s_{v+1},\ s_0 = 0,\ s_{h+1} = 1.
\end{equation}
Recall that $s_1$ is the failure probability conditioned on an initial configuration with one physical error. To get the failure probability, we need to multiply $s_1$ by $p_1$---the probability of creating one error from a zero-error configuration. Therefore, the failure probability for a single random walk starting from a zero-error configuration is $p_1s_1$. The success probability for the random walk and $\Phi(k)$ are therefore $(1- p_1s_1)$ and $(1- p_1s_1)^k$, respectively. This provides a lower bound for the contribution
\begin{equation}
    \min_{\rho_0}\sum_{\{\textbf{a}\}}
    p_{a_1}p_{a_2}\cdots p_{a_k}\Tr[\rho_0\mathcal{R}\mathcal{E}_{a_1}\mathcal{E}_{a_2}\cdots\mathcal{E}_{a_k}\rho_0]\geq \left[(1-\xi)(1- p_1s_1)\right]^k.
\end{equation}
This yields the desired bound stated in the lemma. 
\end{proof}

Plugging the bound on $p_e(k)$ from Eq.\ (\ref{eq:ineq2}) into Eq.\ (\ref{eq:PnC}), we arrive at the desired bound in Theorem \ref{upper_bound_sl1}.

\begin{figure}
    \centering   \includegraphics[scale =0.5]{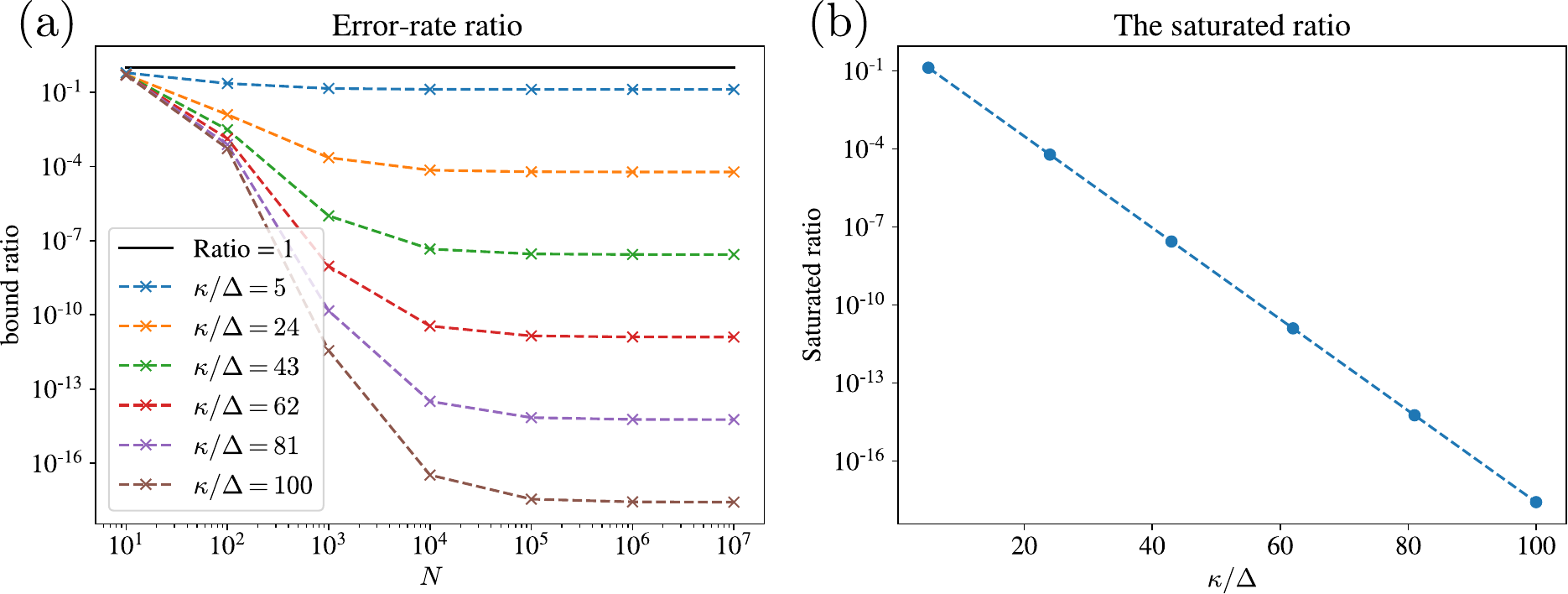}
    \caption{Numerical solutions to the error-rate ratio $s_1/\left(\frac{N\Delta}{\kappa+N\Delta}\right)^h$ for $h = 0.4N$. (a) The log-log plot shows the ratio as a function of $N$ for different $\kappa/\Delta$. We see that the ratio for different $\kappa/\Delta$ eventually saturates to a constant at large $N$. (b) The semi-log plot shows the saturated ratio as a function of $\kappa/\Delta$ with a fitted line $-0.4957\kappa/\Delta +0.0005$. Since $\lim_{N\to\infty}\left(\frac{N\Delta}{\kappa+N\Delta}\right)^{fN} = e^{-f\kappa/\Delta}$ for any constant $f$, the linearity of the plot of the saturated ratio suggests that $s_1$ has the form $\log s_1\sim -\kappa/\Delta$ at $N\to\infty$.}
    \label{fig:check_s1_1}
\end{figure}
\begin{figure}
    \centering    \includegraphics[scale = 0.5]{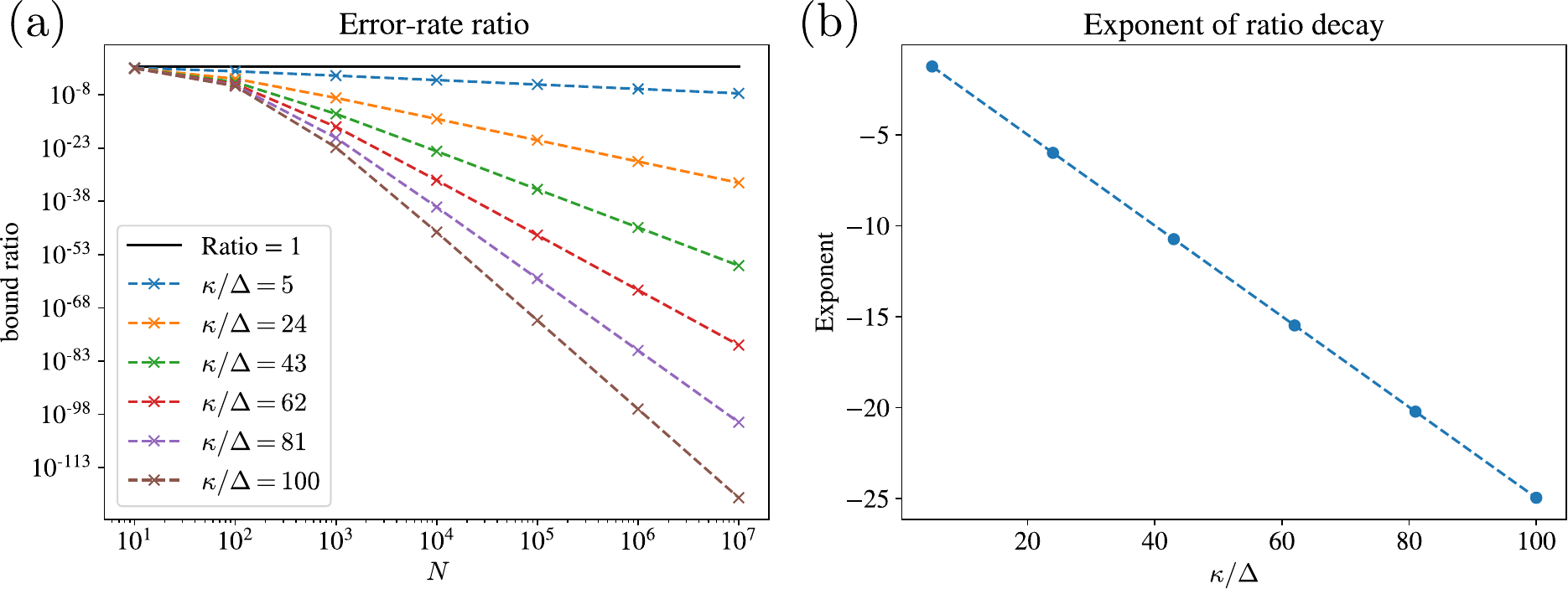}
    \caption{Numerical solutions to the error-rate ratio $s_1/\left(\frac{N\Delta}{\kappa+N\Delta}\right)^h$ for $h = N/2$. (a) The log-log plot shows the error-rate ratio as a function of $N$. We see in this case the ratio is steadily converging to 0 at $N\to\infty$. The linearity of the convergence suggests algebraic decay as a function of $N$ with different exponents. 
    (b) The plot, on a linear scale, shows the fitted exponents obtained by a linear fit to the straight lines in the error-ratio plot (a). A linear fit to the exponents yields $-0.24956\kappa/\Delta-0.00707\approx -\kappa/(4\Delta)$. We therefore deduce empirically that the ratio behaves as $\sim 1/N^{\kappa/(4\Delta)}$ as $N\to\infty$.}
    \label{fig:check_s1_2}
\end{figure}

To understand the behaviour of the improved bound (Theorem \ref{upper_bound_sl1}) compared to the one given by Theorem~\ref{lem:all_time1}, let us solve for $s_1$ numerically. We consider the ratio of error rates appearing in the two theorems: $s_1/\left(\frac{N\Delta}{\kappa+N\Delta}\right)^h$, where $h = fN$ with $f\in(0,1/2]$. We split the scenarios into two cases:
\begin{enumerate}
    \item[(i)] $f<1/2$. We evaluate the ratio for $f = 0.4$ and $0.45$ for different $N$ and different noise rates $\kappa/\Delta$. In Fig.~\ref{fig:check_s1_1}, we show the numerical results for $f = 0.4$. For $f = 0.45$, the plots look qualitatively the same.
    \item[(ii)] $f = 1/2$. The numerical results are shown in Fig.~\ref{fig:check_s1_2}.
\end{enumerate}
From these numerical results, we arrive at the empirical observations presented in the main text.

\subsection*{Appendix F: Proof of Theorem \ref{upper_bound_early_time}}
\labelx[Appendix F]{appendix_generic_boundarly_time_diagrams}

This appendix contains the proof of Theorem~\ref{upper_bound_early_time}. The proof of this theorem partially uses the proof of Theorem~\ref{generic_bound}. We start with Eq.~\eqref{eq:dyson_with_KL} which takes the form
\be
\bigl\<\sq\sss\bigl\> = \frac 1s+ \frac 1s \bigl\<\sq\sss\underbrace{\x\zzz\x\dots\x\zzz\x}_\textrm{$\ell+1$}\bigl\>.
\ee
Next, we recall that the Poissonian error generator 
takes the form
$\Delta\mathcal L_E(\rho) = N \Delta(\mathcal E(\rho) -\rho)$,
where $\mathcal E (\rho) = N^{-1}\sum_{\mu=1}^N E_\mu \rho E^\dag_\mu$ is the error superoperator. Let us use the notation
\be
N\Delta\mathcal E  =  \*, \qquad   N \Delta\mathcal I = \o.
\ee
where $\mathcal I$ is the identity superoperator.
Then, we can state the following diagrammatic relation (see Eq.~\eqref{eq:diagram_notations}):
\be\label{eqs:cross_decomp}
\x = \*-\o.
\ee
Using the decomposition in Eqs.~\eqref{eqs:V_decomposition} and \eqref{eqs:cross_decomp}, we have
\be\label{eqs:three_term_dyson}
\bigl\<\sq\sss\bigl\> = \frac 1s + \frac 1s\bigl\<\sq\sss\underbrace{\x\zzz\x\dots\x\zzz\x}_\textrm{$\ell$}\nnn\*\bigl\>-\frac 1s\bigl\<\sq\sss\underbrace{\x\zzz\x\dots\x\zzz\x}_\textrm{$\ell$}\nnn\o\bigl\>+\frac 1s\bigl\<\sq\sss\underbrace{\x\zzz\x \dots\x\zzz\x}_\textrm{$\ell$}\vvv\x\bigl\>.
\ee
Taking into account that $\mathcal W_t = e^{-\kappa t}\mathcal I$ for global decoder, the first terms can be rewritten as
\be
\frac 1s\bigl\<\sq\sss\underbrace{\x\zzz\x\dots\x\zzz\x}_\textrm{$\ell$}\nnn\*\bigl\> = \frac {1}{s(s+\kappa)}\bigl\<\sq\sss\underbrace{\x\zzz\x\dots\x\zzz\x}_\textrm{$\ell$}\*\bigl\>.
\ee
Using Eq.~\eqref{eqs:KL_cancelling}, we can rewrite the second term as
\be
\frac 1s\bigl\<\sq\sss\underbrace{\x\zzz\x\dots\x\zzz\x}_\textrm{$\ell$}\nnn\o\bigl\> = \frac { N\Delta}{s(s+\kappa)}\bigl\<\sq\sss\underbrace{\x\zzz\x\dots\x\zzz\x}_\textrm{$\ell$}\bigl\> = \frac { N\Delta}{s(s+\kappa)}\bigl\<\sq\sss\underbrace{\x\zzz\x\dots\x\zzz\x}_\textrm{$\ell+1$}\bigl\>.
\ee
The last term in Eq.~\eqref{eqs:three_term_dyson} vanishes due to Lemma~\ref{knill-lafl-lemma}.

By combining these results, we reach the conclusion that
\be
\bigl\<\sq\sss\bigl\> = \frac 1s + \frac 1{s(s+\kappa)}\bigl\<\sq\sss\underbrace{\x\zzz\x\dots\x\zzz\x}_\textrm{$\ell$}\*\bigl\>-\frac {N\Delta }{s(s+\kappa)}\bigl\<\sq\sss\underbrace{\x\zzz\x\dots\x\zzz\x}_\textrm{$\ell+1$}\bigl\>.
\ee
Subsequently, by utilizing Eq.~\eqref{eq:dyson_with_KL}, we can express
\be
\frac 1s\bigl\<\sq\sss\underbrace{\x\zzz\x\dots\x\zzz\x}_\textrm{$ \ell+1$}\bigl\> = \bigl\<\sq\sss\bigl\>-\frac 1s.
\ee
After rearranging the terms, we obtain
\be
\bigl\<\sq\sss\bigl\>\biggl(1+\frac{N \Delta }{s+\kappa}\biggl) = \frac 1s\biggl(1+\frac{ N \Delta}{s+\kappa}\biggl)+\frac 1{s(s+\kappa)}\bigl\<\sq\zzz\underbrace{\x\zzz\x\dots\x\zzz\x}_\textrm{$ \ell$}\*\bigl\>.
\ee
Lastly, by dividing both sides of the equation by $(1+N\Delta/(s+\kappa))$, we arrive at
\be
\bigl\<\sq\sss\bigl\> = \frac 1s +\frac {1}{s(s+\kappa+N \Delta)}\bigl\<\sq\sss\underbrace{\x\zzz\x\dots\x\zzz\x}_\textrm{$ \ell$}\*\bigl\>.
\ee
By repeating this procedure $\ell$ times, we obtain the expression
\be
\bigl\<\sq\sss\bigl\> = \frac 1s +\frac 1{s(s+\kappa+ N\Delta)^{\ell}}\bigl\<\sq\sss\*^{\ell+1}\bigl\>.
\ee
The explicit form of the inverse Laplace transform of this equation can be obtained as
\be
\begin{split}
\mathscr{L}^{-1}\bigl\<\sq\sss\*^{\ell+1}\bigl\>& = \frac 12(N \Delta)^{\ell+1} \Tr \Bigl[Q\exp(\mathcal Lt)\mathcal E^{\ell+1}(\delta\rho)\Bigl].
\end{split}
\ee
As $\mathcal{E}$ and $\exp(\mathcal{L}t)$ are completely positive trace-preserving (CPTP) maps that do not alter the matrix norm, we can put a lower bound on
\be\begin{split}
\Tr\Bigl[Q \exp(\mathcal Lt)\mathcal E^{\ell+1}(\delta\rho)\Bigl] &= \Tr\Bigl[Q \exp(\mathcal Lt)\mathcal E^{\ell+1}(|0\>\<0|)\Bigl] -\Tr\Bigl[Q \exp(\mathcal Lt)\mathcal E^{\ell+1}(|1\>\<1|)\Bigl] \\
&\geq -2\|Q\| = -2.
\end{split}
\ee
Thus we get
\be
\mathscr{L}^{-1}\bigl\<\sq\sss\*^{\ell+1}\bigl\>\geq -(N \Delta)^{\ell+1}.
\ee
Then we derive the bound
\be
\mathscr{L}^{-1}\bigl\<\sq\sss\bigl\> \geq 1 -(N \Delta)^{\ell+1}\mathscr{L}^{-1}\Biggl\{\frac {1}{s^2(s+\kappa+N \Delta)^{\ell}}\Biggl\}.
\ee
Taking the inverse Laplace transform, we get
\be
\epsilon(t),\delta(t) \leq  \frac{1}{(1+\kappa/N \Delta)^{\ell+1}}F\Bigl((\kappa +N \Delta)t\Bigl).
\ee
This expression concludes our proof.

\subsection*{\review{Appendix G: Proof of Theorem 5} }

\review{Let us first consider the logical bit-flip probability in the presence of the recovery process. Let $F^{(0)}$ be the subset of trajectories (see Eq.~\eqref{eq:poisson_form} and below for informal definition) without any recovery events (i.e. $\mu\neq0$), a subset of all possible trajectories $F$. From the Poissonian picture in Eq.~\eqref{eq:poisson_form}, we derive that
\begin{align}\Tr[\ket{1}\bra{1}\mathcal{R}e^{\hL\tau}\ket{0}\bra{0}]&=\Tr[\ket{1}\bra{1}\mathcal{R}\left(\sum_{\boldsymbol{\mu}\in F}p(\boldsymbol{\mu},\tau)\mathcal{E}_{\boldsymbol{\mu}}\right)\ket{0}\bra{0}]\nonumber\\&\geq\Tr[\ket{1}\bra{1}\mathcal{R}\left(\sum_{\boldsymbol{\mu}\in F^{(0)}}p(\boldsymbol{\mu},\tau)\mathcal{E}_{\boldsymbol{\mu}}\right)\ket{0}\bra{0}]\nonumber\\&=\Tr[\ket{1}\bra{1}\mathcal{R}\left(\sum_{\boldsymbol{\mu}\in F^{(0)}}p(\boldsymbol{\mu},\tau)\right)\left(\sum_{\boldsymbol{\mu}\in F^{(0)}}\frac{p(\boldsymbol{\mu},\tau)}{\sum_{\boldsymbol{\mu}'\in F^{(0)}}p(\boldsymbol{\mu}',\tau)}\mathcal{E}_{\boldsymbol{\mu}}\right)\ket{0}\bra{0}]\nonumber\\&=e^{-\kappa\tau}\Tr[\ket{1}\bra{1}\mathcal{R}e^{\hL_{E}\tau}\ket{0}\bra{0}]\geq\alpha_n(\tau)e^{-\kappa\tau}>ae^{-\kappa\tau},\label{eq:finite_{e}rror}
\end{align}
indicating a non-vanishing error probability. Here, we used Assumption~1 and the normalization
\begin{equation}
     \sum_{\boldsymbol \mu\in F^{(0)}} p(\boldsymbol \mu,\tau)=\sum_{n=0}^\infty \left(\frac{N\Delta}{\kappa + N\Delta}\right)^n\frac{[\tau(\kappa+N\Delta)]^n}{n!}e^{-\tau(\kappa+N\Delta)} = e^{-\kappa \tau}.
\end{equation}
Using Assumption 2 with corresponding Eq.~\eqref{eq:inequality_cond}, we then impose a lower bound on the logical error measure by considering the dynamics interleaved with the recovery map at every time interval of length $\tau$.
Let the parity superoperators be $\mathcal{X}(\rho) = \bar X\rho\bar X$ and $\mathcal{Z}(\rho) =\bar Z\rho\bar Z$, where $\bar X$ and $\bar Z$ are the logical operators. For stabilizer codes, these operators are Pauli strings that commute with the stabilizers. Also, the recovery $\mathcal{R}$ can be written in terms of projectors onto Pauli stabilizer configurations followed by Pauli strings that fix the stabilizer configuration (see Section~\ref{sec:examples} for an explicit expression).
Thus, these superoperators satisfy $[\mathcal{R},\mathcal{X}]=[\mathcal{R},\mathcal{Z}]=0$. 
Similarly, for Pauli errors ${\cal{E}}_{\mu>0}$, we have
\begin{equation}
    {\cal X}{\cal E}_{\mu}(\rho)=\mathcal{X}(E_{\mu}\rho E_{\mu})=E_{\mu}\mathcal{X}(\rho)E_{\mu}={\cal E}_{\mu}{\cal X}(\rho)~,
\end{equation}
and similarly for $\cal{X}\to\cal{Z}$.
As a result, the parity operators are conserved during the dynamics, i.e. $[\mathcal{X}, e^{\hL t}]=[\mathcal{Z}, e^{\hL t}] = 0$.
To analyze the logical error rate of our system, we limit our attention to the initial state $\ket{0}\bra{0}$. This state is an eigenvector of the parity operator $\mathcal{Z}$ with eigenvalue $+1$. 
Its eigenvalue is a good quantum number with respect to the noise, meaning that no off-diagonal matrix elements (such as $\ket{0}\bra{1}$) will be created during the evolution. 
Combining this with Eq.~\eqref{eq:finite_{e}rror}, we have
\begin{equation}
    \mathcal{R}e^{\hL\tau}\ket{ 0}\bra{ 0} = p_0\ket{ 0}\bra{ 0}+ p_1\ket{ 1}\bra{ 1},\ \frac{a}{2} e^{-\kappa\tau_c}< p_1\leq \frac{1}{2},\label{eq:dynamics_0}
\end{equation}
where $p_0= 1-p_1$ and here and below we have chosen $\tau = \tau_c + \log 2/\kappa$ such that $ae^{-\kappa\tau} = ae^{-\kappa\tau_c}/2< 1/2$.
By applying $\mathcal{X}$ to both sides of Eq.~\eqref{eq:dynamics_0}, we learn that
\begin{equation}
    \mathcal{R}e^{\hL\tau}\ket{ 1}\bra{ 1} = p_0\ket{ 1}\bra{ 1}+ p_1\ket{ 0}\bra{ 0}.
\end{equation}
Using simple linear algebra, we find:
\begin{equation}
     \left(\mathcal{R}e^{\hL\tau}\right)^m\ket{ 0}\bra{ 0} = \frac{1}{2}\left(1+(1-2p_1)^m\right)\ket{ 0}\bra{ 0}+ \frac{1}{2}\left(1-(1-2p_1)^m\right)\ket{ 1}\bra{1}.
\end{equation}
This yields
\begin{equation}
    \Tr[\ket{ 0}\bra{ 0}\mathcal{R}\left( e^{\hL \tau}\mathcal{R}\right)^m\ket{ 0}\bra{ 0}]\leq \frac{1}{2}+\frac{1}{2}(1-a e^{-\kappa\tau_c})^m.
\end{equation}
For a total time $t = m\tau$, we use Eq.~\eqref{eq:inequality_cond} from Assumption 2 and arrive at the following bound,
\begin{equation}
    \frac{1}{2}-\frac{1}{2}(1-a e^{-\kappa\tau_c})^{t/\tau}\leq 1-\Tr[\ket{ 0}\bra{ 0}\mathcal{R}e^{\hL t}\ket{ 0}\bra{ 0}]\leq \epsilon(t).\label{eq:lower_bound}
\end{equation}
 To make this bound comparable to previous results, we rewrite this inequality as $\epsilon(t) \geq \frac{1}{2}(1-\exp(-\Delta_{\rm eff}(t)t))$, where the effective logical error rate is
\begin{equation}
\Delta_{\rm eff}(t) := -\frac{1}{t}\log(1-2\epsilon(t)) \geq - \frac{1}{\tau}\log(1-a e^{-\kappa\tau_c}) = -\frac{\kappa\log(1-ae^{-\kappa\tau_c})}{\kappa\tau_c +\log 2}.
\end{equation}
This condition concludes our proof.}

\subsection*{Appendix H: Dissipative toric code}
\labelx[Appendix H]{appendix_h}

In this appendix, we closely examine the Lindblad operator for the autonomous decoder based on the two-dimensional toric code as described in Section~\ref{sec:examples}. We demonstrate that, in the absence of noise, the recovery Lindblad operator in Eq.~\eqref{eq:noiseless_lind} is exactly solvable. We use these solutions to perturbatively derive the spectral gap of the full operator.

For simplicity, we will limit our analysis to the case in which only star excitations are allowed, i.e., no plaquettes are excited. Despite this, we emphasize that our main conclusions should hold in the presence of both types of excitations. We will consider noise models where the eigenvalues of $B_p$ are always good quantum numbers, and focus our attention on the gauge sector where $B_p=+1$ for all $p$, i.e., the subspace that contains the ground states. The reduced Hilbert space will consist of states that have an even number of star excitations. We choose the following labeling convention for states that span the reduced Hilbert space:
\begin{equation} \label{eq:convention1}
\ket{0,0;  \vec{0} } = \prod_i (1+A_i) \ket{ \text{vac} }, \qquad \ket{r,s;  \vec{0} } = (g_x)^r (g_y)^s  \ket{0,0;  \vec{0} }, \qquad  \ket{r,s, \vec{k} } = \left(  \prod Z \right)_{\vec{k}} \ket{r,s, \vec{0}},
\end{equation}
where $A_i$ represents different star operators, $\ket{ \text{vac} }$ is the ground state of all $Z_j$ operators: $Z_j \ket{ \text{vac} } = \ket{ \text{vac} } $; $r,s \in 0,1$ label different topological sectors; $g_{x/y} = \Pi_{\text{hor/vert}} X$ is a product of $X$ operators along a string (on the dual lattice) that wraps around the horizontal/vertical direction of the torus. It is easy to check that $\ket{r,s; \vec{0} }$ are orthogonal ground states of $H$ (they are $+1$ eigenstates of all the $A_s$ and $B_p$). Excited states $\ket{r,s, \vec{k}}$ are labeled by $\vec{k}$. For a system with $L\times L$ stars, $\vec{k}$ is an $L^2$-dimensional vector that labels the excited stars with $1$ and de-excited stars with $0$. Excited eigenstates are defined by applying strings of $\left( \prod Z \right)$ operators on the ground state via the smallest number of $Z$ operators. When there are multiple minimal-weight excitation operators, then a $\left( \prod Z \right)_{\vec{k}} $ operator is chosen arbitrarily from the various options.

We consider the following recovery map:

\be \label{eq:global_diss}
\mathcal R(\rho) = \sum_{\vec k} K_{\vec k}\rho K^\dag_{\vec k},\qquad
    K_{\vec k} = \left(  \prod Z \right)_{\vec{k}} P_{\vec{k}},
\end{equation}
where $P_{\vec{k}} = \Pi_j (1+(-1)^{\vec{k}_j} A_j)$ is a projector onto a given star configuration $\vec{k}$, and $\left(  \prod Z \right)_{\vec{k}}$ is an operator that ``fixes'' the error using the minimal number of $Z$ operators. In other words, $K_{ \vec{k} } \ket{r,s; \vec{k} } =\ket{r,s; \vec{0} }$.  Note that the number of star configurations is $2^{L^2-1}$, which means that  the number of dissipators scales exponentially with system size. Note also that the dissipators are non-local, i.e.~they have support on the full lattice. Nevertheless, there is still a notion of locality: The  dissipators fix star errors via strings of $Z$ operators which minimize the total path length between excited stars, a process known as minimal weight matching.

One can easily check that arbitrary superpositions of toric code ground states are the \textit{only} steady states of the system.
It is thus clear that this idealized limit hosts a qudit steady-state structure. We now ask how stable this degeneracy is with respect to dephasing perturbations, which act on each physical qubit,
i.e.~how quickly do ``coherences'' decay in the presence of dephasing. To this end, we introduce the following dissipators on each edge of the lattice:
\begin{equation} \label{eq:dephasing}
E_\mu =  Z_\mu.
\end{equation}

Let us reformulate the question more precisely. We begin by noting that the jump operators $\{ K_{\vec{k}}, E_\mu\}$ commute with $B_p$, and hence indeed it makes sense to focus on the gauge sector with $B_p=+1$ for all $p$.  Even within this gauge sector, we can further partition the  subspace into different ``topological sectors''. Let us define the following projection operators:
\begin{equation} \label{eq:projectors}
P_{rs} = \sum_{ \vec{k} } \ket{r,s; \vec{k} } \bra{ r , s;  \vec{k} },
\end{equation}
where $r,s \in 0,1$ and $P_{rs}$ projects states into topological sector $r,s$.  We now note that all of the  dissipators  commute with $P_{rs}$. Thus these projectors are ``strong symmetries'' of the Lindbladian  \cite{prosen2012}, and hence there exists a basis where the Lindbladian in Eq.~\eqref{eq:markov_process} is be block-diagonalized and consists $4^2=16$ different blocks:
\begin{equation} \label{eq:blockLind}
    \hL= \text{Diag}[  \hL_{0,0}, \hL_{0,1}, \hL_{0,2}, \ldots \hL_{3,3} ],
\end{equation}
where the numbers 0 to 3 label four different topological sectors of the bras and kets accoding to the convention $(r=0, s=0) \rightarrow 0$, $(r=1, s=0) \rightarrow 1$, $(r=0, s=1) \rightarrow 2$, $(r=1, s=1) \rightarrow 3 $.  In words, the Lindbladian $\hL_{0,0}$ acts on operators where both ket and bra belong to the same topological sector $r=0, s=0$; on the other hand, $\hL_{0,1}$ acts on operators where the ket belongs to sector  $r=0, s=0$, while the bra belongs to sector $r=1, s=0$. Operators belonging to different topological sectors evolve independently.

We now show that, in the absence of dephasing ($\Delta=0$), the model is exactly solvable, meaning that we can write down exact expressions for the right and left eigenoperators of the Lindbladian in each topological sector and its spectrum. Let us define right and left eigenoperators of the Lindbladian:
\begin{equation}
    \hL_{p,q}(r_{p,q; m}) = \Lambda_m r_{p,q; m}, \qquad   \hL_{p,q}^\dagger(l_{p,q; m}) = \Lambda_m^* l_{p,q; m},
\end{equation}
where $p,q$ label the topological sector, and  $m$ represents different eigenvalues within a given sector (it turns out that the spectrum is the same in all sectors in this limit, so we suppress the topological labels on $\Lambda$).

For $\Delta=0$, each topological sector has exactly one eigenvalue of zero: $\Lambda_0 = 0$. One can show that the corresponding eigenoperators are
\begin{equation}
    r_{p,q; 0} = \ket{p; \vec{0}} \bra{q; \vec{0}}, \qquad l_{p,q; 0} = \sum_{ \vec{k}} \ket{ p; \vec{k}} \bra{q; \vec{k}}.
\end{equation}
This ensures that arbitrary \textit{superpositions} of toric code ground states are indeed steady states of the model:
\begin{equation}
    \ket{\psi} = \sum_{p=0}^3 c_p \ket{p; \vec{0}},  \qquad \hL(\ket{\psi} \bra{\psi}) = 0, \quad \forall \sum_{p=0}^3 |c_p|^2=1.
\end{equation}
The full Lindbladian $\hL$ therefore has $4^2=16$ eigenvalues of zero. In a slight abuse of notation, we shall call this a ``qubit steady state structure''. (Rather than a qudit steady state structure.)

We now turn to ``diagonal'' eigenoperators which come with a decay rate: $\Lambda_{\vec{k}} = -\kappa$. The corresponding eigenoperators are
\begin{equation}
    r_{p,q; \vec{k}} = \ket{p; \vec{k}} \bra{q; \vec{k}} - \ket{p; \vec{0}} \bra{q; \vec{0}}, \qquad l_{p,q; \vec{k}} = \ket{ p; \vec{k}} \bra{q; \vec{k}}.
\end{equation}
It is clear that $\tr[ r_{p,q; \vec{k}} ] = 0$ and  $\tr[l_{p,q; \vec{k}}^\dagger r_{p,q; \vec{k}} ] = 1$, which are necessary conditions.

Finally, we turn to off-diagonal eigenoperators  which come with a decay rate: $\Lambda_{\vec{k}, \vec{k'}} = -\kappa$. The corresponding eigenoperators are:
\begin{equation}
    r_{p,q; \vec{k}, \vec{k'}} = \ket{p; \vec{k}} \bra{q; \vec{k'}}, \qquad l_{p,q; \vec{k}, \vec{k'}} = \ket{p; \vec{k}} \bra{q; \vec{k'}}.
\end{equation}
In this case, the right and left eigenoperators happen to be identical. 

\subsubsection*{Perturbative results}

Having found the exact expressions for all of the right and left eigenoperators of the Lindbladian  in the absence of dephasing $\Delta=0$,  we would now like to use perturbation theory to examine the effects of small dephasing $\Delta > 0$. In particular, we would like to know the fate of the qubit steady state structure, i.e.~whether the Lindbladian with weak dephasing has 16 eigenvalues of zero in the thermodynamic limit.

The Lindbladian can be block diagonalized into different topological sectors (see Eq.~\eqref{eq:blockLind}) even in the presence of dephasing. Note that  all operators with non-zero trace must belong to one of the diagonal sectors $\hL_{q,q}$, while the off-diagonal sectors act on operators with zero trace. Since it is possible to initialize a valid (traceful) density matrix in each of the diagonal sectors $\hL_{q,q}$, we know that (even in the presence of dephasing) each of these sectors must have an eigenvalue of zero corresponding to the steady state in each topological sector. The full Lindbladian $\hL$ is thus guaranteed to have at least four eigenvalues of zero even in the presence of $Z$-dephasing. The off-diagonal sectors $\hL_{p,q}$ are not guaranteed steady state solutions in general. However, from the exact diagonalization above, we know that, in the absence of dephasing, these sectors  will have an eigenvalue of zero, since arbitrary \textit{superpositions} of toric code ground states are stable (e.g.~$\hL( \ket{0; \vec{0} } \bra{1; \vec{0} } )=0$). We wish to understand how the decay rate of these off-diagonal coherences scales with system size in the presence of dephasing; in particular, we would like to know if the decay  rate scales to zero in the thermodynamic limit.

Let us specialize to an off-diagonal sector $\hL_{p,q}$ (we shall suppress the topological indices henceforth) and examine the shift to the eigenvalue $\Lambda_0=0$  via perturbation theory. The corrections (up to third order) read
\begin{align}
    \delta \Lambda^{(1)} &= \langle l_0 | \mathcal{L'} | r_0  \rangle, \\
    \delta \Lambda^{(2)} &= \sum_{m \neq 0} \frac{ \langle l_0 | \mathcal{L'} | r_m  \rangle \langle l_m | \mathcal{L'} | r_0  \rangle}{\Lambda_0 - \Lambda_m } = \frac{1}{\kappa} \left( \langle l_0 |\mathcal{L'}^2 |r_0 \rangle - \langle l_0 | \mathcal{L'} | r_0 \rangle ^2 \right), \\
    \delta \Lambda^{(3)} &= \sum_{m_1 \neq 0} \sum_{m_2 \neq 0} \frac{ \langle l_0 | \mathcal{L'} | r_{m_1}  \rangle  \langle l_{m_1} | \mathcal{L'} | r_{m_2}  \rangle \langle l_{m_2} | \mathcal{L'} | r_0  \rangle}{(\Lambda_0 - \Lambda_{m_1})  (\Lambda_0 - \Lambda_{m_2} )  } - \delta \Lambda_1 \sum_{m \neq 0} \frac{ \langle l_0 | \mathcal{L'} | r_m  \rangle \langle l_m | \mathcal{L'} | r_0  \rangle}{(\Lambda_0 - \Lambda_m )^2} \\
    &= \frac{1}{\kappa^2} \left( \langle l_0 |\mathcal{L'}^3 |r_0 \rangle  - 3 \langle l_0 | \mathcal{L'} | r_0  \rangle \langle l_0 |\mathcal{L'}^2 |r_0 \rangle  + 2 \langle l_0 | \mathcal{L'} | r_0  \rangle^3 \right),
\end{align}
where we use the shorthand $\langle l | \hL' | r \rangle \equiv \tr[ l^\dagger \hL'(r)]$, and
\begin{equation}
    \hL'(\rho) = \sum_i \left(  Z_i  \rho Z_i  - \rho \right)
\end{equation}
is the dephasing perturbation. We have simplified some of the general expressions by noting that the relevant excited eigenvalues are all $\Lambda_m = -\kappa$ and $\Lambda_0 = 0$. Examining the  expressions for the eigenvalue shift, it is clear that, if $\langle l_0 |(\mathcal{L'})^x |r_0 \rangle =0 $ for all $x \leq y$, then $\delta \Lambda^{(y)} = 0$.

\textbf{One-dimensional system.}---Let's consider a $1 \times L$ lattice where $L=2j+1, j \in \mathbb{Z}$, such that $L$ is odd. We define two different toric code ground states via $\ket{2; \vz} =X_1 \ket{0; \vz} $, where $X_1$ is a global loop on the dual lattice in the vertical direction (in this case, the vertical direction has length 1 so the global loop is a single operator). We define the right and left eigenoperators in the unperturbed limit:
\begin{equation}
r_0 = \ket{2; \vz} \bra{0; \vz}, \qquad l_0 = \sum_{\vk}  \ket{2; \vk} \bra{0; \vk}.
\end{equation}
We next define the perturbation
\begin{equation}
\hL' (\rho) = \hZ(\rho) - L\rho, \qquad \hZ(\rho) = \sum_{i=1}^L Z_i   \rho Z_i,
\end{equation}
where $Z_i$ act only on the horizontal edges of the lattice.

Consider the first term in the perturbation theory:
\begin{align}
\tr[ l_0^\dagger \hL' (r_0)] &= \sum_{\vk}  \text{tr}[\ket{0, \vk} \bra{2; \vk}  \left( \hZ (\ket{2, \vz} \bra{0, \vz}) - L  \ket{2, \vz} \bra{0, \vz} \right) ]  \\
&= - L + \sum_{\vk}  \bra{2; \vk}   \hZ (\ket{2, \vz} \bra{0, \vz})  \ket{0, \vk} ]   \\
&= - L + L = 0.
\end{align}

To order $\alpha$, we find
\begin{align}
\tr[ l_0^\dagger (\hL')^\alpha (r_0)] &=  \sum_{i=0}^\alpha  \left(\begin{array}{c}
\alpha \\
i
\end{array}\right)
(-L)^{\alpha-i}  \sum_{\vk}  \bra{2; \vk}   \hZ^{i} (\ket{2, \vz} \bra{0, \vz})  \ket{0, \vk},
\end{align}
where we have introduced the binomial coefficient. We now note that,  for $i \leq j$,
\begin{equation}
\sum_{\vk}  \bra{2; \vk}   \hZ^{i} (\ket{2, \vz} \bra{0, \vz})  \ket{0, \vk}  = L^i, \qquad \text{for: } i \leq j.
\end{equation}
This implies
\begin{align}
\tr[ l_0^\dagger (\hL')^\alpha (r_0)] &= L^2 \sum_{i=0}^\alpha  \left(\begin{array}{c}
\alpha \\
i
\end{array}\right)
(-1)^{\alpha-i}  1^i = 0 , \qquad \text{for: }  \alpha \leq j,
\end{align}
where we have used a property of binomial coefficients.

We need to be careful at order $\alpha = j+1$. In this case,
\begin{equation}
\sum_{\vk}  \bra{2; \vk}   \hZ^{ j+1} (\ket{2, \vz} \bra{0, \vz})  \ket{0, \vk}  = L^{j+1} - 2 \left( \frac{L !}{ j!} \right).
\end{equation}
The term with the factorial is counting the number of configurations with $j+1$ excited edges which have a logical error after applying the recovery jump operators $L_{\vec{k}}$.  This implies
\begin{align}
\tr[ l_0^\dagger (\hL')^{j+1}(r_0)] &= - 2 \left( \frac{L !}{ j!} \right).
\end{align}
This is the lowest-order non-trivial recovery.

This implies that  the contribution to the eigenvalue at this order is
\begin{equation}
\Lambda/\kappa  = -2  \left( \frac{L !}{ (L/2)!} \right) \left( \frac{\Delta}{\kappa} \right)^{L/2} + O( \left( \Delta /\kappa \right)^{L/2 + 1}  ),
\end{equation}
where we have used $j+1 \approx L/2$.  We note that this term blows up in the thermodynamic limit $L \rightarrow \infty$, since the factorial term grows faster than the exponential term decays. One way to remedy this is to increase the dissipation linearly with the system size: $\kappa = \kappa_0 L$. Then the thermodynamic limit is well-defined:
\begin{equation}
\lim_{L \rightarrow \infty} \left( \frac{L !}{ (L/2)!} \right) \left( \frac{\Delta}{\kappa_0 L} \right)^{L/2} = 0
\end{equation}
for $\Delta / \kappa_0 \ll 1$ such that the recovery to the eigenvalue goes to zero in the thermodynamic limit.

\textbf{Two-dimensional system.}---A similar analysis can be done for a 2D system on an $L \times L$ lattice.  We consider a perturbation
\begin{equation}
\hL' (\rho) = \hZ(\rho) - 2L^2 \rho, \qquad \hZ(\rho) = \sum_{i=1}^{2L^2} Z_i   \rho Z_i,
\end{equation}
where $Z_i$ act on each of the $2L^2$ edges of the lattice.

Again, the lowest-order contribution comes at order $j+1$:
\begin{align}
\tr[ l_0^\dagger (\hL')^{j+1}(r_0)] &= -   2L \left( \frac{L !}{ j!} \right),
\end{align}
where this factor basically counts the number of configurations with $j+1$ excited edges, which have a logical error after applying the recovery jump operators $L_{\vec{k}}$.

The recovery to the eigenvalue at this order is
\begin{equation}
\Lambda/\kappa  = - 2 L \left( \frac{L !}{ (L/2)!} \right)\left( \frac{\Delta}{\kappa} \right)^{L/2} + O( \left( \Delta /\kappa \right)^{L/2 + 1}  ).
\end{equation}
Again, this term blows up in the thermodynamic limit $L \rightarrow \infty$, since the factorial term grows faster than the exponential term decays.  One way to remedy this is to increase the dissipation linearly with the \textit{linear} system size: $\kappa = \kappa_0 L$. Then the thermodynamic limit is well-defined:
\begin{equation}
\lim_{L \rightarrow \infty} L \left( \frac{L !}{ (L/2)!} \right) \left( \frac{\Delta}{\kappa_0 L} \right)^{L/2} = 0
\end{equation}
for $\Delta / \kappa_0 \ll 1$ such that the recovery to the eigenvalue goes to zero in the thermodynamic limit.

In conclusion, we have shown that the dissipative toric code with single-shot recovery jumps will host a qubit steady state structure in the presence of dephasing perturbations if the strength of the recovery process scales with the linear system size. Moreover, our analysis suggests that \textit{any} local perturbation will not destroy the qubit steady state for such a system. As we discuss in the main text, the perturbative analysis also misses a non-perturbative contribution for the system's lifetime, which will result in a quantitatively different scaling dependence on $\Delta$.

\bibliographystyle{quantum}
\bibliography{refs.bib}

\begin{thebibliography}{10}

\bibitem{terhal2015quantum}
Barbara~M. Terhal.
\newblock ``Quantum error correction for quantum memories''.
\newblock \href{https://dx.doi.org/10.1103/RevModPhys.87.307}{Rev. Mod. Phys. {\bf 87}, 307--346}~(2015).

\bibitem{gottes1997}
Daniel~Eric Gottesman.
\newblock ``Stabilizer codes and quantum error correction''.
\newblock \href{https://dx.doi.org/10.7907/rzr7-dt72}{PhD thesis}.
\newblock California Institute of Technology.
\newblock ~(1997).

\bibitem{Paz1998}
Juan~Pablo Paz and Wojciech~Hubert Zurek.
\newblock ``{Continuous error correction}''.
\newblock \href{https://dx.doi.org/10.1098/rspa.1998.0165}{Proc. R. Soc. A: Math. Phys. Eng. Sci. {\bf 454}, 355--364}~(1998).

\bibitem{Ahn2002}
Charlene Ahn, Andrew~C. Doherty, and Andrew~J. Landahl.
\newblock ``{Continuous quantum error correction via quantum feedback control}''.
\newblock \href{https://dx.doi.org/10.1103/PhysRevA.65.042301}{Phys. Rev. A {\bf 65}, 042301}~(2002).

\bibitem{Sarovar2005}
Mohan Sarovar and G.~J. Milburn.
\newblock ``{Continuous quantum error correction by cooling}''.
\newblock \href{https://dx.doi.org/10.1103/PhysRevA.72.012306}{Phys. Rev. A {\bf 72}, 012306}~(2005).

\bibitem{Mabuchi2009}
Hideo Mabuchi.
\newblock ``{Continuous quantum error correction as classical hybrid control}''.
\newblock \href{https://dx.doi.org/10.1088/1367-2630/11/10/105044}{New J. Phys. {\bf 11}, 105044}~(2009).

\bibitem{Oreshkov2013}
Ognyan Oreshkov.
\newblock ``Continuous-time quantum error correction''.
\newblock In Daniel~A. Lidar and Todd~A. Brun, editors, Quantum Error Correction.
\newblock \href{https://dx.doi.org/10.1017/CBO9781139034807.010}{Pages 201--228}.
\newblock Cambridge University Press~(2013).

\bibitem{ErrorCorrectionZoo}
Victor~V. Albert and Philippe Faist.
\newblock ``The error correction zoo''.
\newblock \url{https://errorcorrectionzoo.org/}.

\bibitem{science.aaa2085}
Z.~Leghtas, S.~Touzard, I.~M. Pop, A.~Kou, B.~Vlastakis, A.~Petrenko, K.~M. Sliwa, A.~Narla, S.~Shankar, M.~J. Hatridge, M.~Reagor, L.~Frunzio, R.~J. Schoelkopf, M.~Mirrahimi, and M.~H. Devoret.
\newblock ``Confining the state of light to a quantum manifold by engineered two-photon loss''.
\newblock \href{https://dx.doi.org/10.1126/science.aaa2085}{Science {\bf 347}, 853--857}~(2015).

\bibitem{PhysRevX.8.021005}
S.~Touzard, A.~Grimm, Z.~Leghtas, S.~O. Mundhada, P.~Reinhold, C.~Axline, M.~Reagor, K.~Chou, J.~Blumoff, K.~M. Sliwa, S.~Shankar, L.~Frunzio, R.~J. Schoelkopf, M.~Mirrahimi, and M.~H. Devoret.
\newblock ``Coherent oscillations inside a quantum manifold stabilized by dissipation''.
\newblock \href{https://dx.doi.org/10.1103/PhysRevX.8.021005}{Phys. Rev. X {\bf 8}, 021005}~(2018).

\bibitem{gertler2021protecting}
Jeffrey~M. Gertler, Brian Baker, Juliang Li, Shruti Shirol, Jens Koch, and Chen Wang.
\newblock ``Protecting a bosonic qubit with autonomous quantum error correction''.
\newblock \href{https://dx.doi.org/10.1038/s41586-021-03257-0}{Nature {\bf 590}, 243--248}~(2021).

\bibitem{lescanne_exponential_2020}
Raphaël Lescanne, Marius Villiers, Théau Peronnin, Alain Sarlette, Matthieu Delbecq, Benjamin Huard, Takis Kontos, Mazyar Mirrahimi, and Zaki Leghtas.
\newblock ``Exponential suppression of bit-flips in a qubit encoded in an oscillator''.
\newblock \href{https://dx.doi.org/10.1038/s41567-020-0824-x}{Nat. Phys. {\bf 16}, 509--513}~(2020).

\bibitem{de_neeve_error_2022}
Brennan de~Neeve, Thanh-Long Nguyen, Tanja Behrle, and Jonathan~P. Home.
\newblock ``Error correction of a logical grid state qubit by dissipative pumping''.
\newblock \href{https://dx.doi.org/10.1038/s41567-021-01487-7}{Nat. Phys. {\bf 18}, 296--300}~(2022).

\bibitem{campagne-ibarcq_quantum_2020}
P.~Campagne-Ibarcq, A.~Eickbusch, S.~Touzard, E.~Zalys-Geller, N.~E. Frattini, V.~V. Sivak, P.~Reinhold, S.~Puri, S.~Shankar, R.~J. Schoelkopf, L.~Frunzio, M.~Mirrahimi, and M.~H. Devoret.
\newblock ``Quantum error correction of a qubit encoded in grid states of an oscillator''.
\newblock \href{https://dx.doi.org/10.1038/s41586-020-2603-3}{Nature {\bf 584}, 368--372}~(2020).

\bibitem{sivak_real-time_2023}
V.~V. Sivak, A.~Eickbusch, B.~Royer, S.~Singh, I.~Tsioutsios, S.~Ganjam, A.~Miano, B.~L. Brock, A.~Z. Ding, L.~Frunzio, S.~M. Girvin, R.~J. Schoelkopf, and M.~H. Devoret.
\newblock ``Real-time quantum error correction beyond break-even''.
\newblock \href{https://dx.doi.org/10.1038/s41586-023-05782-6}{Nature {\bf 616}, 50--55}~(2023).

\bibitem{brown2016quantum}
Benjamin~J. Brown, Daniel Loss, Jiannis~K. Pachos, Chris~N. Self, and James~R. Wootton.
\newblock ``Quantum memories at finite temperature''.
\newblock \href{https://dx.doi.org/10.1103/RevModPhys.88.045005}{Rev. Mod. Phys. {\bf 88}, 045005}~(2016).

\bibitem{peierls_1936}
R.~Peierls.
\newblock ``On {I}sing's model of ferromagnetism''.
\newblock \href{https://dx.doi.org/10.1017/S0305004100019174}{Math. Proc. Cambridge Philos. Soc. {\bf 32}, 477--481}~(1936).

\bibitem{griffiths1964}
Robert~B. Griffiths.
\newblock ``Peierls proof of spontaneous magnetization in a two-dimensional ising ferromagnet''.
\newblock \href{https://dx.doi.org/10.1103/PhysRev.136.A437}{Phys. Rev. {\bf 136}, A437--A439}~(1964).

\bibitem{dennis2002}
Eric Dennis, Alexei Kitaev, Andrew Landahl, and John Preskill.
\newblock ``Topological quantum memory''.
\newblock \href{https://dx.doi.org/10.1063/1.1499754}{J. Math. Phys. {\bf 43}, 4452--4505}~(2002).

\bibitem{Alicki:2010}
Robert Alicki, Michal Horodecki, Pawel Horodecki, and Ryszard Horodecki.
\newblock ``On thermal stability of topological qubit in kitaev's 4d model''.
\newblock \href{https://dx.doi.org/10.1142/S1230161210000023}{Open Syst. Inf. Dyn. {\bf 17}, 1--20}~(2010).

\bibitem{pastawski:2011}
Fernando Pastawski, Lucas Clemente, and Juan~Ignacio Cirac.
\newblock ``Quantum memories based on engineered dissipation''.
\newblock \href{https://dx.doi.org/10.1103/PhysRevA.83.012304}{Phys. Rev. A {\bf 83}, 012304}~(2011).

\bibitem{liu2023dissipative}
Yu-Jie Liu and Simon Lieu.
\newblock ``Dissipative phase transitions and passive error correction''.
\newblock \href{https://dx.doi.org/10.1103/PhysRevA.109.022422}{Phys. Rev. A {\bf 109}, 022422}~(2024).

\bibitem{Bravyi2009}
Sergey Bravyi and Barbara Terhal.
\newblock ``A no-go theorem for a two-dimensional self-correcting quantum memory based on stabilizer codes''.
\newblock \href{https://dx.doi.org/10.1088/1367-2630/11/4/043029}{New J.~Phys. {\bf 11}, 043029}~(2009).

\bibitem{cardinal2013local}
Olivier Landon-Cardinal and David Poulin.
\newblock ``Local topological order inhibits thermal stability in 2d''.
\newblock \href{https://dx.doi.org/10.1103/PhysRevLett.110.090502}{Phys. Rev. Lett. {\bf 110}, 090502}~(2013).

\bibitem{haah2012logical}
Jeongwan Haah and John Preskill.
\newblock ``Logical-operator tradeoff for local quantum codes''.
\newblock \href{https://dx.doi.org/10.1103/PhysRevA.86.032308}{Phys. Rev. A {\bf 86}, 032308}~(2012).

\bibitem{hastings2011topological}
Matthew~B. Hastings.
\newblock ``Topological order at nonzero temperature''.
\newblock \href{https://dx.doi.org/10.1103/PhysRevLett.107.210501}{Phys. Rev. Lett. {\bf 107}, 210501}~(2011).

\bibitem{Yoshida2011}
Beni Yoshida.
\newblock ``{Feasibility of self-correcting quantum memory and thermal stability of topological order}''.
\newblock \href{https://dx.doi.org/10.1016/j.aop.2011.06.001}{Ann.~Phys. {\bf 326}, 2566--2633}~(2011).

\bibitem{haah2013commuting}
Jeongwan Haah.
\newblock ``Commuting pauli hamiltonians as maps between free modules''.
\newblock \href{https://dx.doi.org/10.1007/s00220-013-1810-2}{Comm. Math. Phys. {\bf 324}, 351--399}~(2013).

\bibitem{pastawski2015fault}
Fernando Pastawski and Beni Yoshida.
\newblock ``Fault-tolerant logical gates in quantum error-correcting codes''.
\newblock \href{https://dx.doi.org/10.1103/PhysRevA.91.012305}{Phys. Rev. A {\bf 91}, 012305}~(2015).

\bibitem{hallgren2013local}
Sean Hallgren, Daniel Nagaj, and Sandeep Narayanaswami.
\newblock ``The local hamiltonian problem on a line with eight states is qma-complete''.
\newblock \href{https://dx.doi.org/10.26421/QIC13.9-10-1}{Quantum Info. Comput. {\bf 13}, 721--750}~(2013).

\bibitem{Aharonov2009}
Dorit Aharonov, Daniel Gottesman, Sandy Irani, and Julia Kempe.
\newblock ``The power of quantum systems on a line''.
\newblock \href{https://dx.doi.org/10.1007/s00220-008-0710-3}{Comm. Math. Phys. {\bf 287}, 41--65}~(2009).

\bibitem{Schuch2007}
Norbert Schuch, Michael~M. Wolf, Frank Verstraete, and J.~Ignacio Cirac.
\newblock ``Computational complexity of projected entangled pair states''.
\newblock \href{https://dx.doi.org/10.1103/PhysRevLett.98.140506}{Phys. Rev. Lett. {\bf 98}, 140506}~(2007).

\bibitem{Piddock2015}
Stephen Piddock and Ashley Montanaro.
\newblock ``The complexity of antiferromagnetic interactions and 2d lattices''.
\newblock \href{https://dx.doi.org/10.26421/QIC17.7-8-6}{Quantum Inf. Comput. {\bf 17}, 636--672}~(2017).

\bibitem{Haah2011}
Jeongwan Haah.
\newblock ``{Local stabilizer codes in three dimensions without string logical operators}''.
\newblock \href{https://dx.doi.org/10.1103/PhysRevA.83.042330}{Phys. Rev. A {\bf 83}, 042330}~(2011).

\bibitem{hamma2009toric}
Alioscia Hamma, Claudio Castelnovo, and Claudio Chamon.
\newblock ``Toric-boson model: Toward a topological quantum memory at finite temperature''.
\newblock \href{https://dx.doi.org/10.1103/PhysRevB.79.245122}{Phys. Rev. B {\bf 79}, 245122}~(2009).

\bibitem{wootton2013topological}
James~R. Wootton.
\newblock ``Topological phases and self-correcting memories in interacting anyon systems''.
\newblock \href{https://dx.doi.org/10.1103/PhysRevA.88.062312}{Phys. Rev. A {\bf 88}, 062312}~(2013).

\bibitem{chesi2010self}
Stefano Chesi, Beat R\"othlisberger, and Daniel Loss.
\newblock ``Self-correcting quantum memory in a thermal environment''.
\newblock \href{https://dx.doi.org/10.1103/PhysRevA.82.022305}{Phys. Rev. A {\bf 82}, 022305}~(2010).

\bibitem{kapit2015passive}
Eliot Kapit, John~T. Chalker, and Steven~H. Simon.
\newblock ``Passive correction of quantum logical errors in a driven, dissipative system: A blueprint for an analog quantum code fabric''.
\newblock \href{https://dx.doi.org/10.1103/PhysRevA.91.062324}{Phys. Rev. A {\bf 91}, 062324}~(2015).

\bibitem{lieu2024candidate}
Simon Lieu, Yu-Jie Liu, and Alexey~V. Gorshkov.
\newblock ``Candidate for a passively protected quantum memory in two dimensions''.
\newblock \href{https://dx.doi.org/10.1103/PhysRevLett.133.030601}{Phys. Rev. Lett. {\bf 133}, 030601}~(2024).

\bibitem{lihm2018implementation}
Jae-Mo Lihm, Kyungjoo Noh, and Uwe~R. Fischer.
\newblock ``Implementation-independent sufficient condition of the knill-laflamme type for the autonomous protection of logical qudits by strong engineered dissipation''.
\newblock \href{https://dx.doi.org/10.1103/PhysRevA.98.012317}{Phys. Rev. A {\bf 98}, 012317}~(2018).

\bibitem{lebreuilly2021autonomous}
José Lebreuilly, Kyungjoo Noh, Chiao-Hsuan Wang, Steven~M. Girvin, and Liang Jiang.
\newblock ``Autonomous quantum error correction and quantum computation''~(2021).
\newblock  \href{http://arxiv.org/abs/2103.05007}{arXiv:2103.05007}.

\bibitem{Gottesman2013}
Daniel Gottesman.
\newblock ``Fault\-tolerant quantum computation with constant overhead''.
\newblock \href{https://dx.doi.org/10.26421/QIC14.15-16-5}{Quantum Inf. Comput. {\bf 14}, 1338--1371}~(2014).

\bibitem{lin2022good}
Ting\-Chun Lin and Min\-Hsiu Hsieh.
\newblock ``Good quantum {LDPC} codes with linear time decoder from lossless expanders''~(2022).
\newblock  \href{http://arxiv.org/abs/2203.03581}{arXiv:2203.03581}.

\bibitem{leverrier2022quantum}
Anthony Leverrier and Gilles Z{\'e}mor.
\newblock ``Quantum tanner codes''.
\newblock In 2022 IEEE 63rd Annual Symposium on Foundations of Computer Science (FOCS).
\newblock \href{https://dx.doi.org/10.1109/FOCS54457.2022.00117}{Pages 872--883}.
\newblock IEEE~(2022).

\bibitem{Panteleev2022asymptotically}
Pavel Panteleev and Gleb Kalachev.
\newblock ``Asymptotically good quantum and locally testable classical ldpc codes''.
\newblock In Proceedings of the 54th Annual ACM SIGACT Symposium on Theory of Computing (STOC ’22).
\newblock \href{https://dx.doi.org/10.1145/3519935.3520017}{Pages 375--388}.
\newblock ACM~(2022).

\bibitem{Alicki:2009}
Robert Alicki, Mark Fannes, and Michal Horodecki.
\newblock ``On thermalization in kitaev's 2d model''.
\newblock \href{https://dx.doi.org/10.1088/1751-8113/42/6/065303}{J.~Phys.~A: Math.~Theor. {\bf 42}, 065303}~(2009).

\bibitem{lucia2023thermalization}
Angelo Lucia, David Pérez-García, and Antonio Pérez-Hernández.
\newblock ``Thermalization in kitaev's quantum double models via tensor network techniques''.
\newblock \href{https://dx.doi.org/10.1017/fms.2023.98}{Forum of Mathematics, Sigma{\bf 11}}~(2023).

\bibitem{bardet:2023}
Ivan Bardet, \'Angela Capel, Li~Gao, Angelo Lucia, David P\'erez-Garc\'{\i}a, and Cambyse Rouz\'e.
\newblock ``Rapid thermalization of spin chain commuting hamiltonians''.
\newblock \href{https://dx.doi.org/10.1103/PhysRevLett.130.060401}{Phys. Rev. Lett. {\bf 130}, 060401}~(2023).

\bibitem{breuer2002theory}
Heinz-Peter Breuer and Francesco Petruccione.
\newblock ``The theory of open quantum systems''.
\newblock \href{https://dx.doi.org/10.1093/acprof:oso/9780199213900.001.0001}{Oxford University Press}. Oxford~(2002).

\bibitem{gorini1976completely}
Vittorio Gorini, Andrzej Kossakowski, and Ennackal Chandy~George Sudarshan.
\newblock ``Completely positive dynamical semigroups of n-level systems''.
\newblock \href{https://dx.doi.org/10.1063/1.522979}{J. Math. Phys. {\bf 17}, 821--825}~(1976).

\bibitem{Lindblad1976}
G.~Lindblad.
\newblock ``On the generators of quantum dynamical semigroups''.
\newblock \href{https://dx.doi.org/10.1007/BF01608499}{Comm. Math. Phys. {\bf 48}, 119--130}~(1976).

\bibitem{PhysRevA.55.900}
Emanuel Knill and Raymond Laflamme.
\newblock ``Theory of quantum error-correcting codes''.
\newblock \href{https://dx.doi.org/10.1103/PhysRevA.55.900}{Phys. Rev. A {\bf 55}, 900--911}~(1997).

\bibitem{PhysRevX.10.011058}
Arne~L. Grimsmo, Joshua Combes, and Ben~Q. Baragiola.
\newblock ``Quantum computing with rotation-symmetric bosonic codes''.
\newblock \href{https://dx.doi.org/10.1103/PhysRevX.10.011058}{Phys. Rev. X {\bf 10}, 011058}~(2020).

\bibitem{cubitt2015stability}
Toby~S. Cubitt, Angelo Lucia, Spyridon Michalakis, and David Perez-Garcia.
\newblock ``Stability of local quantum dissipative systems''.
\newblock \href{https://dx.doi.org/10.1007/s00220-015-2355-3}{Commun. Math. Phys. {\bf 337}, 1275--1315}~(2015).

\bibitem{brandao2015area}
Fernando~G.S.L. Brandao, Toby~S Cubitt, Angelo Lucia, Spyridon Michalakis, and David Perez-Garcia.
\newblock ``Area law for fixed points of rapidly mixing dissipative quantum systems''.
\newblock \href{https://dx.doi.org/10.1063/1.4932612}{J. Math. Phys. {\bf 56}, 102202}~(2015).

\bibitem{kubica2019cellular}
Aleksander Kubica and John Preskill.
\newblock ``Cellular-automaton decoders with provable thresholds for topological codes''.
\newblock \href{https://dx.doi.org/10.1103/PhysRevLett.123.020501}{Phys. Rev. Lett. {\bf 123}, 020501}~(2019).

\bibitem{NIST:DLMF}
``{NIST} digital library of mathematical functions''.
\newblock \url{https://dlmf.nist.gov/}~(2023).
\newblock 1.1.9, released 2023-03-15.

\bibitem{kossakowski1972quantum}
A.~Kossakowski.
\newblock ``On quantum statistical mechanics of non-hamiltonian systems''.
\newblock \href{https://dx.doi.org/10.1016/0034-4877(72)90010-9}{Reports on Mathematical Physics {\bf 3}, 247--274}~(1972).

\bibitem{gottesman1999}
D.~Gottesman.
\newblock ``Fault-tolerant quantum computation with higher-dimensional systems''.
\newblock \href{https://dx.doi.org/10.1016/S0960-0779(98)00218-5}{Chaos Solit. Fractals {\bf 10}, 1749--1758}~(1999).

\bibitem{bierbrauer2000quantum}
J{\"u}rgen Bierbrauer and Yves Edel.
\newblock ``Quantum twisted codes''.
\newblock \href{https://dx.doi.org/10.1002/(SICI)1520-6610(2000)8:3<174::AID-JCD3>3.0.CO;2-T}{J. Comb. Des. {\bf 8}, 174--188}~(2000).

\bibitem{ketkar2006nonbinary}
Avanti Ketkar, Andreas Klappenecker, Santosh Kumar, and Pradeep~Kiran Sarvepalli.
\newblock ``Nonbinary stabilizer codes over finite fields''.
\newblock \href{https://dx.doi.org/10.1109/TIT.2006.883612}{IEEE Trans. Inf. Theory {\bf 52}, 4892--4914}~(2006).

\bibitem{stace:2009}
Thomas~M. Stace, Sean~D. Barrett, and Andrew~C. Doherty.
\newblock ``Thresholds for topological codes in the presence of loss''.
\newblock \href{https://dx.doi.org/10.1103/PhysRevLett.102.200501}{Phys. Rev. Lett. {\bf 102}, 200501}~(2009).

\bibitem{tuckett:2019}
David~K. Tuckett, Andrew~S. Darmawan, Christopher~T. Chubb, Sergey Bravyi, Stephen~D. Bartlett, and Steven~T. Flammia.
\newblock ``Tailoring surface codes for highly biased noise''.
\newblock \href{https://dx.doi.org/10.1103/PhysRevX.9.041031}{Phys. Rev. X {\bf 9}, 041031}~(2019).

\bibitem{PhysRevLett.78.405}
A.~R. Calderbank, E.~M. Rains, P.~W. Shor, and N.~J.~A. Sloane.
\newblock ``Quantum error correction and orthogonal geometry''.
\newblock \href{https://dx.doi.org/10.1103/PhysRevLett.78.405}{Phys. Rev. Lett. {\bf 78}, 405--408}~(1997).

\bibitem{wang2003confinement}
Chenyang Wang, Jim Harrington, and John Preskill.
\newblock ``Confinement-higgs transition in a disordered gauge theory and the accuracy threshold for quantum memory''.
\newblock \href{https://dx.doi.org/10.1016/S0003-4916(02)00019-2}{Ann.~Phys. {\bf 303}, 31--58}~(2003).

\bibitem{cormen2001introduction}
Thomas~H. Cormen, Charles~E. Leiserson, Ronald~L. Rivest, and Clifford Stein.
\newblock ``Introduction to algorithms''.
\newblock \href{https://dx.doi.org/10.5555/580470}{MIT Press}. Cambridge, MA~(2001).
\newblock 2nd edition.

\bibitem{kitaev2006anyons}
Alexei Kitaev.
\newblock ``Anyons in an exactly solved model and beyond''.
\newblock \href{https://dx.doi.org/10.1016/j.aop.2005.10.005}{Ann. Phys. {\bf 321}, 2--111}~(2006).

\bibitem{dengis2014optimal}
John Dengis, Robert König, and Fernando Pastawski.
\newblock ``An optimal dissipative encoder for the toric code''.
\newblock \href{https://dx.doi.org/10.1088/1367-2630/16/1/013023}{New J. Phys. {\bf 16}, 013023}~(2014).

\bibitem{Ticozzi2019}
Francesco Ticozzi, Giacomo Baggio, and Lorenza Viola.
\newblock ``Quantum information encoding from stabilizing dynamics''.
\newblock In Proceedings of the 58th IEEE Conference on Decision and Control (CDC).
\newblock \href{https://dx.doi.org/10.1109/CDC40024.2019.9029402}{Pages 413--418}.
\newblock IEEE~(2019).

\bibitem{michael206new}
Marios~H. Michael, Matti Silveri, R.~T. Brierley, Victor~V. Albert, Juha Salmilehto, Liang Jiang, and S.~M. Girvin.
\newblock ``New class of quantum error-correcting codes for a bosonic mode''.
\newblock \href{https://dx.doi.org/10.1103/PhysRevX.6.031006}{Phys. Rev. X {\bf 6}, 031006}~(2016).

\bibitem{nielsen_quantum_2010}
Michael~A. Nielsen and Isaac~L. Chuang.
\newblock ``Quantum computation and quantum information''.
\newblock \href{https://dx.doi.org/10.1017/CBO9780511976667}{Cambridge University Press}. Cambridge; New York~(2010).
\newblock 10th anniversary edition edition.

\bibitem{prosen2012}
Berislav Bu{\v{c}}a and Toma{\v{z}} Prosen.
\newblock ``A note on symmetry reductions of the lindblad equation: transport in constrained open spin chains''.
\newblock \href{https://dx.doi.org/10.1088/1367-2630/14/7/073007}{New J.~Phys. {\bf 14}, 073007}~(2012).

\end{thebibliography}

\end{document}